\documentclass[11pt,leqno]{article}

\usepackage[margin=1.1in]{geometry}
\usepackage{amsmath,amssymb,amsthm}
\usepackage{graphicx}
\usepackage{float}
\usepackage{array}
\usepackage{xcolor}
\usepackage{enumitem}
\usepackage[authoryear,round]{natbib}
\usepackage{hyperref}
\hypersetup{
  colorlinks=true,
  citecolor=black,     
  linkcolor=black,      
  urlcolor=black,   
  filecolor=black,
}

\newtheorem{theorem}{Theorem}
\newtheorem{corollary}{Corollary}
\newtheorem{lemma}{Lemma}
\newtheorem{proposition}{Proposition}

\theoremstyle{definition}
\newtheorem{assumption}{Assumption}
\newtheorem{remark}{Remark}
\theoremstyle{plain}

\newenvironment{aparts}
  {\begin{enumerate}[label=\theassumption.\arabic*,ref=\theassumption.\arabic*,
                     leftmargin=2.7em,labelsep=0.6em,topsep=0.3em,
                     itemsep=0.6em,parsep=0.2em,partopsep=0pt]}
  {\end{enumerate}}

\makeatletter
\renewenvironment{proof}[1][]{\par\pushQED{\qed}\normalfont
  \topsep6\p@\@plus6\p@\relax
  \trivlist\item[\hskip\labelsep\itshape Proof #1\@addpunct{.}]\ignorespaces
}{\popQED\endtrivlist\@endpefalse}
\makeatother

\def\ind{ {{\rm 1}\hskip-2.2pt{\rm l}}}
\def\P{\mathbb P}
\def\E{\mathbb E}
\begin{document}

\title{Statistics of multivariate extremes under random censoring}
\author{Martin Bladt\\[2pt]
\normalsize Department of Mathematical Sciences, University of Copenhagen}
\date{}
\maketitle

\begin{abstract}
\noindent
We study tail dependence of a $d$-dimensional random vector whose coordinates are
subject to random right censoring.  Along each fixed direction the censored problem
reduces exactly to a univariate one, and the observed data determine the radius and
whether it was produced by the event or censoring vector.  An
ordinary Kaplan--Meier product limit therefore estimates the joint tail probability in
that direction, in every dimension, and with no multivariate survival function, no
smoothing and no tuning parameter beyond the threshold.  The theory of this directional estimator is formulated
under an arbitrary marginal standardization and conditions only imposed on the standardized
laws, in particular for any max-domain of attraction.  We prove uniform consistency and functional weak convergence at the square root
of the effective number of joint extremes, allowing the standardization to be
estimated.  A multiplicative standardization recovers the heavy-tailed theory, whereas
a standardization built from the marginal (non-directional) Kaplan--Meier estimators requires no marginal
tail model and targets the normalized tail copula itself.  The standardization error is negligible for the former, and for the latter under a
mild condition on the joint censoring. Simulation studies validate the finite-sample performance of the estimator.  An application to the National Flood Insurance Program dataset comprised of
claims generated by Hurricane Ian estimates the joint upper tail of building and
contents losses from indemnities, which are subject to capping, simultaneously in both coordinates.
\end{abstract}



\addtocontents{toc}{\protect\setcounter{tocdepth}{-5}}

\section{Introduction}\label{sec:intro}

Risk is usually multivariate.  A single storm can damage a building and its contents at the same
time; a flood defense may fail when several gauges are high together; a portfolio is
threatened when several positions lose together; or a patient may die when certain biological indices are jointly elevated.  In each case the quantity that determines
the decision is the probability that every coordinate of a random vector is
simultaneously extreme, and that probability is not solely determined by the marginal tails.  Estimating that probability is challenging, as it is small, so the observations reaching the
region are scarce, and every quantity computed from it propagates this instability.  Multivariate extreme-value theory is
the established framework for this task.

Let $X=(X^{(1)},\ldots,X^{(d)})$, $d\geq2$, be a multivariate random vector.  The problem treated
here is that its coordinates are observed under componentwise random right censoring:
one sees $Z^{(j)}=X^{(j)}\wedge C^{(j)}$ and the associated censoring indicators, with
the vector $C$ independent of $X$ but with arbitrary dependence among the coordinates of
either vector.  Such settings arise in many applications, for instance as capped values, limitations in measurement tools, or as partial observations of lifetimes as in classical medical studies in survival analysis. Censoring changes the mathematical problem and its effective sample size.  The extreme
observations of $Z$ depend on the tails of both $X$ and $C$, and the proportion of
censored observations need not vanish in the tail, and may converge to a constant determined by the relative tail decay.  An uncensored tail
estimator applied to $Z$ therefore targets the wrong joint tail, whereas univariate
marginal corrections for censoring do not suffice by themselves to reconstruct a joint exceedance probability.

Without censoring the problem is classical.  Multivariate extreme-value theory separates
marginal tail behavior from extremal dependence, in the sense that after marginal
standardization, the tail behavior is described by an exponent
measure, or equivalently by a spectral measure or a stable tail dependence function.  Its nonparametric estimation
is based on the empirical tail measure and its weak limit; see for instance
\citet{EinmahlDeHaanSinha1997}, \citet{DreesHuang1998},
\citet{EinmahlDeHaanPiterbarg2001}, \citet{SchmidtStadtmuller2006},
\citet{EinmahlKrajinaSegers2012}, \citet{BucherVolgushev2013} and
\citet{deHaanFerreira2006}.

Under censoring, the univariate theory of extremes was developed by
\citet{BeirlantEtAl2007}, \citet{EinmahlFilsVilletardGuillou2008} and
\citet{WormsWorms2014}, and subsequently investigated via trimming, covariates, and extreme
Kaplan--Meier integrals; see, respectively,
\citet{BladtAlbrecherBeirlant2021}, \citet{Stupfler2016}, and \citet{BladtRodionov2025}.
Recent related product-limit
methods for bias correction and tail risk functionals are given by
\citet{BladtGoegebeurGuillou2026a}.  On the other hand, in several dimensions, for non-extreme regions, the classical survival
literature includes the planar product-limit theory of
\citet{Dabrowska1988,Dabrowska1989}, the moment and estimating-equation constructions
of \citet{PrenticeCai1992}, the efficiency analysis of
\citet{GillVanDerLaanWellner1995} and \citet{vanderLaan1996}, and the polar
constructions of \citet{DaiFu2012,DaiRestainoWang2016}.  Work directed specifically at
censored extremal dependence includes the Hall-type analysis of joint-tail probabilities
by \citet{HashorvaLingPeng2014} and \citet{GoegebeurGuillouQin2019}. The recent contributions \citet{BladtGoegebeurGuillou2026b} and
\citet{BladtGuillou2026} tackle the bivariate and regularly varying case by the arduous construction of \citet{Dabrowska1988} with considerable technical effort.  In contrast, the approach taken here allows
arbitrary dimension and extends to any marginal tail model, with the same rate of convergence.

More precisely, the present approach reduces each joint-tail probability to an ordinary
product limit before passing to the extreme-value limit.  Fix a direction $q$ in the
positive orthant and consider the first time, along the ray through $q$, at which the
standardized event vector fails to have all coordinates above the moving point.  The analogous construction for the censoring
vector produces a time which can be shown to be independent of the event time.  The observed data then determine the minimum of these two times
and whether the minimum was caused by the event vector, exactly and without
approximation.  Consequently, an ordinary
one-dimensional Kaplan--Meier product limit estimates the joint tail probability in the
direction $q$.  Letting $q$ vary weaves all these estimators together, and recovers the tail surface; normalizing by the
product limit at $\boldsymbol 1_d$ yields an estimator of the normalized joint-tail
limit $F^\circ$; on upper orthants this is the object corresponding to the exponent
measure of the uncensored theory.  In the absence of censoring the construction reduces
to the empirical joint-exceedance estimator.

This directional formulation avoids estimating a full multivariate survival
distribution but still recovers the joint extremal dependence structure. The
mathematical difficulties it creates are nonetheless substantial. We highlight, for instance, that the resulting Kaplan--Meier estimators have a sample which varies with $q$, and a risk set which is not bounded from below. The marginal standardizations are a second ingredient requiring careful analysis, and we propose two choices.  A multiplicative one gives
the usual heavy-tailed formulation, in which $F^\circ$ keeps the marginal tail scales.
In contrast, the one obtained from applying the marginal (non-directional) Kaplan--Meier estimators to each coordinate, imposes no marginal tail model
and makes $F^\circ$ the normalized tail copula itself.

The paper has three mathematical contributions.  First, the directional identity converts the
multivariate censored problem into an indexed family of ordinary univariate product limits, which, when woven together, are able to recover tail dependence structures in upper orthants.
Second, a general theory gives uniform consistency and functional weak convergence when
the marginal standardization is known or estimated.  Stability theory is provided, which in particular separates changes arising from
estimation error in the radius from changes in the censoring indicator, which is the delicate part of
plug-in standardization.  Third, the two concrete standardizations give explicit
procedures and their estimation error is negligible at the joint effective
sample-size scale. This always holds for the multiplicative standardization, and for the
Kaplan--Meier one when the joint censoring is of smaller order than its margins; that
condition is met, for instance, when the censoring coordinates are independent and at
least two of them are censored, but also in other settings.  The limiting rate
is $r_n$, the square root of the expected number of observations reaching the extreme
level in all coordinates, rather than the square root of the sample fraction. In the absence of censoring, the rates agree.

Some nonparametric survival theory had to be strengthened to achieve our limit theorems. Related but unfortunately inapplicable results include the weak
convergence and whole-line theory for an ordinary Kaplan--Meier estimator due to
\citet{BreslowCrowley1974} and \citet{Gill1983}; and uniform and increasing-set rates of convergence developed by \citet{CsorgoHorvath1983}, \citet{Stute1994a,Stute1994b} and
\citet{Csorgo1996}. All of those results assume a censored sample with fixed observations.  Here, the estimated marginal standardization changes both the reduced times and their censoring indicators, uniformly in direction and at the scale
of a triangular rare-event array.  We therefore establish, in the Supplementary Material,
the stability of the directional product limit under such distortions.  With those tools in place, we establish the asymptotic theory of the univariate Kaplan--Meier standardizations,
uniformly on a fixed standardized interval, in exactly the form required by some generic
plug-in conditions. We use empirical-process
arguments throughout, in the framework of \citet{vanderVaartWellner1996}.

The remainder of the paper is organized as follows. The directional construction, the
estimator, and the assumptions are given in
Section~\ref{sec:setup}.  Consistency and functional weak convergence are established
in Section~\ref{sec:asymptotics}, and the simulation experiments in
Section~\ref{sec:simulations}.  An application to insurance losses censored by policy
limits occupies Section~\ref{sec:ian}, and Section~\ref{sec:discussion} concludes and discusses some open problems.  
The Supplementary Material contains the two concrete standardizations, further numerical studies, the plug-in stability theory, the additional marginal Kaplan--Meier theory, and proofs of all the results.

\section{Setup, estimator, and conditions}\label{sec:setup}

We first define setup and the notation, then build the estimator from
the directional reduction, and finally state the mathematical assumptions required throughout.

Let $X=(X^{(1)},\ldots,X^{(d)})$, $d\geq 2$ be a positive random vector and let
$C=(C^{(1)},\ldots,C^{(d)})$ be a positive censoring vector, independent of $X$.
From the independent copies $(X_i,C_i)$, $i=1,\ldots,n$, one observes
\begin{align}
 Z_i^{(j)}=X_i^{(j)}\wedge C_i^{(j)},\qquad
 \delta_i^{(j)}=\ind_{\{X_i^{(j)}\leq C_i^{(j)}\}},\qquad j=1,\ldots,d.
 \label{obs}
\end{align}
Dependence between the coordinates of either vector is unrestricted.
For a positive random variable $V$, write its tail as $\overline F_V=1-F_V$ and
the associated tail quantile function by $U_V(r)=\inf\{x:F_V(x)\geq1-r^{-1}\}$ for $r>1$,
and write $\boldsymbol 1=\boldsymbol 1_d$ for the vector of ones.

\subsection{Construction of the estimator}\label{sec:reduction}\label{sec:plugin}

We consider $d$
deterministic, strictly increasing, and continuous maps
\begin{align*}
 \psi_{1,n},\ldots,\psi_{d,n}:(0,\infty)\longrightarrow(0,\infty),
\end{align*}
which standardize the margins at the extreme level of interest, and we work
throughout with the standardized observations $\psi_{j,n}(Z^{(j)}_i)$.  The reader may
keep in mind the two choices, studied as special cases of the general theory, of the Supplementary Material,
\begin{align}
 \psi_{j,n}(x)={x}/{u_{j,n}}
 \qquad\text{and}\qquad
 \psi_{j,n}(x)=\frac{k/n}{\overline F_{X^{(j)}}(x)} ,
 \label{twochoices}
\end{align}
where $k=k_n$ is an intermediate sequence, $k\to\infty$ and $k/n\to0$, and
$u_{j,n}=U_{X^{(j)}}(n/k)$ is the corresponding intermediate quantile.  The first is the
multiplicative standardization used when $X^{(j)}$ has a heavy tail, and makes
$\psi_{j,n}(X^{(j)})$ approximately standard Pareto at the level $k/n$; the second is
the marginal rank transformation, which makes it exactly standard Pareto at that level
whatever the margin, and which is estimated by the marginal Kaplan--Meier estimator.
Nothing below depends on either
choice; only the two displayed joint probabilities
\begin{equation}
\begin{aligned}
p_n(q)
  &= \P\bigl(\psi_{j,n}(X^{(j)})>q_j,\ j=1,\ldots,d\bigr),\\
c_n(q)
  &= \P\bigl(\psi_{j,n}(C^{(j)})>q_j,\ j=1,\ldots,d\bigr),
\qquad q=(q_1,\ldots,q_d)\in(0,\infty)^d.
\end{aligned}
\label{pncn}
\end{equation}
are affected by these functions, and by the independence of $X$ and $C$, we have
\begin{align}
 \P\bigl(\psi_{j,n}(Z^{(j)})>q_j,\ j=1,\ldots,d\bigr)=p_n(q)\,c_n(q).
 \label{ZisXC}
\end{align}

Define
\begin{align}
 T_{i,n}(q)&=\min_{1\leq j\leq d}\frac{\psi_{j,n}(X_i^{(j)})}{q_j}, \quad D_{i,n}(q)=\min_{1\leq j\leq d}\frac{\psi_{j,n}(C_i^{(j)})}{q_j}.
 \label{TD}
\end{align}
The two variables in \eqref{TD} are independent, but neither is observable.  What is
observable, once $\psi_{1,n},\ldots,\psi_{d,n}$ are known, is
\begin{align*}
 A_{i,n}^{(j)}(q)=\frac{\psi_{j,n}(Z_i^{(j)})}{q_j},\qquad j=1,\ldots,d,
 \qquad
 W_{i,n}(q)=\min_{1\leq j\leq d}A_{i,n}^{(j)}(q),
\end{align*}
together with the censoring indicator, or mark,
\begin{align}
 \Delta_{i,n}(q)=\max\bigl\{\delta_i^{(j)}:\ A_{i,n}^{(j)}(q)=W_{i,n}(q)\bigr\},
 \label{Delta}
\end{align}
the maximum being over the coordinates that attain the minimum.  For $d=2$ this is
$\delta^{(1)}_i$ or $\delta^{(2)}_i$ according to which of the two standardized
coordinates is smaller, and $\delta^{(1)}_i\vee\delta^{(2)}_i$ in case of a tie.
The point of these definitions is that the $d$-dimensional censored problem collapses,
in each fixed direction, to a one-dimensional one. Indeed, the following identities hold
exactly, for every $i$, $n$, and $q$,
\begin{align}
 W_{i,n}(q)=T_{i,n}(q)\wedge D_{i,n}(q),\qquad
 \Delta_{i,n}(q)=\ind_{\{T_{i,n}(q)\leq D_{i,n}(q)\}}.
 \label{reduction}
\end{align}
Thus $(W_{i,n}(q),\Delta_{i,n}(q))$ is a randomly right-censored observation of
$T_{i,n}(q)$, with independent censoring variable $D_{i,n}(q)$, and the Kaplan--Meier
product limit of these pairs estimates $\P\bigl(T_{1,n}(q)>1\bigr)=p_n(q)$.  The
tie-breaking convention in \eqref{Delta} is retained throughout; no assumption
concerning atoms of the observed angle is needed for consistency.

Since the above construction does not distinguish two marginal models that differ by increasing marginal
transformations, only the standardized laws \eqref{pncn} play an asymptotic role.
Accordingly, the conditions below are stated in terms of $p_n$ and $c_n$ alone.

Fix throughout a constant $T>1$; the estimator is studied at directions lying in
$K=[1,T]^d$.  In practice $\psi_{1,n},\ldots,\psi_{d,n}$ are unknown and require replacement by suitable estimators
$\widehat\psi_{1,n},\ldots,\widehat\psi_{d,n}$.  These are required to be nondecreasing, and to
be injective on $\{Z^{(j)}_i:\delta^{(j)}_i=1\}$, the uncensored observations of the
$j$th coordinate.  Replacing $\psi_{j,n}$ by $\widehat\psi_{j,n}$ in the definitions of $A_{i,n}^{(j)}(q)$,
$W_{i,n}(q)$ and \eqref{Delta} gives a sample array
$(\widehat W_{i,n}(q),\widehat\Delta_{i,n}(q))$, $i=1,\ldots,n$.

For $v>0$, $q\in(0,\infty)^d$, and Borel sets $B\subset(0,\infty)$, let
\begin{align*}
 \widehat Y_n(v;q)=\sum_{i=1}^n\ind_{\{\widehat W_{i,n}(q)\geq v\}},\qquad
 \widehat N_n(B;q)=\sum_{i=1}^n
 \widehat\Delta_{i,n}(q)\ind_{\{\widehat W_{i,n}(q)\in B\}}.
\end{align*}
Define
\begin{align}
 \widehat p_n(q)=
 \prod_{0<v\leq1}\bigl\{1-{\widehat N_n(\{v\};q)\over
                                      \widehat Y_n(v;q)}\bigr\}.
 \label{KM}
\end{align}
where the product is over the atoms of $\widehat N_n(\cdot\,;q)$.  This is the
Kaplan--Meier product limit \citep{KaplanMeier1958}
for the transformed observations.  Finally,
\begin{align}
 \widehat F_n^\circ(q)={\widehat p_n(q)\over\widehat p_n(\boldsymbol 1)},
 \label{estimator}
\end{align}
the ratio being set equal to zero if its denominator vanishes.  We call the sample
built from the true $(W_{i,n}(q),\Delta_{i,n}(q))$ the \emph{oracle} sample, and write
$\widetilde p_n(q)$ for the product limit \eqref{KM} formed by it.

\subsection{Regularity conditions}\label{sec:conditions}

The conditions are of two kinds, and we state them in order.
Assumption~\ref{ass:structure} constrains the standardized laws $p_n$ and $c_n$ alone, and it
concerns the oracle problem, in which $\psi_{1,n},\ldots,\psi_{d,n}$ are known.
Assumption~\ref{ass:plugin} then controls the fact of estimating them,
which turns out to be an unproblematic one.

\begin{assumption}[Standardized tails and risk set]\label{ass:structure}\leavevmode
\begin{aparts}
\item\label{a:margins} The $2d$ marginal distribution functions of $X$ and $C$ are
continuous.

\item\label{a:FG} The standardized laws stabilize: there are functions $F^\circ$ and
$G^\circ$ with, locally uniformly on $(0,\infty)^d$,
\begin{align*}
 {p_n(q)\over p_n(\boldsymbol 1)}\longrightarrow F^\circ(q),\qquad
 {c_n(q)\over c_n(\boldsymbol 1)}\longrightarrow G^\circ(q),
\end{align*}
where $F^\circ$ and $G^\circ$ are continuous and strictly positive on $(0,\infty)^d$
and, for each $j$, converge to zero as $q_j\to\infty$ with the remaining
coordinates held fixed.  Being limits of $p_n(\cdot)/p_n(\boldsymbol 1)$ and
$c_n(\cdot)/c_n(\boldsymbol 1)$, they are automatically nonincreasing in each argument and
equal to one at $\boldsymbol 1$.

\item\label{a:rate} The effective sample size diverges:
\begin{align*}
 r_n:=\bigl\{n\,p_n(\boldsymbol 1)\,c_n(\boldsymbol 1)\bigr\}^{1/2}\longrightarrow\infty.
\end{align*}

\item\label{a:radialmodulus} The risk set does not grow too fast along a ray: there is
$\rho>0$ such that, for all $n$ large,
\begin{align*}
 {p_n(vq)\,c_n(vq)\over p_n(Avq)\,c_n(Avq)}\ \leq\ A^{\rho},
 \qquad q\in[1/4,16T]^d,\ \ A\geq1,\ \ 0<v\leq Av\leq16 .
\end{align*}
\end{aparts}
\end{assumption}

\begin{remark}\label{rem:asm1}
The function $F^\circ$ is the target of the estimator, and
$G^\circ$ describes the censoring on the same scale.
By \eqref{ZisXC}, $r_n^2$ is the expected number of observations exceeding the
standardized level $\boldsymbol 1$ in every coordinate; it is this sequence, and neither $n$ nor $k$, that is the rate at which
the estimator is shown to converge below.
By \eqref{ZisXC}, Assumption~\ref{a:radialmodulus} states that the survival function of the reduced observation
$W_{1,n}(q)$ decreases at least as slowly as the fixed power $v^{-\rho}$ along every ray  $q\in[1/4,16T]^d$; it is
the Potter bound of that univariate survival function, and it is equivalent to the monotonicity of
$v\mapsto v^{\rho}p_n(vq)c_n(vq)$ on the same range.  The constants $1/4$, $16$ and
$16T$ merely enlarge the direction set $K$ by a fixed factor, which the plug-in argument requires.
\end{remark}

Probabilities of events involving a single
observation are normalized by $r_n^2$ throughout, and for that reason we introduce the shorthand measure notation
\begin{align}
 \nu_n(A)={\P(A)\over\P\bigl(\psi_{j,n}(Z^{(j)})>1,\ j=1,\ldots,d\bigr)}
 ={n\over r_n^2}\,\P(A).
 \label{nu}
\end{align}
Two survival objects appear repeatedly throughout the proofs.  For $v>0$ and
$q\in(0,\infty)^d$, put
\begin{align}
 H_n(v;q)=\nu_n\{W_{1,n}(q)\geq v\},\qquad
 H_{1,n}(dv;q)=\nu_n\{W_{1,n}(q)\in dv,\ \Delta_{1,n}(q)=1\},
 \label{HH1}
\end{align}
which are the normalized risk function and the normalized failure measure of the reduced
observations.  Both are fully determined by $p_n$ and $c_n$.  In this notation
Assumption~\ref{a:radialmodulus} is $H_n(v;q)\leq A^{\rho}H_n(Av;q)$.

The estimated sample is the oracle sample with random coordinate
distortions $\widehat\psi_{j,n}\circ\psi_{j,n}^{-1}$, and its stochastic effect on the product limit is
delicate.  Assumption~\ref{ass:plugin} states what is required of the estimated
standardization.

\begin{assumption}[Estimated standardization]\label{ass:plugin}\leavevmode
\begin{aparts}
\item\label{a:plugin} Accuracy is required only below the fixed
level $16T$, and
is measured by a positive deterministic sequence $\eta_n\downarrow0$ with
\begin{align*}
 \P\Bigl(\sup_{x:\ \psi_{j,n}(x)\leq16T}
 \bigl|{\widehat\psi_{j,n}(x)/\psi_{j,n}(x)}-1\bigr|
 \leq\eta_n\ \text{ for }j=1,\ldots,d\Bigr)\longrightarrow1 .
\end{align*}

\item\label{a:plugincons} The speed of $\eta_n$ satisfies
\begin{align*}
 \eta_n\,\log{n\over r_n^2}\longrightarrow0 .
\end{align*}

\item\label{a:pluginnorm} The same quantity, at the rate $r_n$, satisfies
\begin{align*}
 r_n\,\eta_n\,\log{n\over r_n^2}\longrightarrow0 .
\end{align*}
\end{aparts}
\end{assumption}

\begin{remark}\label{rem:asm2}
A distortion as in Assumption~\ref{a:plugin} moves each failure time by a factor in $[1-\eta_n,1+\eta_n]$, and it may also change which coordinate attains the
minimum in $W_{i,n}(q)$.  By \eqref{reduction} the latter alters the indicator only when the
reduced failure and censoring times are within that factor of each other.
The quantity $n/r_n^2=1/\{p_n(\boldsymbol 1)c_n(\boldsymbol 1)\}$
is the reciprocal of the probability that a single observation reaches the standardized
level $\boldsymbol 1$ in every coordinate, so Assumption~\ref{a:plugincons} requires that the marginal standardizations should be estimated
to a relative accuracy of smaller order than $1/\log(n/r_n^2)$.
The stronger Assumption~\ref{a:pluginnorm} is often easily satisfied, as marginal rates are often faster than joint ones, which makes the estimator asymptotically equivalent to its oracle counterpart;
since $r_n\to\infty$, it implies Assumption~\ref{a:plugincons}.
\end{remark}

\section{Main results}\label{sec:asymptotics}
This section establishes
uniform consistency of the directional estimator on $K$ and functional weak convergence at the
rate $r_n$; in both cases, the proofs treat the oracle problem first and then the
distortion of the coordinates is shown to be asymptotically negligible.
The section closes with a proposed threshold selection rule.

\subsection{Large sample behaviour}\label{sec:largesample}

Recall the fixed constant $T>1$ and the direction set $K=[1,T]^d$.

\begin{theorem}\label{thm:consistency}
Suppose that Assumptions~\ref{ass:structure} and~\ref{ass:plugin}.1--\ref{ass:plugin}.2
hold.  Then
\begin{align*}
 \sup_{q\in K}|\widehat F_n^\circ(q)-F^\circ(q)|
 \overset{\P}{\longrightarrow}0.
\end{align*}
\end{theorem}

The distributional limit requires the following additional conditions.

\begin{assumption}[Refinements for weak convergence]\label{ass:wc}\leavevmode
\begin{aparts}
\item\label{a:divergence} The risk set diverges as one moves toward the origin along
a ray; equivalently, the
limit of the normalized risk function $H_n$ blows up at the origin:
\begin{align*}
 \inf_{q\in K'}F^\circ(vq)\,G^\circ(vq)\longrightarrow\infty
 \qquad (v\downarrow0)
\end{align*}
for every compact $K'\subset(0,\infty)^d$.

\item\label{a:bias} The estimator is centered at
$p_n(\cdot)/p_n(\boldsymbol 1)$, and we require the deterministic approximation of
$F^\circ$ by this ratio to be small at the rate $r_n$. Namely, for every compact
$\widetilde K\subset(0,\infty)^d$,
\begin{align*}
 \sup_{q\in\widetilde K}r_n
 \bigl|{p_n(q)\over p_n(\boldsymbol 1)}-F^\circ(q)\bigr|\longrightarrow0 .
\end{align*}

\end{aparts}
\end{assumption}

We now construct the Gaussian process appearing in the weak limit. Since it depends on the direction $q$, it requires some care, and we proceed in three steps.

First we identify a key limit measure. Notice that the measures $\nu_n$ are defined on $(0,\infty)^d\times\{0,1\}^d$.  At a point
$(z,\epsilon)$ put $W(q)=\min_{1\leq j\leq d}z_j/q_j$ and
$\Delta(q)=\max\{\epsilon_j:z_j/q_j=W(q)\}$; at $z_j=\psi_{j,n}(Z_1^{(j)})$ and
$\epsilon_j=\delta_1^{(j)}$ these are simply $W_{1,n}(q)$ and $\Delta_{1,n}(q)$.
For every $\lambda>0$, $\nu_n$ converges weakly on
$[\lambda,\infty)^d\times\{0,1\}^d$ to the image $\nu$ of the product of the two
measures on $(0,\infty)^d$ with survival functions $F^\circ$ and $G^\circ$ under the
coordinatewise map $(x,c)\mapsto(x\wedge c,\ \ind_{\{x\leq c\}})$.  

Next, we identify some $q$-dependent functions which drive the covariance structure. Writing
$H(v;q)=F^\circ(vq)G^\circ(vq)$ and $H_1(dv;q)=G^\circ(vq)\{-dF^\circ(vq)\}$ for the limits
of \eqref{HH1}, let
\begin{align}
 \xi_q=\int_{(0,1]}{\ind_{\{W(q)\geq v\}}\over H(v;q)^2}\,H_1(dv;q)
 -{\Delta(q)\,\ind_{\{W(q)\leq1\}}\over H(W(q);q)}.
 \label{scorelimit}
\end{align}
Notice that $\xi_q$ belongs to $L^2(\nu)$, since
\begin{align}
 \nu\bigl(\xi_q^2\bigr)=\int_{(0,1]}{H_1(dv;q)\over H(v;q)^2}
 \leq\int_{(0,1]}{-dH(v;q)\over H(v;q)^2}\leq{1\over H(1;q)},
 \label{limitlower}
\end{align}
and we recognise a Gill--Greenwood-type variance formula for each $q$.  

Finally, we can define a centered Gaussian process
$\mathbb G$ on $K$ with covariance
\begin{align}
 \E[\mathbb G(q)\mathbb G(q')]
 =\nu\bigl\{(\xi_q-\xi_{\boldsymbol 1})(\xi_{q'}-\xi_{\boldsymbol 1})\bigr\}.
 \label{Glimit}
\end{align}
We can now state functional weak convergence of the directional estimators.

\begin{theorem}\label{thm:normality}
Suppose that Assumptions~\ref{ass:structure}, \ref{ass:plugin}
and~\ref{ass:wc} hold.
Then
\begin{align*}
 r_n\{\widehat F_n^\circ-F^\circ\}
 \ \leadsto\ F^\circ\mathbb G
 \qquad\text{in }\ell^\infty(K),
\end{align*}
where $\mathbb G$ is the centered Gaussian process of \eqref{Glimit}.
\end{theorem}

\begin{remark}\label{rem:regular}
Call the standardization \emph{regular} if, for every compact
$\widetilde K\subset[1/4,4T]^d$,
$r_n\{\log\widehat p_n-\log\widetilde p_n\}\leadsto\mathbb M$ in
$\ell^\infty(\widetilde K)$, jointly with $r_n\{\log\widetilde p_n-\log p_n\}$.
Under regularity the proof of
Theorem~\ref{thm:normality} applies unchanged except that now
$r_n\{\widehat F_n^\circ-F^\circ\}\leadsto
F^\circ\{\mathbb G+\mathbb M-\mathbb M(\boldsymbol 1)\}$ in $\ell^\infty(K)$.

Assumption~\ref{a:pluginnorm} gives regularity with $\mathbb M\equiv0$, by
Proposition~\ref{prop:stability}(ii), which also depends on the
standardization.  For instance, for the multiplicative
standardization only the quantile $u_{j,n}$ is estimated, and $\mathbb M\equiv0$ whatever the dependence among the
coordinates of $C$.  For the Kaplan--Meier standardization the marginal survival function
is estimated nonparametrically at the threshold, and $\mathbb M\equiv0$ requires
$c_n(\boldsymbol 1)$ to be of smaller order than $\min_j\P\bigl(C^{(j)}>u_{j,n}\bigr)$; this holds for instance
for independent censoring coordinates with at least two of them censored in the tail.
Otherwise, identifying $\mathbb M$ is
required, albeit not straightforward; see Section~\ref{sec:discussion}.
\end{remark}

\begin{remark}\label{rem:uncensored}
Without censoring, that is when $\Delta_{i,n}\equiv1$, the limit reduces to the
classical one.  Indeed $H_1(dv;q)=-dH(v;q)$, so by
Assumption~\ref{a:divergence} we get 
$\xi_q=\ind_{\{W(q)>1\}}/F^\circ(q)$.  Since $\nu\{W(q)>1\}=F^\circ(q)$ and
$\{W(q)>1\}\cap\{W(q')>1\}=\{W(q\vee q')>1\}$, \eqref{Glimit} becomes
$\E[\mathbb G(q)\mathbb G(q')]=F^\circ(q\vee q')/\{F^\circ(q)F^\circ(q')\}-1$, so that
the limit $F^\circ\mathbb G$ has covariance
$F^\circ(q\vee q')-F^\circ(q)F^\circ(q')$, which agrees with the uncensored case treated in \citet{DreesHuang1998}. 
\end{remark}

\begin{proposition}\label{prop:rateest}
Suppose that Assumptions~\ref{a:margins}--\ref{a:rate} and~\ref{a:plugin} hold.
Then the rate is consistently estimated by
\begin{align*}
 \widehat r_n=\Bigl[\sum_{i=1}^n
 \ind_{\{\widehat\psi_{j,n}(Z_i^{(j)})\geq1,\ j=1,\ldots,d\}}\Bigr]^{1/2}.
\end{align*}
More precisely, $\widehat r_n/r_n\overset{\P}\longrightarrow1$.
\end{proposition}

The covariance representation above also allows for a direct plug-in estimator of the
pointwise variance, obtained by replacing every component of $\xi_q$ by its
empirical counterpart.  On the event that
$\widehat Y_n(v;q)>\widehat N_n(\{v\};q)$ at every atom $v\in(0,1]$ of
$\widehat N_n(\cdot\,;q)$, for $i\leq n$ define
\begin{align}
 \widehat\xi_{i,n}(q)
 ={}&n\int_{(0,1]}
 {\ind_{\{\widehat W_{i,n}(q)\geq v\}}\over
  \widehat Y_n(v;q)
  \{\widehat Y_n(v;q)-\widehat N_n(\{v\};q)\}}
 \,\widehat N_n(dv;q)-{n\widehat\Delta_{i,n}(q)
       \ind_{\{\widehat W_{i,n}(q)\leq1\}}
 \over \widehat Y_n(\widehat W_{i,n}(q);q)
       -\widehat N_n(\{\widehat W_{i,n}(q)\};q)}.
 \label{greenwoodscore}
\end{align}
With
$q_0=\boldsymbol 1$, set
\begin{align}
 \widehat\sigma_n^2(q)
 ={\widehat r_n^2\over n(n-1)}\sum_{i=1}^n
 \bigl\{\widehat\xi_{i,n}(q)-\widehat\xi_{i,n}(q_0)\bigr\}^2.
 \label{greenwoodvariance}
\end{align}
and set $\widehat\sigma_n^2(q)=0$ if a denominator is zero.  The latter event
has probability tending to zero under the conditions below.

\begin{proposition}\label{prop:variance}
Suppose that the hypotheses of Theorem~\ref{thm:normality} hold.  Then, for every
fixed $q\in K$,
\begin{align*}
 \widehat\sigma_n^2(q)
 \overset{\P}\longrightarrow \E\bigl[\mathbb G(q)^2\bigr].
\end{align*}
In particular, if $\E[\mathbb G(q)^2]>0$, then
\begin{align*}
 {\widehat r_n\over\widehat\sigma_n(q)}
 \bigl\{\log\widehat F_n^\circ(q)-\log F^\circ(q)\bigr\}
 \ \leadsto\ N(0,1).
\end{align*}
Consequently, an asymptotic pointwise confidence interval of level $1-\alpha$ is
\begin{align}
 \widehat F_n^\circ(q)
 \exp\bigl\{\mathord\pm z_{1-\alpha/2}
                 \widehat\sigma_n(q)/\widehat r_n\bigr\}.
 \label{greenwoodCI}
\end{align}
\end{proposition}

\begin{remark}
At $q=q_0$, $\widehat F_n^\circ(q_0)=F^\circ(q_0)=1$ identically and
$\mathbb G(q_0)=0$, so no confidence region is either needed or intended there.  For $q\ne q_0$, a
sufficient nondegeneracy condition is $F^\circ(q)<1$. Indeed, the normalized probability
that $X$ exceeds $q_0$ but not $q$, while $C$ exceeds $q$, converges to
$\{1-F^\circ(q)\}G^\circ(q)>0$. In both concrete standardizations, this condition holds
whenever every coordinate of $q$ exceeds one.
\end{remark}

As concrete illustrations of the above asymptotic results, two standardizations are provided in
the Supplementary Material, with their corresponding asymptotic theory as direct corollaries. The
multiplicative standardization divides each coordinate by an intermediate
marginal quantile, and requires regularly varying margins.
The Kaplan--Meier standardization divides instead by the marginal
Kaplan--Meier survival function, and assumes no marginal model. Asymptotically, the multiplicative standardization is negligible at the scale
$r_n$, and so is the Kaplan--Meier one when the joint censoring is of smaller order than
$\min_j\P\bigl(C^{(j)}>u_{j,n}\bigr)$; otherwise there is an additional term
in the limit.

\subsection{Threshold selection}\label{sec:threshold}
In practice, $k$ must be chosen from the data, and Proposition~\ref{prop:variance}
motivates the following simple rule. It consists of taking the smallest candidate beyond which the estimate
stops moving, across directions, by more than a fixed multiple of its estimated standard
deviation.  Concretely, let ${\cal Q}\subset K$ be a finite direction grid, put
$q_0=\boldsymbol 1_d$ and ${\cal Q}_0={\cal Q}\setminus\{q_0\}$, and let
$k_1<\cdots<k_L$ be a grid of candidate threshold counts.  Write
$\widehat F_{n,k}^\circ$ for the estimator computed using $k$ and let
$\widehat{\rm sd}_{n,k}(q)=\widehat F^\circ_{n,k}(q)\widehat\sigma_{n,k}(q)/\widehat r_{n,k}$
for its estimated standard deviation from \eqref{greenwoodvariance}.  Put
\begin{align}
 A(k_l)&=\max_{l'\leq l}\ \max_{q\in{\cal Q}_0}
 {\bigl|\widehat F^\circ_{n,k_l}(q)-\widehat F^\circ_{n,k_{l'}}(q)\bigr|
  \over{\widehat{\rm sd}_{n,k_{l'}}(q)}},\qquad
 \overline A(k_l)=\max_{m\leq l}A(k_m) .
 \label{stabilityrule}
\end{align}
which corresponds to the largest standardized deviation of the estimator at $k_l$ from any estimate at a
smaller threshold $k_{l'}$; a ratio with zero denominator is understood as zero if its
numerator vanishes and as $+\infty$ otherwise, so that $A(k_1)=0$ always.  Then take
$\widehat k$ to be the largest $k_l$ with
$\overline A(k_l)\leq\kappa$ for a fixed constant $\kappa$; we use $\kappa=2$ in the numerical studies below.  This is
the sequential rule of \citet{DreesKaufmann1998} adapted to the present setting,
with the deterministic normalization $\sqrt{j}$ of that rule replaced by the estimated
standard deviation, which makes the criterion comparable across directions and
dimensions.  The rule is designed to detect a stable region; see also \citet{ResnickStarica1997} and \citet{DreesDeHaanResnick2000}
and the tail-count rule of \citet{BladtGuillou2026}.

\begin{proposition}\label{prop:adaptive}
Let the candidate counts in \eqref{stabilityrule} depend on $n$, with
$k_1\to\infty$ and $k_L/n\to0$.  If the hypotheses of
Proposition~\ref{prop:variance} hold for $k=k_1$ and $F^\circ(q)<1$ for every
$q\in{\cal Q}_0$, then $\widehat k\to\infty$, $ {\widehat k/ n}\to0$, and
\begin{align*}
 \max_{q\in{\cal Q}_0}
 \bigl|\widehat F^\circ_{n,\widehat k}(q)-F^\circ(q)\bigr|
 =O_{\P}(r_n^{-1}),
\end{align*}
where $r_n$ is evaluated at $k=k_1$.  In particular, the adaptively
selected estimator is consistent on ${\cal Q}_0$.
\end{proposition}

\section{Finite-sample numerical studies}\label{sec:simulations}
We consider two numerical experiments with one common
data-generating mechanism, which we call the \emph{prime case} and which is the
one in which every hypothesis of Section~\ref{supp:mda} holds.  The event
vector has a Gumbel copula,
\begin{align}
 C_\theta(u_1,\ldots,u_d)
 =\exp\Bigl[-\Bigl\{\sum_{j=1}^d(-\log u_j)^\theta\Bigr\}^{1/\theta}\Bigr],
 \qquad \theta>1,
 \label{gumbelcopula}
\end{align}
so that the event vector has a nondegenerate tail copula $R^X$ with
$R^X(\boldsymbol 1)>0$; the censoring coordinates are
mutually independent and independent of the event vector; and each censoring margin is
chosen so that its tail, measured on the quantile scale of the corresponding event margin,
is regularly varying with a common index $-\beta_j=c/(c-1)$,
 $c$ being the asymptotic marginal fraction of censored observations in
the tail.  In two
dimensions the associated tail copula is
$R_\theta(x,y)=x+y-(x^\theta+y^\theta)^{1/\theta}$.  Different departures from the prime case are
studied separately in the Supplementary Material.

The first experiment
quantifies finite-sample error as a function of the effective number of joint extremes,
rather than by the classical sample fraction $k$.  The second explores the effect of dimension compared to the effect from the
number of joint exceedances, and gives convenient one-dimensional views of a
higher-dimensional tail surface.
The multiplicative standardization is used in both, whose event margins are
regularly varying; the Kaplan--Meier standardization, which imposes no
marginal tail model, is also studied in the Supplementary Material.

For replication
$b$, set the squared error averaged over the direction grid, and the corresponding integrated root mean squared error are given by
\begin{align}
 E_b({\cal Q})=\Bigl[{1\over|{\cal Q}_0|}\sum_{q\in{\cal Q}_0}
 \bigl\{\widehat F_b^\circ(q)-F^\circ(q)\bigr\}^2\Bigr]^{1/2}, \quad  \operatorname{IRMSE}({\cal Q})
 =\Bigl[{1\over B}\sum_{b=1}^B E_b({\cal Q})^2\Bigr]^{1/2}.
 \label{irmse}
\end{align}
The naive estimator, for comparison, is the normalized empirical joint exceedance estimator applied
directly to the standardized observed vector $Z$.  It deliberately ignores censoring
after the correct marginal standardization.  All studies use $B=500$ independent replications.

\subsection{Regular variation and effective sample size}
The goal is to measure the performance of our estimator, and we choose to perform such analysis as a function of the effective sample size
$r_n$, rather than the classical $\sqrt{k}$.
The event margins are Burr distributed with quantile function
$Q_{\tau,\lambda}(u)=\{(1-u)^{-1/\lambda}-1\}^{1/\tau}$, and respective parameters
$(\tau_1,\tau_2)=(10,5)$ and $\lambda_1=\lambda_2=1/2$.  Their implied tail indices are
$(\gamma_1,\gamma_2)=(0.2,0.4)$.  The censoring margins are also Burr distributed, with a common
$\lambda_C=1/2$ and tail indices
$\gamma_{C,j}=(1-c)\gamma_j/c$, for $c\in\{0.10,0.25,0.40\}$.  The event copula has
$\theta=2$.  We consider sample sizes $n$ from the set
${\cal N}=\{3000,6000,12000,24000\}$ and directions parametrized by
\begin{align*}
 q(a)=\bigl(s(a),t(a)\bigr)=\bigl(e^{a_+},e^{(-a)_+}\bigr),\qquad
 a\in[-\log2,\log2],
\end{align*}
on an equally spaced grid; $a=0$ is the reference and the other six points form
${\cal Q}_0$.  The target object is
\begin{align*}
 F^\circ(s,t)={R_2(s^{-1/\gamma_1},t^{-1/\gamma_2})\over R_2(1,1)}.
\end{align*}
The marginal thresholds are estimated by a censoring-adjusted Weissman construction, as follows.
Writing $Z^{(j)}_{n-k:n}$ for the threshold, $\widehat\gamma_{Z,j}$ for the Hill
estimate from the $k$ largest observations of marginal $j$, and $\widehat p_j$ for their average uncensored
indicators, we use
\begin{align}
 \widehat\gamma_j={\widehat\gamma_{Z,j}\over\widehat p_j},\qquad
 \widehat u_j=Z^{(j)}_{n-k:n}
 \Biggl\{{\widehat{\overline F}_{j,n}(Z^{(j)}_{n-k:n})\over k/n}\Biggr\}^{\widehat\gamma_j},
 \label{simthreshold}
\end{align}
where $\widehat{\overline F}_{j,n}$ is the marginal Kaplan--Meier estimator.  This is
the usual construction but with both the tail index and the exceedance
probability corrected for censoring, see also \citet{Bladt21042026}.  For each pair $(n,c)$, $k$ is chosen so that the
effective sample size $r_n^2$ takes a nominal value $r_\ast^2\in\{25,50\}$.  With
$\beta=c/(1-c)$, the relation
$r_\ast^2=kR_2(1,1)(k/n)^{2\beta}$ gives
\begin{align*}
 k=\operatorname{round}\bigl[
 \bigl\{{r_\ast^2 n^{2\beta}/ R_2(1,1)}\big\}^{1/(1+2\beta)}\big].
\end{align*}
The relation is a first order approximation, obtained by replacing $R^X_{n/k}(1,1)$ by its limit and
$G_1(n/k)G_2(n/k)$, where $G_j=\overline F_{C^{(j)}}\circ U_{X^{(j)}}$, by the pure power $(k/n)^{2\beta}$, so the effective sample size
actually obtained deviates from $r_\ast^2$ depending on censoring.  Computed
exactly for this model, $r_n^2$ goes from $17.3$ to $19.1$ at $c=0.10$ and from $25.2$
to $25.3$ at $c=0.40$ for $r_\ast^2=25$, the ranges running over the four
sample sizes in ${\cal N}$; the mean observed joint count $\widehat r_n^2$ of
Proposition~\ref{prop:rateest} tracks these values setting by setting, spanning $17.5$
to $25.5$ overall.  For $r_\ast^2=50$ the realized means go from $33.9$
to $52.2$.

We compare the performance of the directional estimator against the bivariate Dabrowska-based
estimator of \citet{BladtGoegebeurGuillou2026b} and the naive estimator.  The former is
available only in this bivariate and regularly varying setting.  For method $M$, define the
relative summary
\begin{align}
 \rho_M(c,r_\ast^2)={1\over|{\cal N}|}\sum_{n\in{\cal N}}
 {\operatorname{IRMSE}_M(n,c,r_\ast^2)\over
  \operatorname{IRMSE}_{\rm KM}(n,c,r_\ast^2)}.
 \label{relativeIRMSE}
\end{align}
The six nearly horizontal trajectories in Figure~\ref{fig:sim-effective}(a) show that
increasing $n$, and hence $k$, while targeting the same amount of joint information
does not materially reduce error.  This highlights the importance of varying the joint count, rather than the
marginal threshold count, when using the estimator in finite samples.
Across the four sample sizes, directional IRMSE ranges from $0.105$ to $0.110$,
$0.106$ to $0.111$ and $0.119$ to $0.122$ for a nominal count of $25$ as censoring
increases; for a nominal count of $50$, the corresponding ranges are $0.075$ to $0.078$,
$0.075$ to $0.076$ and $0.083$ to $0.085$.  The directional and Dabrowska estimators are
nearly indistinguishable, the six Dabrowska ratios lying between $0.99$ and $1.07$.
For the naive estimator, $\rho_{\rm naive}(c,25)$ is $0.99$, $1.26$, and $1.72$, and
$\rho_{\rm naive}(c,50)$ is $1.05$, $1.65$, and $2.38$.  Alignment with the Dabrowska
estimator is somewhat expected here, since this benchmark is designed specifically for (and is only possible in) this
setting; in this specific case, the
directional construction provides neither an accuracy advantage nor a disadvantage.  However, under
heavier censoring, the Dabrowska plane-based construction must estimate a bivariate survival function
on a grid on sets that thin out faster than those of a one-dimensional ray; the effect can already be observed near $c=0.4$, and the ratio reaches $1.26$ at $c=0.6$, as examined closer in the Supplementary Material, in favor of the directional estimator.

\begin{figure}[]
 \centering
 \includegraphics[width=\textwidth]{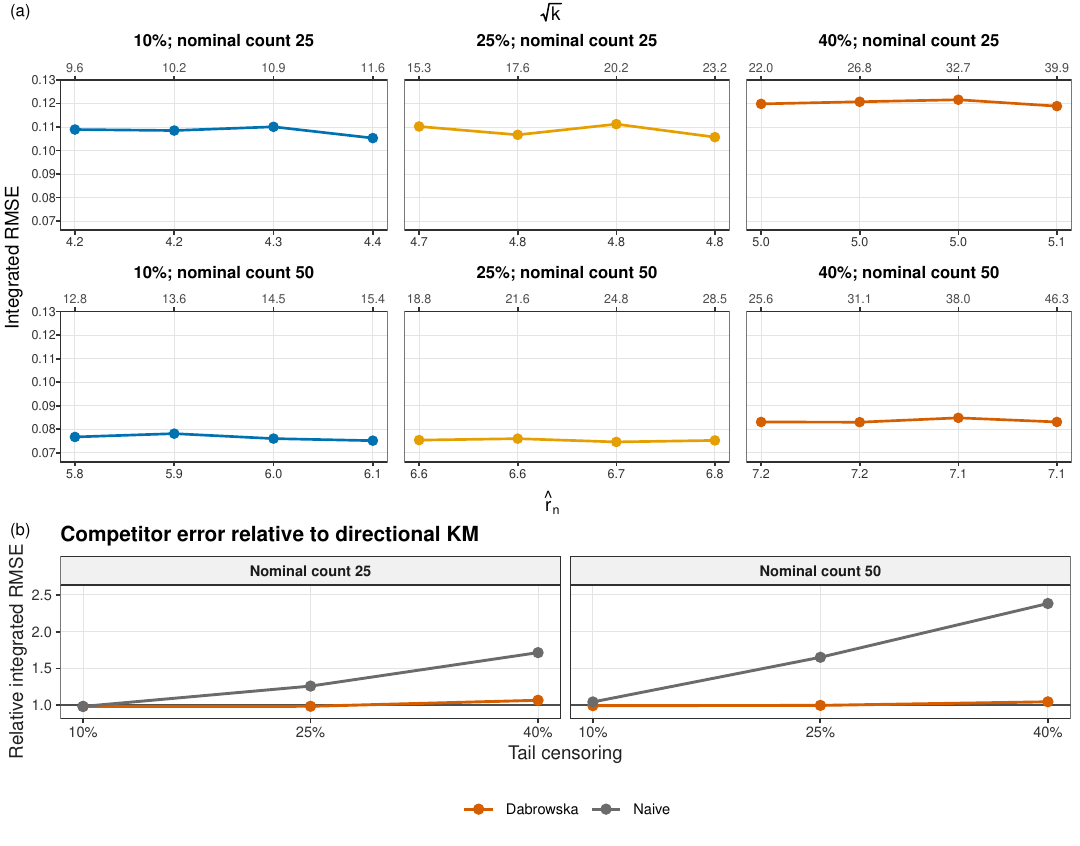}
 \caption{Regularly varying margins and effective sample size.  In each display of panel
 (a), the four points from left to right correspond to
 $n=3000,6000,12000,24000$.  The lower ticks give the rate estimator
 $\widehat r_n$ of Proposition~\ref{prop:rateest}, averaged over the replications,
 whereas the upper ticks give the
 associated marginal sample fraction $\sqrt{k}$; the censoring level and nominal
 joint exceedance count are fixed.  The two displays of panel (b) give $\rho_M(c,r_\ast^2)$ from
 \eqref{relativeIRMSE} for nominal counts $25$ and $50$.}
 \label{fig:sim-effective}
\end{figure}

\subsection{Higher dimensions and directional paths}
The main goal is to investigate deterioration of estimation performance as $d$ grows, and whether it is
intrinsic to dimension (curse of dimensionality) or whether it is mainly caused by a shrinking joint risk set.  Additionally, we inspect directional performance in a pathwise fashion, rather than only its grid
average;  since $F^\circ$ is a function of $d$ arguments, neither its surface nor its estimation error can be easily visually inspected for $d>2$.  
We study dimensions
$d=2,\ldots,10$.  The event margins are standard Pareto
with common tail index $1/2$, the event copula has $\theta=3$, and each censoring margin
is Pareto with tail index $9.5$, giving asymptotic $5\%$ marginal tail censoring.  Thresholds are again
estimated by \eqref{simthreshold}.  The $d$-variate upper tail copula is
\begin{align}
 R_{\theta,d}(x_1,\ldots,x_d)
 =\sum_{\varnothing\neq A\subseteq\{1,\ldots,d\}}(-1)^{|A|+1}
 \Bigl(\sum_{j\in A}x_j^\theta\Bigr)^{1/\theta},
 \label{gumbelR}
\end{align}
by inclusion--exclusion, and the target object is
\begin{align*}
 F^\circ(q)={R_{3,d}(q_1^{-2},\ldots,q_d^{-2})\over R_{3,d}(1,\ldots,1)}.
\end{align*}
For $S\subset\{1,\ldots,d\}$ define the path
\begin{align*}
 q_{S,j}(\lambda)=
 \begin{cases}1,&j\in S,\\2^\lambda,&j\notin S,\end{cases}
 \qquad 0\leq\lambda\leq1,
\end{align*}
at 21 equally spaced values of $\lambda$, and take
$|S|\in\{1,\lfloor d/2\rfloor,d-1\}$.  Along a path, the coordinates in $S$ remain at
the reference threshold while the remaining coordinates are changed continuously.  All paths start at $q_S(0)=\boldsymbol 1_d$, where $F^\circ=1$, and
at $\lambda=1$ every threshold outside $S$ has been doubled, at a log-linear rate.  The choices $|S|=d-1$,
$\lfloor d/2\rfloor$, and $1$ change, respectively, one, half, and $d-1$
coordinates, and hence are designed to investigate sparse, intermediate, and dense changes of the joint
exceedance set.  Exchangeability makes the particular choice of $S$ immaterial. We refer to this construction as a
directional path.

We fix $n=10000$ and first consider $k=600$ in every dimension; its mean joint count decreases from $334$
at $d=2$ to $75$ at $d=10$.  Then we consider
\begin{align*}
 k=(204,267,330,393,458,525,594,664,736)
\end{align*}
for $d=2,\ldots,10$, targeting roughly 100 joint observations; the mean lying between
$100$ and $103$.

The results in Figure~\ref{fig:sim-dimension} show that, under a common $k$, directional
IRMSE increases from $0.025$ to $0.056$ between $d=2$ and $d=10$, as qualitatively expected.  Under matched joint
effective sample size it remains between $0.045$ and $0.048$.  The naive IRMSE still increases from $0.051$ to $0.102$ in the second case.  Thus, most
of the apparent dimensional deterioration of the directional estimator is accounted
for by the scarcity of observations that reach the joint region, i.e., in this setting, the directional estimator does not seem to accumulate additional marginal or construction biases in higher dimensions; in contrast, the accumulated
censoring bias of the naive estimator remains dimension-dependent.  At $d=10$ and
$\lambda=1$, the directional means are within $0.009$ of the truth along all three
paths, whereas the naive estimator underestimates the target by $0.041$ when one
coordinate is changed and by $0.126$ when nine are changed.  The path profiles
therefore highlight the importance of censoring correction. At the heavier censoring setting
$c=0.25$, matching the joint count removes less of that
deterioration, shown in detail in the Supplementary Material, where the
scaled error $\widehat r_n\times\operatorname{IRMSE}$ rises from $0.51$ at $d=2$ to
$0.91$ at $d=10$. 
\begin{figure}[]
 \centering
 \includegraphics[width=\textwidth]{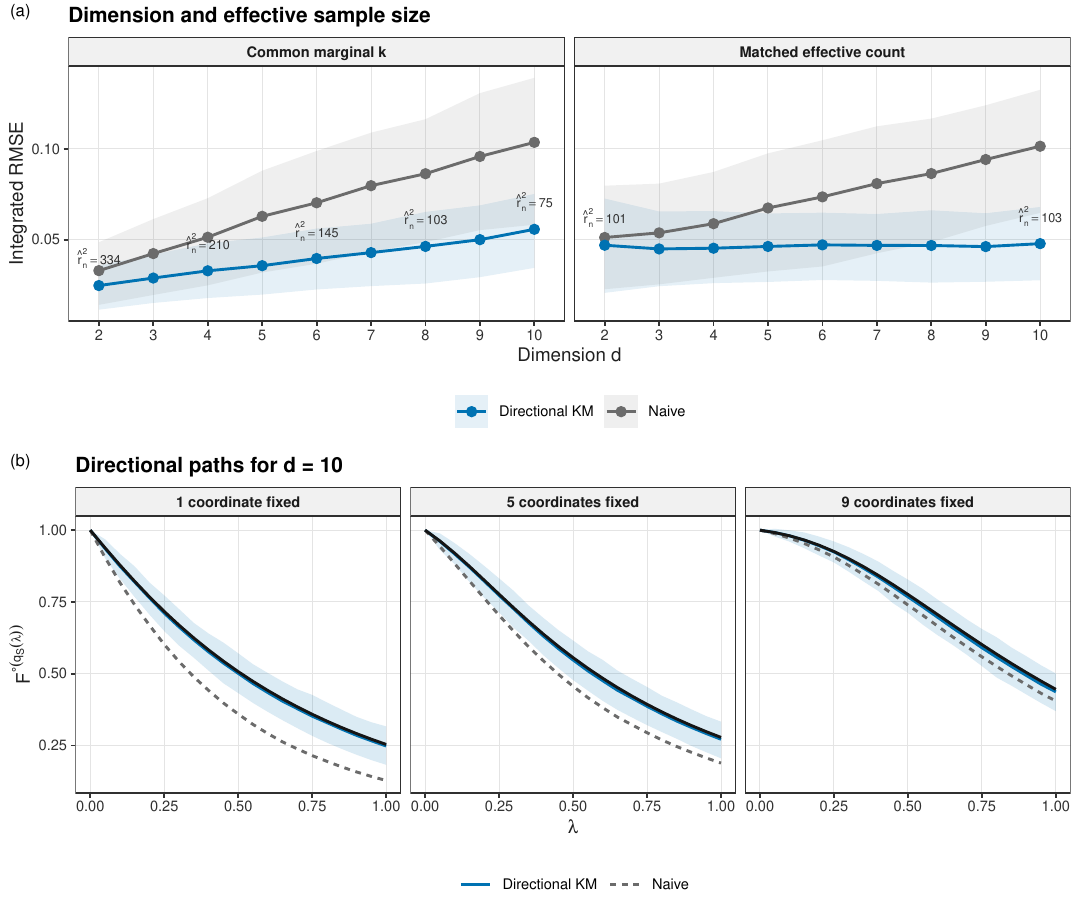}
 \caption{Dimension and effective sample size.  In panel~(a), IRMSE for a common
 marginal $k$ and for matched realized joint counts; its shading is the 10th--90th percentile range of
 $E_b({\cal Q})$ across the 500 replications.  In panel~(b), the $d=10$ directional
 path under the matched design, whose shading is again the 10th--90th percentile
 range of $\widehat F_b^\circ(q_S(\lambda))$ across the 500 replications.}
 \label{fig:sim-dimension}
\end{figure}
\subsection*{Supplementary numerical studies}
The Supplementary Material provides four further
studies. The first takes unequal Weibull margins, which are not regularly varying,
and analyzes both the target recovered by the Kaplan--Meier standardization across three
censoring regimes and the coverage of confidence intervals.  The second
replaces the estimated marginal survivals by their known values, measuring the standardization error as the sample grows and at a
shrinking sample fraction.  The third examines the coverage of confidence intervals across
dimensions two to ten at matched joint counts, and displays all $500$ scaled error
trajectories behind the present dimension-wise mean.  The fourth departs from the prime case
in four separate ways; namely, dependent censoring coordinates, heavier censoring,
asymptotic independence, and dimension growth under heavy censoring.

\section{Flood losses censored by policy limits}\label{sec:ian}
In flood insurance, policies usually cap the structure coverage and the contents coverage
separately, and the payment that a claim incurs is the loss only below the cap.
What the insurer is exposed to in a catastrophic event is a joint effect, in which a
single flood may drive both coverages far into their upper tails at once, possibly surpassing the cap. In this section, two quantities
are estimated from Hurricane Ian claims, the probability of a joint event and
the factor by which it is understated when using the capped payments
as if they were the true losses, or multiplying the two marginal tails as if the coverages
were independent.

\subsection{Data description}\label{sec:iandata}

Residential insurance policies of the United States National Flood Insurance Program insure the
structure of a building and its contents under two separate coverages.  For a single-family dwelling
the two coverages are capped by the insurer, at different levels chosen by the insured, and also overall by statute at maximal values of \$250{,}000 and \$100{,}000, respectively. What the
program reports is the indemnity paid, which equals the loss net of the deductible
when that quantity is below the cap and equals the cap otherwise.  Writing $X^{(1)}$
and $X^{(2)}$ for the indemnities that the two coverages would have paid in the absence
of the caps, and $C^{(1)},C^{(2)}$ for the two coverage limits specific to the policy, the
recorded payments are exactly the observation scheme \eqref{obs} with $d=2$.

We use the claims generated by Hurricane Ian, which made landfall in southwest Florida
on September 28, 2022, and are recorded in the OpenFEMA claims file \citep{FEMA2026}.  Of the
$48{,}776$ recorded claims we retain the $n=3774$ that are single-family primary
residences in Lee County, the landfall county, insured on a replacement-cost basis under
a single non-condominium policy and with a positive payment under both coverages. The subset is chosen to be both catastrophic and homogeneous. A
further $52$ records are discarded in which a payment strictly exceeds its own limit.  This gives $7.79\%$ of building payments censored, $13.51\%$ of
contents payments, $3.58\%$ censored in both coordinates simultaneously, and $17.73\%$
in at least one.

Some remarks are worth making.  First, no independence between the two coverage limits is
required, since Assumption~\ref{a:RC} is an assumption on the joint standardized censoring law.  The
ratio of
$\P(C^{(1)}>x_1,C^{(2)}>x_2)$ to the product of its margins is $1.002$ at
$x=(\$100{,}000,\$40{,}000)$ and $1.065$ at $(\$200{,}000,\$40{,}000)$, eventually reaching the joint atom at the statutory maximum, so the limits are mildly dependent in
the tail; this is admissible, and affects the conclusions only through the joint
censoring index $\beta^\ast$.  The intervals \eqref{greenwoodCI} reported below use
Proposition~\ref{prop:variance}, hence Assumption~\ref{a:pluginnorm}, which under this
standardization requires that $\beta^\ast$ exceed $\max_j\beta_j$; the feeble dependence supports it here. Second, the independence of $X$ and $C$ is important but not identifiable from \eqref{obs}, since the law of
$(Z,\delta)$ does not determine that of $(X,C)$; see for instance \citet{Tsiatis1975}.  It is therefore
assumed here.  What supports it on
these data is the design, since the coverage is fixed by
contract before any event generates a loss.

\subsection{Standardization, estimable range, and threshold}\label{sec:ianstd}
We use the Kaplan--Meier standardization, whose target $R^X(q^{-1})/R^X(\boldsymbol1)$ is the
tail copula of the loss vector itself. The
multiplicative standardization would instead target
$R^X(q^{-1/\gamma})/R^X(\boldsymbol1)$, the tail copula on the marginal tail scales,
and require $\gamma_{X^{(j)}}$ in each coordinate.  The choice matters little for this dataset, as the margins seem regularly varying under simple diagnostics (omitted for brevity); the difference between standardizations is actually only between $0.1\%$ and $4.0\%$
at the six threshold pairs considered below.

In the context of an upper limit, Remark~\ref{rem:identifiability} is applicable here, and restricts the choice of sample fraction in the analysis.
In this case, a policy coverage limit has a distributional upper endpoint below that of the loss, so
$\overline F_{X^{(j)}}$ is identified only below the eventual plateau of the marginal
Kaplan--Meier estimator.  Those plateaus have heights $0.0750$ and $0.0668$. This implies, for instance, that with $T=2$ the estimator is well
defined for $k\geq567$.  Within that range we take the direction family
\begin{align*}
 q(a)=\bigl(e^{a_+},e^{(-a)_+}\bigr),\qquad a\in[-\log2,\log2],
\end{align*}
on thirteen equally spaced points, and select $k$ automatically by the rule
\eqref{stabilityrule} over the choices $k\in\{575,600,\ldots,1150\}$.  It returns $\widehat
k=750$, that is, a tail fraction of $0.199$, with a resulting $351$ joint observations used to construct the directional estimator.  The two edge directions of the grid are shown across the whole range in
panel~(a) of Figure~\ref{fig:ian}, together with their pointwise confidence intervals.  The
estimates move by less than one standard deviation over $k\in[575,800]$ and drift
downward thereafter, which is also what \eqref{stabilityrule} detected.

Since the tail fraction satisfies the fixed-sample constraint $k/n\ge 0.150$, the estimate
$\widehat r_n$ cannot be arbitrarily small. From a practical perspective, this means that the flood data exhibit an interplay between
censoring and extremes;  the censoring is heaviest
exactly where extremes occur, and thus the two cannot be completely disentangled. This type of censoring is quite common in applications involving payment sizes, which are hard-capped. Larger insurers have larger portfolios, larger reserves, and thus allow for higher caps, such that in this case for the upper endpoint of the censoring distribution can be argued to be increasing with $n$. However, without such an argument, the far tail of the distribution is asymptotically intractable by design.

\subsection{Results}\label{sec:ianresults}
The estimator $\widehat F^\circ_n$ is shown over the directions in the grid in panel~(b) of
Figure~\ref{fig:ian}, together with its pointwise confidence intervals, against two
references.  These are the naive estimator, which applies the same standardization but
ignores the censoring; and the value $\prod_jq_j^{-1}$ that the ratio
$p_n(q)/p_n(\boldsymbol 1)$ would take if the two standardized coordinates were exactly
independent.  At the two ends of the grid, the directional estimates are $0.641$, with
interval $(0.588,0.699)$, and $0.595$, with interval $(0.547,0.647)$, against $0.500$
under independence, so the joint tail departs from independence by relatively
the same amount in the two directions.  The naive estimator gives $0.581$ and $0.610$.
The fact that it lies below the directional estimate in one direction and above it in the other
is the typical feature of asymmetric censoring, the contents coordinate being censored almost
twice as severely as the building coordinate. This is also in line with the effects observed in
the simulations of the Supplementary Material. Contents totals are usually compiled from an itemized inventory, while a building figure has to be downright estimated, which could contribute to the asymmetry of the two curves.

The estimator is further reported on the monetary scale in panel~(c), applied directly to the payments at six increasing pairs
of dollar thresholds, and the $y$-axis showing the joint exceedance probability.  The
directional estimates run from $0.340$ at $(\$120{,}000,\$25{,}000)$ to $0.0475$ at
$(\$245{,}000,\$75{,}000)$, with $1229$ and $145$ observations reaching the two
edge thresholds.  Against the naive estimator, the
correction grows monotonically with the threshold, from a factor $1.074$
to a factor $1.255$, so that ignoring the censoring effects understates the joint tail, and
understates it more the further into the tail one goes.  Against the product of the two marginal
Kaplan--Meier survival probabilities, the estimate is larger by an important factor rising from
$1.64$ to $4.82$.  Treating the two coverages as censored but independent observations would therefore severely understate
the probability of a simultaneous exceedance.

\begin{figure}[]
 \centering
 \includegraphics[width=\textwidth]{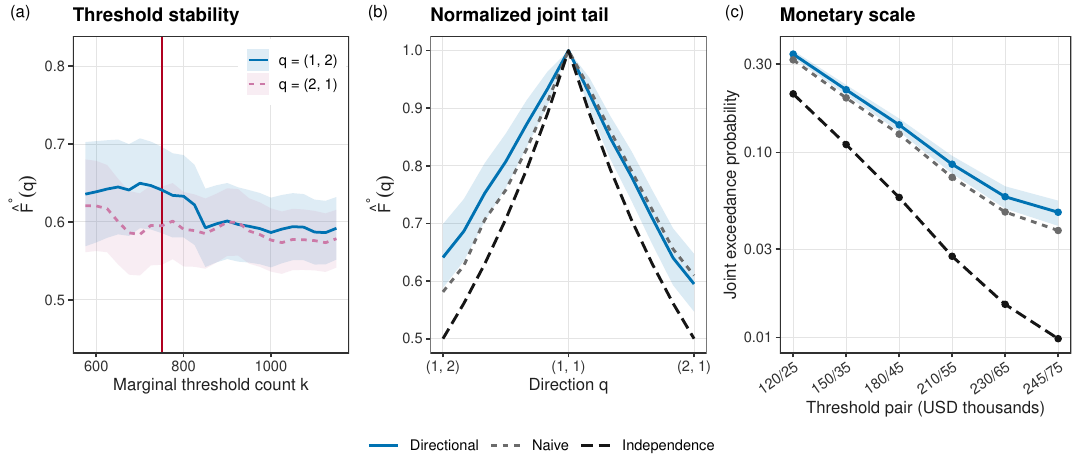}
 \caption{Hurricane Ian claims.  In panel~(a),
 $\widehat F^\circ_n(q)$ is shown at the two extreme directions of the grid as a function of the
 marginal sample fraction $k$, with pointwise $95\%$ intervals; the vertical line is
 the $\widehat k=750$ selected by \eqref{stabilityrule}.  In panel~(b),
 $\widehat F^\circ_n$ is shown over the direction grid at $\widehat k$, with plug-in pointwise $95\%$
 intervals \eqref{greenwoodCI}, the naive estimator, and the independence
 reference $\prod_jq_j^{-1}$.  In panel~(c), joint exceedance probabilities are provided as a function of
 actual monetary claim sizes, on a logarithmic scale, again together with the two reference estimators; each tick label gives the
 building and the contents threshold, in USD thousands, of the pair at which the joint
 exceedance probability is evaluated.}
 \label{fig:ian}
\end{figure}

Two further checks are performed, both provided in Figure~\ref{fig:iancheck}.  First, a nonparametric bootstrap that draws $500$ samples of the
$3774$ claims and repeats the entire procedure on each, including the two marginal
Kaplan--Meier standardizations at the same $\widehat k$. This gives bootstrap standard
deviations of $\log\widehat F^\circ_n$ that are between $0.69$ and $0.95$ times the
plug-in ones of Proposition~\ref{prop:variance}, with median $0.82$; the intervals
\eqref{greenwoodCI} are therefore conservative here, which matches the conclusion from the coverage
simulation study.  Finally, the dataset also provides an adjuster's assessment of the
damage, which is available for $3770$ of the $3774$ claims and is in principle a direct
proxy for $X$.  It is not used as a substitute for $X$, and should not be.  Among the claims whose
building payment was censored, two thirds have an assessed damage no larger than the
coverage plus the deductible, with a median ratio to the coverage being equal to $0.988$. Thus,
beyond the cap, the assessment is an unreliable estimate of the
loss. In panel~(c) we show what happens when considering the assessed damage as
the loss and counting joint exceedances directly. We observe probabilities below the
directional estimates at every pair, by a factor falling from $0.90$ to $0.52$ as the
pair deepens, and below the naive estimator as well, which falls only from $0.93$ to
$0.80$.  Substituting the assessments for $X$ is therefore even worse than ignoring the censoring
altogether.  The mathematical reason is that the building assessments do not reach the cap but stop
just short of it; for instance, at the highest threshold pair $47\%$ of the claims whose building payment exceeds
$\$245{,}000$ carry an assessment that does not. There is also a possible causal explanation. An adjuster settles a claim rather than accurately measuring a
loss, and once the damage is clearly past the coverage there is little to be gained from
refining the figure. This highlights the importance of using survival methods to get accurate benchmark estimates.

\begin{figure}[h]
 \centering
 \includegraphics[width=\textwidth]{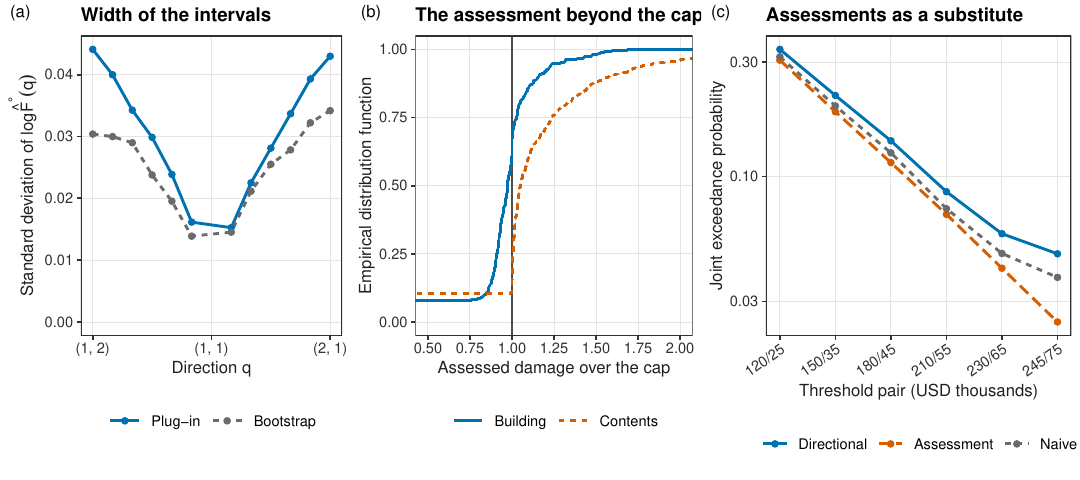}
 \caption{The two checks.  In panel~(a), the standard deviation of
 $\log\widehat F^\circ_n(q)$ is shown over the direction grid, as given by
 Proposition~\ref{prop:variance} and as returned by the nonparametric bootstrap; the
 reference direction is omitted, both vanishing there.  In panel~(b), we show the empirical
 distribution function of the assessed damage divided by the coverage plus the
 deductible, taken over the claims censored in that coordinate, with the vertical line
 marking the standardized cap, and the $x$-axis restricted to $[0.5,2]$.  In panel~(c), the
 joint exceedance probabilities of Section~\ref{sec:ianresults} are plotted against those obtained
 by treating the assessed damage as the loss and counting exceedances directly, with the
 naive estimator for reference, in logarithmic scale.}
 \label{fig:iancheck}
\end{figure}

\section{Discussion and outlook}\label{sec:discussion}
The key directional projection in \eqref{reduction} is an exact identity, and it thus allows for a mathematical treatment through decompositions of the Kaplan--Meier estimator.  Its
inevitable difficulty, however, is that the reduced data are not easily handled by elementary methods. Nonetheless, in this paper we proposed conditions which are understandable from either survival analysis or extreme value theory literature, and which ensure tractability.  Assumption~\ref{a:radialmodulus} is the most specialized, and
bounds the growth of the risk set along a ray, while the simpler Assumption~\ref{a:rate}
provides a tail rate; both are expected from an extremes perspective.  Independence of
$X$ and $C$ is crucial, as is standard in survival analysis, but also unfortunately not testable without any further parametric assumptions.  The condition $R^X(\boldsymbol 1)>0$
excludes asymptotic independence for the simple but crucial reason that the
effective sample size is $r_n^2\sim kR^X(\boldsymbol 1)c_n(\boldsymbol 1)$, so that asymptotic independence would require a different analysis with another rate of convergence. A notable extension to the bivariate literature is that no dependence structure is imposed on the censoring vector, and Assumption~\ref{a:RC} allows for independent and
asymptotically dependent censoring coordinates alike.

Three questions are left open for future work.

\emph{Efficiency.}  No efficiency claim is made here.  For the full multivariate
survival function under censoring the nonparametric maximum likelihood estimator is
inefficient \citep{GillVanDerLaanWellner1995,vanderLaan1996}. However,
the directional product limit uses, at each $q$, the one-dimensional reduction
of the data.  Whether the resulting asymptotic variance attains the
nonparametric bound for $F^\circ$ is unknown.  The boundary case is
settled, since without censoring, i.e., when $\Delta_{i,n}\equiv1$, the product limit \eqref{KM} is
the empirical survival function of the reduced sample, and \eqref{estimator} is the
empirical joint-exceedance estimator, whose optimality is classical, see
\citet{DreesHuang1998}.

\emph{Identifying cases where the standardization contributes.}  When
the regularity of Remark~\ref{rem:regular} is not implied by
Assumption~\ref{a:pluginnorm}, the process
$\mathbb M$ requires identification for specific standardizations. For instance, we showed that it is equal to zero for the multiplicative standardization. However, in some instances, for the Kaplan--Meier standardization, a joint weak limit of the marginal processes and the directional one is required, as well as an asymptotic expansion of the
displacement of the pair $(W_{i,n}(q),\Delta_{i,n}(q))$ under the estimated
standardization; the indicator and the radius move together at a near-tie, so the two cannot
be separated.

\emph{Asymptotic independence.}  Extending the theory to the near-independence regime
of \citet{LedfordTawn1996} and to hidden regular variation \citep{Resnick2002}, where
$p_n(q)$ decays at a rate slower than the one $r_n$ is built from, requires a different
scaling regime, but otherwise seems promising under similar techniques, albeit under additional conditions.

\section*{Acknowledgments}

The author was supported by the Carlsberg Foundation, grant CF23-1096.

\section*{Data and code availability}

The Hurricane Ian data analyzed in Section~\ref{sec:ian} are part of the public OpenFEMA data set
\textit{FIMA NFIP Redacted Claims} \citep{FEMA2026}, retrieved from
\url{https://www.fema.gov/api/open/v2/FimaNfipClaims}.  The subset in
Section~\ref{sec:iandata} is reproduced exactly in the code, which also
reproduces every simulation of Section~\ref{sec:simulations} and
of the Supplementary Material.  The code is openly available at
\url{https://github.com/martinbladt/Multivariate-Censored-Extremes}.

\clearpage

\setcounter{section}{0}
\setcounter{equation}{0}
\setcounter{figure}{0}
\setcounter{table}{0}
\setcounter{theorem}{0}
\setcounter{corollary}{0}
\setcounter{lemma}{0}
\setcounter{proposition}{0}
\setcounter{remark}{0}
\setcounter{assumption}{0}

\renewcommand{\thesection}{S\arabic{section}}
\renewcommand{\theequation}{S\arabic{equation}}
\renewcommand{\thefigure}{S\arabic{figure}}
\renewcommand{\thetable}{S\arabic{table}}
\renewcommand{\thetheorem}{S\arabic{theorem}}
\renewcommand{\thecorollary}{S\arabic{corollary}}
\renewcommand{\thelemma}{S\arabic{lemma}}
\renewcommand{\theproposition}{S\arabic{proposition}}
\renewcommand{\theremark}{S\arabic{remark}}
\renewcommand{\theassumption}{S\arabic{assumption}}

\def\theHsection{S\arabic{section}}
\def\theHequation{S\arabic{equation}}
\def\theHfigure{S\arabic{figure}}
\def\theHtable{S\arabic{table}}
\def\theHtheorem{S\arabic{theorem}}
\def\theHcorollary{S\arabic{corollary}}
\def\theHlemma{S\arabic{lemma}}
\def\theHproposition{S\arabic{proposition}}
\def\theHremark{S\arabic{remark}}
\def\theHassumption{S\arabic{assumption}}

\addtocontents{toc}{\protect\setcounter{tocdepth}{2}}

\begin{center}
{\large\bfseries Supplementary material for\\[2pt]
``Statistics of multivariate extremes under random censoring''}\\[8pt]
{\normalsize Martin Bladt}\\[2pt]
{\normalsize Department of Mathematical Sciences, University of Copenhagen}
\end{center}

\begingroup
\hypersetup{linkcolor=black}   
\tableofcontents
\endgroup

\newpage
\section{Concrete standardizations}\label{supp:mda}
The conditions of Section~\ref{sec:setup} hold for general standardizing functions
$\psi_{1,n},\ldots,\psi_{d,n}$.  This section makes the two concrete choices
\begin{align*}
 \psi_{j,n}(x)={x}/{u_{j,n}}
 \qquad\text{and}\qquad
 \psi_{j,n}(x)=\frac{k/n}{\overline F_{X^{(j)}}(x)} ,
\end{align*}
and, for each, states the resulting consistency and weak-convergence
results explicitly.  Recall that, throughout this section, $k=k_n$ is an intermediate sequence,
$k\to\infty$ and $k/n\to0$, and $u_{j,n}=U_{X^{(j)}}(n/k)$.

The two subsections have a similar structure for convenience to the reader. The proofs are collected in
Section~\ref{sec:corollaryproofs} and they simply verify the conditions of
Sections~\ref{sec:setup} and~\ref{sec:asymptotics} one by one.

Both subsections use the same structural hypotheses, on the dependence and
on the censoring, collected in Assumption~\ref{ass:supp}; Assumption~\ref{a:margins}
is assumed throughout.

\begin{assumption}[Dependence and censoring]\label{ass:supp}\leavevmode
\begin{aparts}
\item\label{a:RX} There is a function $R^X$ on
$(0,\infty)^d$, the tail copula of $X$, with
\begin{align*}
 R_r^X(x):=
 r\,\P\Bigl(
 \overline F_{X^{(j)}}(X^{(j)})\leq{x_j\over r},\ j=1,\ldots,d
 \Bigr)\longrightarrow R^X(x),\qquad r\to\infty,
\end{align*}
locally uniformly on $(0,\infty)^d$, and asymptotic independence is excluded by assuming
$R^X(\boldsymbol 1)>0$.

\item\label{a:RC} Let
$\Gamma(y;x):=\P\bigl(\overline F_{X^{(j)}}(C^{(j)})\leq x_j/y,\ j=1,\ldots,d\bigr)$.
There are a joint censoring index $\beta^\ast\in(0,\infty)$ and a function $R^C$ on
$(0,\infty)^d$ such that $y\mapsto\Gamma(y;\boldsymbol 1)$ is regularly varying at infinity
with index $-\beta^\ast$, and
\begin{align*}
 {\Gamma(y;x)\over\Gamma(y;\boldsymbol 1)}\longrightarrow R^C(x),\qquad y\to\infty,
\end{align*}
locally uniformly on $(0,\infty)^d$, where $R^C$ is continuous and strictly positive,
$R^C(\boldsymbol 1)=1$, and, for each $j$, $R^C(x)\to0$ as $x_j\to0$ with the remaining
coordinates held fixed.

\item\label{a:Cbeta} For $j=1,\ldots,d$, the composition
\begin{align*}
 G_j:=\overline F_{C^{(j)}}\circ U_{X^{(j)}}\ \text{ is regularly varying at infinity
 with index }-\beta_j,\qquad \beta_j\in[0,\infty).
\end{align*}

\item\label{a:effective} The effective sample size diverges:
\begin{align*}
 {k\,\P\bigl(C^{(j)}>u_{j,n},\ j=1,\ldots,d\bigr)\over\log^2(n/k)}
 \longrightarrow\infty .
\end{align*}

\end{aparts}
\end{assumption}

\begin{remark}\label{rem:asmS1}
Because the margins of $X$ are continuous, $\overline
F_{X^{(j)}}(X^{(j)})$ is uniform on $(0,1)$, and $R^X_r$ is nondecreasing in each
argument and Lipschitz with constant one in each argument; the limit $R^X$ inherits
both properties and, replacing $r$ by $r/w$ in Assumption~\ref{a:RX}, one obtains homogeneity:
\begin{align}
 R^X(wx)=w\,R^X(x),\qquad w>0 .
 \label{homog}
\end{align}
Taking $w=\min_jx_j$ and using monotonicity, one also obtains
\begin{align}
 \bigl(\min_jx_j\bigr)\,R^X(\boldsymbol 1)\ \leq\ R^X(x)\ \leq\ \min_jx_j,
 \qquad x\in(0,\infty)^d ,
 \label{copulalower}
\end{align}
the upper bound holding because $R^X_r(x)$ is at most its value when all coordinates
but the minimizing one are set to infinity.

The function $G_j$ of Assumption~\ref{a:Cbeta} measures the censoring tail on the quantile scale of the
corresponding margin of $X$.  The value $\beta_j=0$
includes, in particular, a coordinate that is observed exactly, for which $G_j\equiv1$;
it may also describe genuine censoring whose survival decreases only slowly.  We reserve
\emph{exactly uncensored} for $G_j\equiv1$. 

When the coordinates of $C$ are independent and every $\beta_j>0$,
Assumption~\ref{a:RC} holds with
$R^C(x)=\prod_jx_j^{\beta_j}$ and $\beta^\ast=\sum_j\beta_j$, since then the joint
probability factorises and each factor is regularly varying by Assumption~\ref{a:Cbeta}.  Clearly
$R^C$ is nondecreasing in each argument, and writing $\Gamma(y;\lambda x)=\Gamma(y/\lambda;x)$
and using the regular variation of $\Gamma(\cdot\,;\boldsymbol 1)$ gives the homogeneity
\begin{align}
 R^C(\lambda x)=\lambda^{\beta^\ast}R^C(x),\qquad \lambda>0 .
 \label{homogC}
\end{align}
Since $\Gamma(y;\boldsymbol 1)\leq G_j(y)$ for every $j$, comparing indices of
regular variation gives $\beta^\ast\geq\max_j\beta_j$. In case of equality we say that
the coordinates of $C$ asymptotically dependent.  

The following consequences are worth noting.  First, by Karamata's theorem applied to
 $\P\bigl(\delta^{(j)}=1,\ \overline F_{X^{(j)}}(X^{(j)})\leq1/y\bigr)
=\int_y^\infty G_j(z)z^{-2}dz$, the limiting proportion of non-censored observations
above a high marginal level is
\begin{align}
 \lim_{u\uparrow\,x^*_{Z^{(j)}}}\P\bigl(\delta^{(j)}=1\mid Z^{(j)}>u\bigr)
 ={1\over1+\beta_j}\in(0,1],
 \label{pbeta}
\end{align}
so $\beta_j<\infty$ excludes the degenerate regime in which the extreme risk set
consists asymptotically of censored observations only.  Second, $G_j$ is regularly
varying and hence strictly positive on $(0,\infty)$, so the upper endpoint of
$C^{(j)}$ is not below that of $X^{(j)}$ and the tail of $X^{(j)}$ remains
identifiable.
\end{remark}

\begin{remark}\label{rem:identifiability}
If $x^*_{C^{(j)}}<x^*_{X^{(j)}}$ for some $j$, then Assumption~\ref{a:Cbeta} cannot hold.  Nothing about $\overline F_{X^{(j)}}$ below the level
$\overline F_{X^{(j)}}(x^*_{C^{(j)}})$ is then identifiable, and neither
is $p_n(q)$ for any $q$ whose $j$th standardized coordinate lies beyond that level.  All the
statements below are therefore to be understood on the estimable direction set
$\{q:\ c_n(q)>0\}$, whose empirical counterpart is the set of directions at which
the marginal Kaplan--Meier estimators of Section~\ref{sec:KM} have not yet reached
their plateau.
\end{remark}

\medskip

\noindent\textit{The plug-in rate.} Under the Kaplan--Meier standardization the
standardizing functions are estimable at the rate
$\zeta_n=\max_j\{kG_j(n/k)\}^{-1/2}$, the reciprocal square root of the number of
observations above the marginal threshold of the most heavily censored coordinate; see for instance, in the regularly varying case,
\citet{EinmahlFilsVilletardGuillou2008S,BeirlantEtAl2007S,WormsWorms2014S}. Under the
multiplicative standardization only a marginal quantile has to be estimated.  Since $r_n^2=n\,p_n(\boldsymbol 1)c_n(\boldsymbol 1)$ by
definition,
\begin{align}
 r_n^2\zeta_n^2={n\,p_n(\boldsymbol 1)\,c_n(\boldsymbol 1)\over k\,\min_jG_j(n/k)},
 \label{freeplugin}
\end{align}
so the marginal standardization is estimated at a rate faster than $r_n$, and
the standardization is regular (Remark~\ref{rem:regular}) with $\mathbb M\equiv0$, when the joint
censoring probability $c_n(\boldsymbol 1)$ is of smaller order than
$\min_jG_j(n/k)$ by more than a power of $\log(n/k)$.  That happens when the coordinates of $C$ are independent and at least
two censoring indices are positive, and it fails when they are asymptotically dependent,
the joint probability then being of the order of the smallest margin.  In the latter case, the marginal estimation contributes $\mathbb M$ in the limit.

\subsection{Regularly varying margins and the multiplicative standardization}\label{sec:RV}
Here the margins of $X$ are regularly varying, as in Assumption~\ref{a:RVX} below, and
we take the multiplicative standardization
\begin{align}
 \psi_{j,n}(x)={x\over u_{j,n}},\qquad
 \widehat\psi_{j,n}(x)={x\over\widehat u_{j,n}},
 \label{RVpsi}
\end{align}
where $\widehat u_{j,n}>0$ is a suitable estimator of $u_{j,n}$.

\medskip

\noindent\textit{The target.}  Since $\overline F_{X^{(j)}}(u_{j,n})=k/n$ by continuity
of the margins, Assumptions~\ref{a:RVX} and~\ref{a:RX} give, writing
$q^{-1/\gamma}=(q_1^{-1/\gamma_{X^{(1)}}},\ldots,q_d^{-1/\gamma_{X^{(d)}}})$,
\begin{align}
 {n\over k}\,p_n(q)\longrightarrow R^X\bigl(q^{-1/\gamma}\bigr)
 \label{pnlimit}
\end{align}
locally uniformly on $(0,\infty)^d$, so that the estimator has the target
\begin{align}
 F^\circ(q)={R^X\bigl(q^{-1/\gamma}\bigr)\over R^X(\boldsymbol 1)} ,
 \label{Fcirc}
\end{align}
which corresponds to the tail copula of $X$ on the marginal tail scales.  This is the
classical estimand in the heavy-tailed context.

\medskip

\noindent\textit{The conditions.}  The conditions specific to this standardization are
given in Assumption~\ref{ass:RV}.

\begin{assumption}[The multiplicative standardization]\label{ass:RV}\leavevmode
\begin{aparts}
\item\label{a:RVX} The margins of $X$ are regularly varying: for $j=1,\ldots,d$ and some
$\gamma_{X^{(j)}}>0$, as $r\to\infty$,
\begin{align*}
 {\overline F_{X^{(j)}}(rx)\over \overline F_{X^{(j)}}(r)}
 \longrightarrow x^{-1/\gamma_{X^{(j)}}},\qquad x>0 ,
\end{align*}
and $F_{X^{(j)}}$ is strictly increasing on $\bigl(0,x^*_{X^{(j)}}\bigr)$.

\item\label{a:radialRV} A radial regularity requirement on the joint law of the
observed vector: there is $\rho>0$ such that
\begin{align*}
 \P\bigl(Z^{(j)}>z_j,\ j=1,\ldots,d\bigr)
 \ \leq\ A^{\rho}\,\P\bigl(Z^{(j)}>Az_j,\ j=1,\ldots,d\bigr),
 \qquad A\geq1,\ z\in(0,\infty)^d .
\end{align*}

\item\label{a:quantilecons} Only $u_{j,n}$ has to be estimated here, and we
assume for consistency
\begin{align*}
 \log{n\over k}\,\bigl({\widehat u_{j,n}/ u_{j,n}}-1\bigr)
 \overset{\P}{\longrightarrow}0,\qquad j=1,\ldots,d.
\end{align*}

\item\label{a:quantilerate} For weak convergence, the stronger
\begin{align*}
 {\sqrt k\over\log^2(n/k)}\bigl({\widehat u_{j,n}/ u_{j,n}}-1\bigr)=O_{\P}(1),
 \qquad j=1,\ldots,d .
\end{align*}

\item\label{a:Hall} For the margins
$V\in\{X^{(1)},\ldots,X^{(d)}\}$, the Hall-type model
\begin{align*}
 \overline F_V(x)=A_Vx^{-1/\gamma_V}\bigl\{1+{1\over\gamma_V}\delta_V(x)\bigr\},
\end{align*}
where $A_V>0$, $\gamma_V>0$, and $|\delta_V|$ is regularly varying with negative index
\citep{deHaanFerreira2006S}.

\item\label{a:S} For the tail copula function, there are
$\alpha(r)\downarrow0$ and a nonzero continuous function $M$ with
\begin{align*}
 {R_r^X(x)-R^X(x)\over\alpha(r)}\longrightarrow M(x)
\end{align*}
uniformly on compact subsets of $(0,\infty)^d$.

\item\label{a:rateXR} The rates for the two preceding functions satisfy
\begin{align*}
 \sqrt{k}\,\bigl|\delta_{X^{(j)}}\bigl(U_{Z^{(j)}}(n/k)\bigr)\bigr|=O(1),
 \quad j=1,\ldots,d,
 \qquad\text{and}\qquad
 r_n\,\alpha(n/k)\longrightarrow0 .
\end{align*}
\end{aparts}
\end{assumption}

\begin{remark}\label{rem:asmS2}
Assumption~\ref{a:effective} is the effective sample
size condition of the heavy-tailed setting.  Assumption~\ref{a:radialRV}
holds, for instance, whenever the joint survival function of $Z$ is of the form
$x^{-\rho}\ell_{z}(x)$ along each ray $z\in(0,\infty)^d$ with $\ell_{z}$
nondecreasing, which
covers the Pareto-type joint models considered in practice.  Any $\rho$ exceeding
$\sum_j1/\gamma_{Z^{(j)}}$, where
$\gamma_{Z^{(j)}}=(\gamma_{X^{(j)}}^{-1}+\gamma_{C^{(j)}}^{-1})^{-1}$ with
$\gamma_{C^{(j)}}=\gamma_{X^{(j)}}/\beta_j$, is then admissible
for the tail behavior. The condition additionally rules out oscillation at large
$x$.

See for instance \citet{Bladt21042026} for an estimator satisfying Assumptions~\ref{a:quantilecons} and~\ref{a:quantilerate} under censoring.  The last three parts control the bias,
which here derives from standard second-order conditions of heavy-tailed theory.
\end{remark}

\begin{corollary}\label{cor:RVcons}
Assume Assumptions~\ref{ass:supp} and~\ref{ass:RV}.1--\ref{ass:RV}.3.  Then
\begin{align*}
 \sup_{q\in K}|\widehat F_n^\circ(q)-F^\circ(q)|\overset{\P}\longrightarrow0 .
\end{align*}
\end{corollary}

\begin{corollary}\label{cor:RVnorm}
Assume the hypotheses of Corollary~\ref{cor:RVcons} with Assumption~\ref{a:quantilecons}
replaced by Assumption~\ref{a:quantilerate}, together with the second order conditions
Assumptions~\ref{a:Hall}--\ref{a:rateXR}.
Then
\begin{align*}
 r_n\bigl\{\widehat F_n^\circ-F^\circ\bigr\}\ \leadsto\ F^\circ\,\mathbb G
 \qquad\text{in }\ell^\infty(K),
\end{align*}
with $\mathbb G$ as in \eqref{Glimit}.
\end{corollary}

\subsection{Arbitrary margins and the Kaplan--Meier standardization}\label{sec:KM}
We now assume no marginal model at all. Let $x^*_{V}=\sup\{x:F_V(x)<1\}$, and take the standardization
\begin{align}
 \psi_{j,n}(x)={k/n\over\overline F_{X^{(j)}}(x)},\qquad
 \widehat\psi_{j,n}(x)={k/n\over\widehat{\overline F}_{j,n}(x)} ,
 \label{KMpsi}
\end{align}
where $\widehat{\overline F}_{j,n}$ is the (non-directional) Kaplan--Meier estimator of
$\overline F_{X^{(j)}}$ computed from $(Z^{(j)}_i,\delta^{(j)}_i)$, $i\leq n$.  Under
Assumption~\ref{a:Mstrict} the map $\psi_{j,n}$ is continuous and strictly increasing on
$(0,x^*_{X^{(j)}})$, which contains $Z^{(j)}$ almost surely, with
$\psi_{j,n}(u_{j,n})=1$; and $\widehat\psi_{j,n}$ is nondecreasing and, since the
Kaplan--Meier estimator is decreasing strictly at each uncensored observation, injective
on $\{Z^{(j)}_i:\delta^{(j)}_i=1\}$, as Section~\ref{sec:plugin} requires.

\medskip

\noindent\textit{The target.}  Because $F_{X^{(j)}}$ is assumed continuous,
$\P\bigl(\psi_{j,n}(X^{(j)})>s\bigr)=(k/n)/s$ for $s\geq k/n$, so that each standardized
margin is \emph{exactly} standard Pareto above the level $k/n$ and for every $n$.  The same computation in $d$ dimensions
gives, directly from the definition of $R^X_r$ in Assumption~\ref{a:RX} with $r=n/k$ and writing
$q^{-1}=(q_1^{-1},\ldots,q_d^{-1})$,
\begin{align}
 p_n(q)=\P\Bigl(\overline F_{X^{(j)}}(X^{(j)})<{k/n\over q_j},\ j=1,\ldots,d\Bigr)
 ={k\over n}\,R^X_{n/k}\bigl(q^{-1}\bigr).
 \label{pnKM}
\end{align}
Dividing by its value at $\boldsymbol 1$ and letting
$n\to\infty$ in Assumption~\ref{a:RX}, we get
\begin{align}
 {p_n(q)\over p_n(\boldsymbol 1)}={R^X_{n/k}(q^{-1})\over R^X_{n/k}(\boldsymbol 1)}
 \longrightarrow F^\circ(q)={R^X(q^{-1})\over R^X(\boldsymbol 1)}
 \label{KMF}
\end{align}
locally uniformly on $(0,\infty)^d$.  Hence the target is the tail copula of $X$ itself, now
on the marginal rank scale and having no marginal parameters.  By
\eqref{homog} it is homogeneous of order $-1$.

\medskip
\noindent\textit{The conditions.}  The conditions specific to this standardization are
given in Assumption~\ref{ass:KM}.

\begin{assumption}[The Kaplan--Meier standardization]\label{ass:KM}\leavevmode
\begin{aparts}
\item\label{a:Mstrict} The marginal distribution functions satisfy
\begin{align*}
 F_{X^{(j)}}\ \text{ is strictly increasing on }
 \bigl(0,x^*_{X^{(j)}}\bigr),\qquad j=1,\ldots,d.
\end{align*}

\item\label{a:radialKM} A radial regularity condition: there is $\rho>0$ such that
\begin{align*}
 \P\bigl(\overline F_{X^{(j)}}(Z^{(j)})<c_ju,\ j=1,\ldots,d\bigr)
 \ \leq\ A^{\rho}\,\P\bigl(\overline F_{X^{(j)}}(Z^{(j)})<c_ju/A,\ j=1,\ldots,d\bigr),
\end{align*}
for $A\geq1$, $c\in[1/(16T),4]^d$ and $u>0$.

\item\label{a:secondKM} For weak convergence we assume the second-order condition
\begin{align*}
 r_n\,\sup_{x\in[a,b]^d}\bigl|R^X_{n/k}(x)-R^X(x)\bigr|\longrightarrow0
 \qquad\text{for every }0<a<b<\infty .
\end{align*}
\end{aparts}
\end{assumption}

\begin{remark}\label{rem:asmS3}
The quantity $kG_j(n/k)$ is the expected
number of observations of the $j$th coordinate
that exceed the level $u_{j,n}$; since the joint probability is at most each marginal
one, Assumption~\ref{a:effective} forces $kG_j(n/k)\to\infty$ for every $j$.
Like Assumption~\ref{a:radialRV}, Assumption~\ref{a:radialKM} is a statement about the distribution of $Z$ alone
with no dependence on $n$. It is understood on the standardized scale rather than on the
original one, and it is therefore invariant under increasing transformations of
all margins.
Finally, the standardization contributes no bias, so, unlike in
the multiplicative standardization case, Assumption~\ref{a:secondKM} is a condition on the tail copula alone.
\end{remark}

\medskip

\noindent\textit{The standardization.}  Nothing corresponding to
Assumption~\ref{ass:plugin} has to be assumed here; instead a detailed analysis of the Kaplan--Meier transformation is required.  By
Proposition~\ref{prop:KM}, applied on each coordinate at the level $16T$,
\begin{align}
 \sup_{x:\ \psi_{j,n}(x)\leq16T}
 \bigl|{\widehat\psi_{j,n}(x)\over\psi_{j,n}(x)}-1\bigr|
 =O_{\P}\bigl(\bigl\{kG_j(n/k)\bigr\}^{-1/2}\bigr),\qquad j=1,\ldots,d,
 \label{KMrate}
\end{align}
so that, with $\zeta_n=\max_j\{kG_j(n/k)\}^{-1/2}$ which is the rate of \eqref{freeplugin},
Assumption~\ref{a:plugin} is seen to hold with $\eta_n=a_n\zeta_n$ for \emph{every} deterministic
$a_n\to\infty$. Note that the freedom
in the choice of the sequence $a_n$ is not retained when having to show Assumption~\ref{a:plugincons}, which additionally requires $\eta_n\log(n/r_n^2)\to0$;
such a choice exists because $\zeta_n\log(n/k)\to0$, shown in the proof of
Corollary~\ref{cor:KMcons} below.  To establish Assumption~\ref{a:pluginnorm} one would further require
$r_n\zeta_n\log(n/k)\to0$, which holds when $\beta^\ast>\max_j\beta_j$ but not in general,
as the proof of Corollary~\ref{cor:KMnorm} shows.

\begin{corollary}\label{cor:KMcons}
Assume Assumptions~\ref{ass:supp} and~\ref{ass:KM}.1--\ref{ass:KM}.2.  Then
\begin{align*}
 \sup_{q\in K}|\widehat F_n^\circ(q)-F^\circ(q)|\overset{\P}\longrightarrow0 .
\end{align*}
\end{corollary}

\begin{corollary}\label{cor:KMnorm}
Assume the hypotheses of Corollary~\ref{cor:KMcons} together with
Assumption~\ref{a:secondKM}.  If $\beta^\ast>\max_j\beta_j$, in particular if the
coordinates of $C$ are independent and at least two censoring indices are positive, then
\begin{align*}
 r_n\bigl\{\widehat F_n^\circ-F^\circ\bigr\}\ \leadsto\ F^\circ\,\mathbb G
 \qquad\text{in }\ell^\infty(K),
\end{align*}
with $\mathbb G$ as in \eqref{Glimit}.  If instead the standardization is regular
in the sense of Remark~\ref{rem:regular}, the same holds with $\mathbb G$ replaced by
$\mathbb G+\mathbb M-\mathbb M(\boldsymbol 1)$.
\end{corollary}

\newpage
\section{Supplementary numerical studies}\label{supp:numerics}
This section provides additional numerical studies.
Section~\ref{supp:sim-km} analyzes the Kaplan--Meier marginal standardization under
unequal, non-regularly varying margins and the quality of the Gaussian approximation.
Sections~\ref{sec:addsim1} and~\ref{sec:addsim2} study the finite-sample contribution of the
estimated marginal standardization and the effective-sample-size-scaled error as a
function of dimension.  Section~\ref{sec:addsim3} departs from the prime case of
Section~\ref{sec:simulations} of the paper in four separate ways, changing one feature
at a time and holding the rest fixed.  Throughout, the data-generating mechanisms, the direction
grids and the number $B=500$ of replications are same as in the main paper.

\subsection{Arbitrary margins and Gaussian approximation}\label{supp:sim-km}
The goal of this study is to analyze the Kaplan--Meier marginal
standardization for unequal, non-regularly varying margins, and to study
the quality of the
plug-in Gaussian approximation for finite-sample inference. The two event margins are Weibull with shapes
$(0.75,1.5)$ and scales $(1,1.4)$, and the
event copula has $\theta=2$.  The censoring margin has the same shape as the
corresponding event margin and scale
$\sigma_{C,j}=\sigma_j\{c_j/(1-c_j)\}^{-1/\alpha_j}$; this guarantees that the tail censoring fraction
is precisely $c_j$.  We consider
$(c_1,c_2)=(0.10,0.10),(0.25,0.25),(0.10,0.40)$, with $n=8000$, $k=800$, and 25
equally spaced values of $a\in[-\log(2.5),\log(2.5)]$ with
$(s,t)=\bigl(e^{a_+},e^{(-a)_+}\bigr)$.  The marginal Kaplan--Meier standardization
\eqref{KMpsi} is required, so the target is
$F^\circ(s,t)=R_2(s^{-1},t^{-1})/R_2(1,1).$

Pointwise confidence intervals are constructed from \eqref{greenwoodCI}, with $\alpha=0.05$ and the
plug-in variance estimator of Proposition~\ref{prop:variance}. The results are provided in Figure~\ref{fig:sim-directional} and show that the Kaplan--Meier standardization is able to recover
the target throughout the three censoring regimes.  At $s/t =0.54$, for
example, the true value is $0.691$ and the mean directional estimates are
$0.681$, $0.682$, and $0.682$, whereas the naive means are $0.635$, $0.557$, and $0.453$.
Under asymmetric censoring the naive error changes sharply with direction.  Pointwise
coverage over the grid ranges from $0.940$ to $0.996$.  Averaged over that
grid, the variance estimator \eqref{greenwoodvariance} reproduces the Monte Carlo spread of
$\log\widehat F_n^\circ$ to within $1.2\%$ for the estimator built from the known
Weibull survival functions. However, the Kaplan--Meier standardization makes the
estimator less variable, by $21$, $14$, and $13$ percent
as the censoring goes from light to moderate to asymmetric.  By
Theorem~\ref{thm:normality} the two are equivalent in the limit, so the interval is wide in finite samples and in this setting.

\begin{figure}[]
 \centering
 \includegraphics[width=\textwidth]{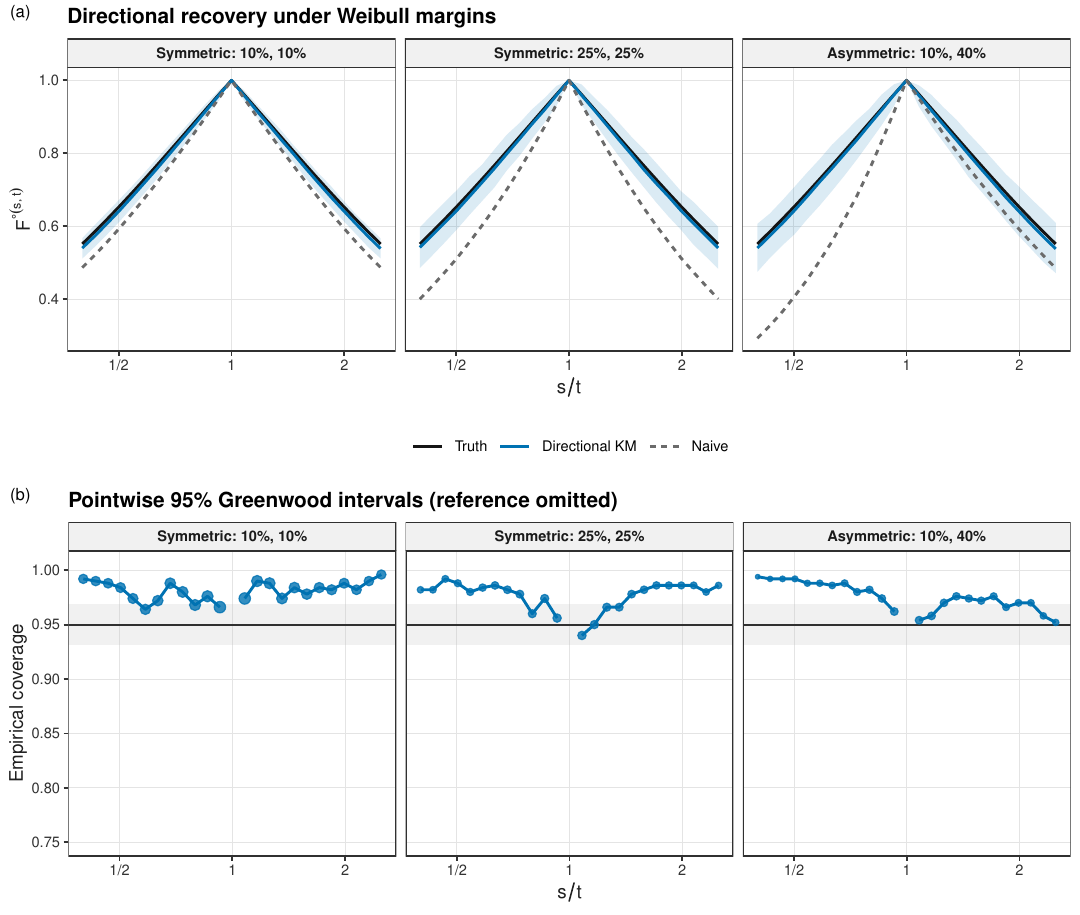}
 \caption{Weibull margins and Kaplan--Meier standardization.  In panel~(a), we show the truth,
 the mean of the directional estimator, the pointwise empirical
 10th--90th percentile band of its estimates, and the naive mean.  In panel~(b),
 the empirical coverage of the Gaussian 95\% intervals
 \eqref{greenwoodCI} is plotted. Here, point size is proportional to the mean directional joint count.
 The horizontal band is $0.95\pm1.96\{0.95(0.05)/500\}^{1/2}$, corresponding to the pointwise binomial
 Monte Carlo range under exact coverage.}
 \label{fig:sim-directional}
\end{figure}

\subsection{Plug-in estimation and oracle standardization}\label{sec:addsim1}

The goal of this study is to study in more detail how the estimated marginal standardization affects estimation.  Let  the oracle estimator $\widehat F_{b,\mathrm{or}}^\circ$
be the directional product limit computed with the known Weibull survival functions,
every other step unchanged, and let $\widehat r_b$ be the rate estimator of
Proposition~\ref{prop:rateest} in replication $b$.  We consider
\begin{align*}
 D_b=\widehat r_b\sup_{q\in{\cal Q}}
 \bigl|\widehat F_b^\circ(q)-\widehat F_{b,\mathrm{or}}^\circ(q)\bigr|.
\end{align*}
This is the sup-norm separation between the plug-in and oracle estimators on the
$25$-point direction grid, blow up by the joint asymptotic scale.  The
results are provided in panel~(a) of Figure~\ref{fig:sim-plugin} and show medians
$0.63$, $0.66$, and $0.69$ in the light, moderate, and asymmetric settings, with 90th
percentiles $0.95$, $1.00$, and $1.10$.  Hence, at $n=8000$ and $k=800$,
marginal estimation contributes a visible finite-sample error, and the error does not
increase under asymmetric censoring.

However, under Assumption~\ref{a:pluginnorm} this error vanishes asymptotically, making the
plug-in and oracle estimators differ by only $o_{\P}(r_n^{-1})$, and in particular $D_b$ should vanish as
the sample grows.  The theoretical rate at which it does is worth recalling. Combining
\eqref{KMrate} with Assumption~\ref{a:rate} gives
\begin{align}
 r_n\zeta_n=\Bigl\{R^X(\boldsymbol 1)\,{c_n(\boldsymbol 1)\over\min_jG_j(n/k)}\Bigr\}^{1/2}
 \ =\ \bigl\{R^X(\boldsymbol 1)\textstyle\prod_{j\neq j^\ast}G_j(n/k)\bigr\}^{1/2},
 \qquad j^\ast=\operatorname{arg\,min}_jG_j(n/k),
 \label{plugingrowth}
\end{align}
the second expression holding because the censoring coordinates are independent in
this study, and the error depends on $n$ only through $n/k$.  Take for example $k=\lceil n^{3/4}\rceil$, gives returns $k=846$ at $n=8000$ which aligns roughly to the setting of Section~\ref{supp:sim-km}.  The results fir growing $n$ are provided in panel~(b) and
show the median error falling from $0.699$, $0.732$, and $0.752$ at $n=2000$ to
$0.547$, $0.494$, and $0.558$ at $n=128,000$, with the 90th percentiles falling from
$1.02$, $1.12$, and $1.16$ to $0.85$, $0.74$, and $0.95$.  The decline is monotone in
every regime, but surprisingly rather slow. So although the standardized and oracle directional estimators are aymptotically equivalent, reaching their equivalence may be slow in setttings as the one considered here.

\begin{figure}[]
 \centering
 \includegraphics[width=0.84\textwidth]{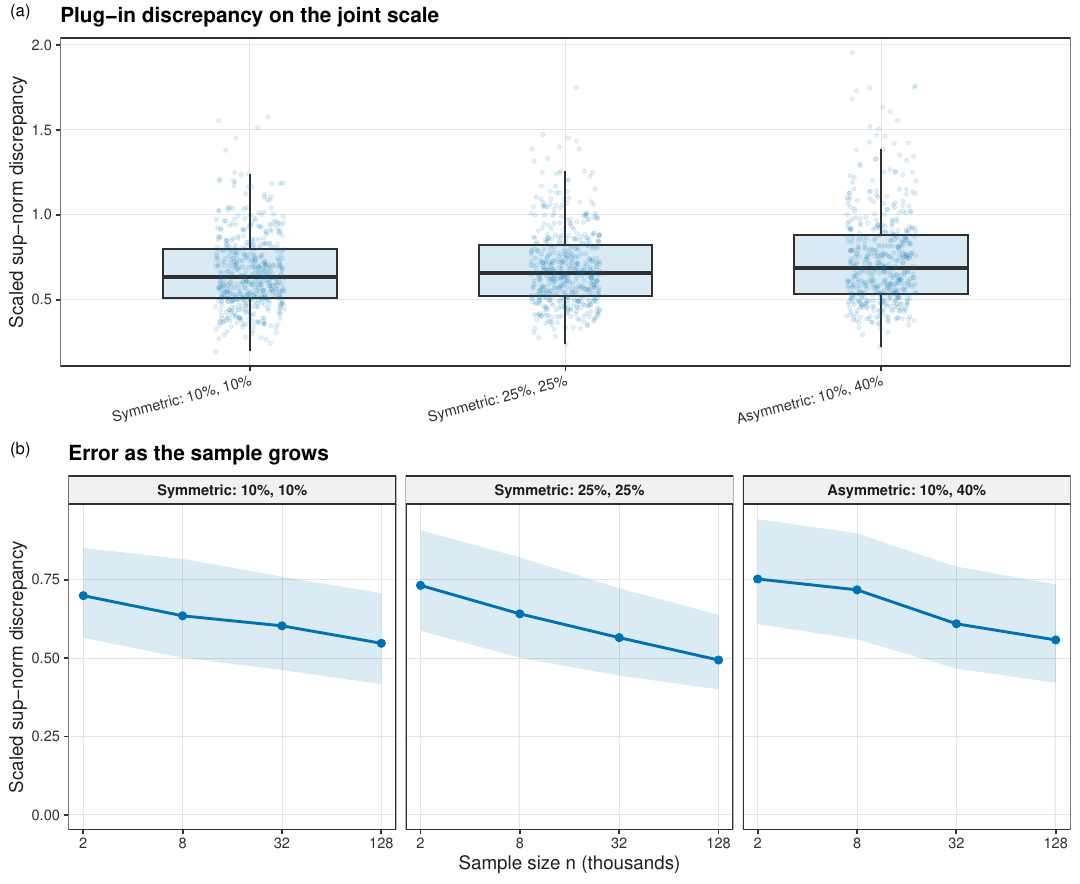}
 \caption{Finite-sample contribution of the Kaplan--Meier marginal standardization in
 Section~\ref{supp:sim-km}.  The oracle uses the known Weibull survival functions, and $D_b$ is the
 sup-norm plug-in versus oracle error on the same 25-point grid as in Figure~\ref{fig:sim-directional},
 multiplied by the square root of the realized joint count.  In panel~(a), at the
 $n=8000$ and $k=800$ of Section~\ref{supp:sim-km}, boxes give medians and interquartile ranges over 500
 replications, whiskers use the $1.5$ interquartile-range rule, and points
 show all replications.  In panel~(b), the median of $D_b$ is plotted against $n$ at
 $k=\lceil n^{3/4}\rceil$, the band being the interquartile range over 500 replications
 at each sample size.}
 \label{fig:sim-plugin}
\end{figure}

\subsection{Inference and scaled error across dimension}\label{sec:addsim2}
We now study whether the effective rate $r_n$ summarizes the effect of dimension
for inference and for error. The first was studied for $d=2$ in
Section~\ref{supp:sim-km} and the second only in aggregate form in the second study of the paper.

For the inference part, we repeat the matched joint count setting at $d=2,\ldots,10$ with the
intervals \eqref{greenwoodCI} at $\alpha=0.05$, and provide the coverage at each
direction of the $21$-point path grid.  Matching holds the realized joint count at the
reference between $100.5$ and $101.9$ across the nine dimensions.  The results
are provided in Figure~\ref{fig:sim-coverage}.  Averaged over the grid, coverage lies
between $0.930$ and $0.950$ with no trend in $d$, so matching the joint count stabilizes
the inferential problem as well.  Departure from nominal is observed along
direction rather than dimension.  It is concentrated at the smallest $\lambda$,
where coverage falls to between $0.816$ and $0.898$, and it has gone by
$\lambda=0.25$, beyond which coverage averages $0.9459$ in every dimension.  The reason
is that $\mathbb G(q)$ vanishes as
$q\to\boldsymbol 1$, so the interval collapses there, its mean width being $0.075$ at
$\lambda=0.05$ against $0.195$ at $\lambda=1$, and the Gaussian approximation is hardest exactly where 
differences are smallest.  The three paths separate this clearly;  for instance, at $\lambda=0.05$ and
$d\geq4$, the path holding $d-1$ coordinates at the reference and departing from it
in one coordinate only, has mean coverage $0.853$ against $0.921$ and $0.925$ for the
other two.

\begin{figure}[]
 \centering
 \includegraphics[width=\textwidth]{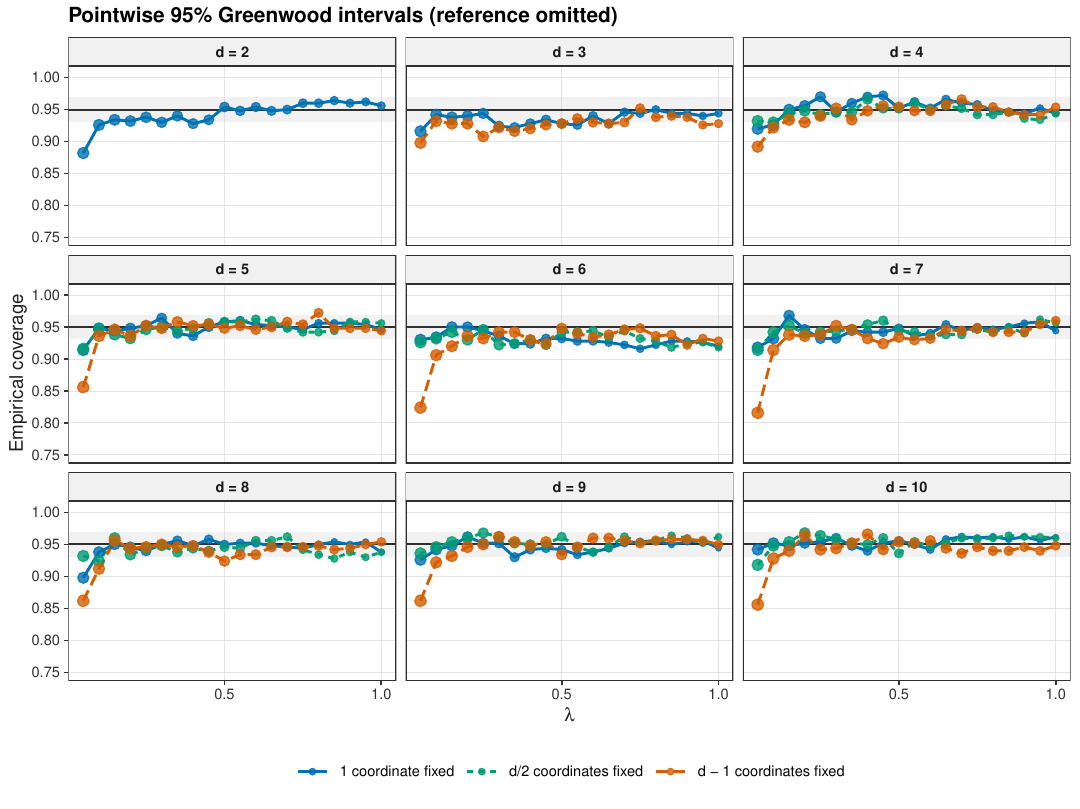}
 \caption{Pointwise empirical coverage of the intervals \eqref{greenwoodCI} along the
 $21$-point path grid, under the matched joint count setting of the second study of the paper with the
 realized count at the reference held near $100$ in every dimension.  Three paths are shown, which correspond to holding $1$, $\lfloor d/2\rfloor$ and
 $d-1$ coordinates at the reference, and coincide when $d=2$; point size is proportional
 to the mean directional joint count, and the reference $\lambda=0$ is omitted.  The horizontal band is
 $0.95\pm1.96\{0.95(0.05)/500\}^{1/2}$, the pointwise binomial range under
 exact coverage.}
 \label{fig:sim-coverage}
\end{figure}

For the error, we revisit the deterioration analysis as a function of $d$ found in the second study of the paper by providing the unaggregated version.  Let
$\widehat r_{b,d}$ be the rate
estimator of Proposition~\ref{prop:rateest} for replication $b$ and dimension $d$, and
define
\begin{align*}
 T_{b,d}=\widehat r_{b,d}\,E_b({\cal Q}_d),\qquad d=2,\ldots,10.
\end{align*}
Because the first $d$ coordinates of the same $10$-dimensional sample are used at every
dimension, joining $T_{b,d}$ across $d$ gives a within-replication trajectory.  The results are provided in
Figure~\ref{fig:sim-scaled}, which contains all $500$ trajectories and whose
dimension-wise means range from $0.412$ to $0.454$ under common $k$ and from $0.418$ to
$0.453$ under matched joint counts.  Location and also spread are stable in $d$,
which supports that the effective rate $r_n$, which is consistently estimated by $\widehat r_n$, is
the main driver of the behavior in Figure~\ref{fig:sim-dimension} of the main
paper, over the
range of dimensions studied. 

In conclusion, once the joint count is held fixed, dimension keeps both the error
and the coverage of the intervals at similar levels; what remains of the
departure from nominal level appears to be a property of the direction, and not intrinsically of $d$.

\begin{figure}[H]
 \centering
 \includegraphics[width=0.80\textwidth]{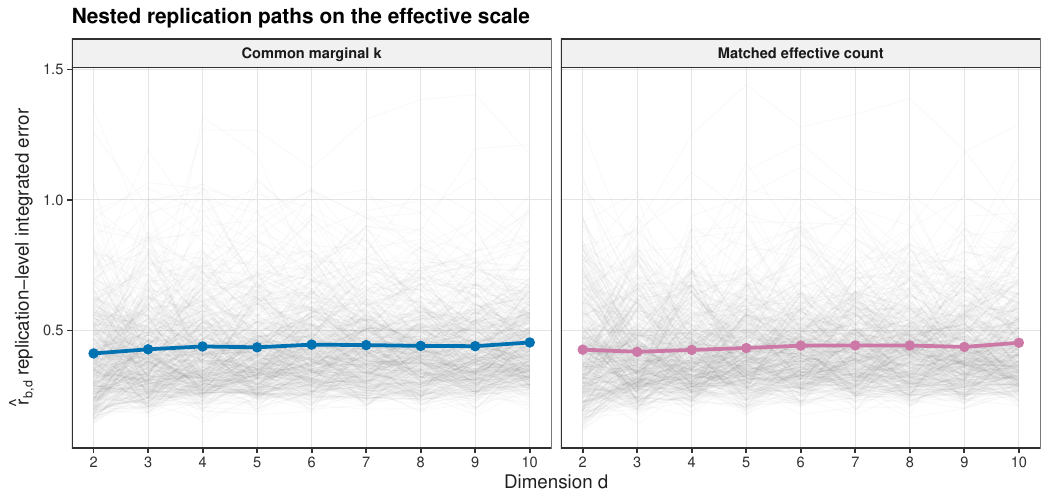}
 \caption{Replication-wise (unaggregated) scaled errors $T_{b,d}$ from the second study of the main paper for every
 $d=2,\ldots,10$, evaluated on the 21-point path grid.  The left and right panels use,
 respectively, the common-$k$ and the matched joint count settings.  Each thin gray line provides errors from one nested replication, and the thick line is the
 dimension-wise mean over all 500 replications.}
 \label{fig:sim-scaled}
\end{figure}

\subsection{Departures from the prime case}\label{sec:addsim3}
We explore how the conclusions of Section~\ref{sec:simulations} change when there are departures from the main hypotheses. Throughout, $d=2$ unless
stated otherwise, $n=12000$, the event margins are the Burr margins of the first study of the paper with
$(\gamma_1,\gamma_2)=(0.2,0.4)$, the multiplicative standardization is
used with the thresholds \eqref{simthreshold}, and the direction grid, the error
criterion \eqref{irmse} and the $B=500$ replications are those of the first study of the paper.  The results
are collected in Figure~\ref{fig:sim-robustness}, and Table~\ref{tab:sim-robustness} provides additional details. Namely, it provides the integrated RMSE, the
realized joint count, the scaled error $\widehat r_n\times\operatorname{IRMSE}$, and the
coverage of the intervals \eqref{greenwoodCI}.  In~(a) and~(c) the scaled error and the coverage are both flat, so the
departure acts mainly on $r_n$; in~(b) and~(d) they move together, the
scaled error rising from $0.446$ to $0.923$ and from $0.511$ to $0.906$ while coverage
falls from $0.944$ to $0.859$ and from $0.950$ to $0.892$. Below we provide more details for each case.

\medskip

\noindent\textit{(a) Asymptotically dependent censoring coordinates.}  This case is
covered by Assumption~\ref{a:RC}; what changes is that $\beta^\ast=\max_j\beta_j$.  The
multiplicative standardization used here still gives $\mathbb M\equiv0$ by
Corollary~\ref{cor:RVnorm}, whatever the dependence among the coordinates of $C$; it is
under the Kaplan--Meier standardization that the limit has the additional term
$\mathbb M-\mathbb M(\boldsymbol 1)$ of Remark~\ref{rem:regular}.  We consider the censoring vector a Gumbel copula
\eqref{gumbelcopula} with parameter $\theta_C\in\{1,1.25,1.5,2,2.5,3,4\}$, $\theta_C=1$ being the
prime case, keeping the censoring margins and hence the marginal censoring fractions
$c\in\{0.25,0.40\}$ unchanged.  The target $F^\circ$ is unchanged as well, since it
depends only on $X$.  Raw error falls as $\theta_C$ grows, from $0.075$ to $0.046$ at
$c=0.25$ and from $0.082$ to $0.044$ at $c=0.40$, but this is an effect on the
effective sample size, which rises from $45.8$ to $110.4$ and from $50.6$ to $178.8$
respectively, because dependent censoring makes the two coordinates tend to be censored
together, so that more observations survive in both at once.  The scaled
error $\widehat r_n\times\operatorname{IRMSE}$, plotted in panel~(a), is flat in
$\theta_C$,
lying in $[0.474,0.508]$ at $c=0.25$ and in $[0.549,0.586]$ at $c=0.40$.  Coverage of
the intervals \eqref{greenwoodCI} stays between $0.933$ and $0.961$ throughout.  For this setting,
violating independence in finite samples changes drastically the rate $r_n$ (which is also explicit in the asymptotic theory) but accuracy is maintained.

\medskip

\noindent\textit{(b) Heavier censoring.}  In the main text we compared the directional estimator with
the bivariate Dabrowska-type estimator of \citet{BladtGoegebeurGuillou2026bS} at
$c\leq0.40$ and found that they performed similarly.  We extend the comparison to
$c\in\{0.10,0.25,0.40,0.50,0.60\}$, keeping the nominal count at $50$.  The
integrated RMSE of the directional estimator is $0.074$, $0.077$, $0.085$, $0.096$ and
$0.118$; that of the Dabrowska-type estimator is $0.073$, $0.076$, $0.089$, $0.108$ and
$0.148$.  The ratios are given in panel~(b) and show the two to be virtually indistinguishable up
to $c=0.25$, with the directional estimator ahead by $5\%$, $12\%$, and $26\%$ at
$c=0.40$, $0.50$, and $0.60$.  Heavy censoring is nonetheless the one departure that the effective rate $r_n$ does not account for, in the sense that the scaled error increases from $0.446$ to
$0.923$ and the coverage deteriorates from $0.944$ to $0.859$ over the same range, so the
gain is only relative.  The mathematical explanation is that the
Dabrowska-type estimator needs a bivariate risk set to be nonempty over a grid of the
quadrant, and heavy censoring thins out parts of that grid; on the other hand, the directional estimator
only needs a one-dimensional risk set along a ray, and that risk set is thinned out through the single variable $D_{i,n}(q)$.  For reference, the naive
estimator is worse by factors $1.03$, $1.68$, $2.35$, $2.60$, and $2.58$.

\medskip

\noindent\textit{(c) Asymptotic independence of the $X$ vector.}  Assumption~\ref{a:RX} with
$R^X(\boldsymbol 1)>0$ excludes asymptotic independence.  We now replace the Gumbel event
copula by a Gaussian one with correlation $\rho\in\{0.75,0.65,0.55,0.5\}$, for which
$R^X\equiv0$, and weaken the Gumbel copula to $\theta\in\{1.75,1.5,1.25\}$ for
comparison; $c=0.25$ and $k=600$ in all eight settings.  When $R^X\equiv0$ the limit $F^\circ$ does
not exist, but the finite-sample ratio $p_n(q)/p_n(\boldsymbol 1)$ does, and it is
computable in closed form for Burr margins and either copula.  Measured against it, the mean errors over the non-reference grid
lie between $-0.0017$ and $0.0054$ across the eight settings, so the estimator remains
essentially unbiased for its modified target under asymptotic independence.  What naturally
collapses is the expected joint count, falling from
$42.8$ at $\theta=2$ to $17.6$ at $\rho=0.5$, and the integrated RMSE rising
correspondingly from $0.077$ to $0.127$.  For reference, in the four Gumbel settings the finite-sample
ratio and the limit $F^\circ$ differ over the grid by at most $0.006$ to $0.023$, and
the errors measured against either agree to three decimal places.  Thus, we see that asymptotic independence changes the rate, and hence effective sample size, but does not break down estimation.  Indeed, Panel~(c) shows integrated
errors falling in order of the joint count,
and only slightly in order of the dependence class; the scaled error stays between
$0.501$ and $0.543$ and the coverage between $0.941$ and $0.949$.

\medskip

\noindent\textit{(d) Dimension at heavier censoring.}  In the second study of the paper we used $5\%$ marginal tail
censoring.  We repeat its matched-joint-count setting now at $c=0.25$ for
$d\in\{2,4,6,8,10\}$, with $n=10,000$, standard Pareto margins of index $1/2$,
copula parameter $\theta=3$ for $X$, the path grid of the second study of the paper and $k$ chosen so that the
expected joint count is $100$ in every dimension.  The results are given in panel~(d), where the directional integrated RMSE rises
from $0.051$ at $d=2$ to $0.084$ at $d=10$ and the naive one from $0.128$ to $0.316$,
so the ratio between them widens from $2.5$ to $3.8$.  The realized joint counts lie
between $102$ and $116$, and the scaled error $\widehat r_n\times
\operatorname{IRMSE}$ rises from $0.51$ to $0.91$.  Unlike in
Section~\ref{sec:addsim2}, where censoring was light and the scaled error was flat,
matching the joint count at the reference direction no longer removes the
dimensional effect totally.  The reason is that the count is matched at $q=\boldsymbol 1$ only,
while the error is integrated over directions at which several coordinate thresholds have been
increased.  By \eqref{Hn} and Assumption~\ref{a:Cbeta} the risk at such a $q$ carries the extra
factor $\prod_jq_j^{-\beta_j}$ with $\beta_j=c/(1-c)$, which is $1/3$ here and $1/19$
at the $5\%$ censoring of the second study of the paper, so each coordinate increased to $q_j=2$ thins out the
directional risk set by $2^{-1/3}=0.79$ in the first case and by $2^{-1/19}=0.96$ in
the second, which is why the effect is more visible here.  The coverage
also deteriorates, from $0.950$ at $d=2$ to $0.892$ at $d=10$.  At $c=0.40$ the same
study should not be run.  Matching the joint count holds
$\prod_jG_j(n/k)=(k/n)^{d\beta}$ against $k$, and $\beta$ increases to $2/3$, so the
required tail fraction $k/n$ reaches $0.32$ at $d=4$ and $0.59$ at $d=10$, which is no longer
a tail. A larger sample does not fix this effect as $d$ grows, since censoring increases geometrically in the dimension, restricting the available tail information; it does, however, fix it for a given $d$.

\begin{figure}[]
 \centering
 \includegraphics[width=\textwidth]{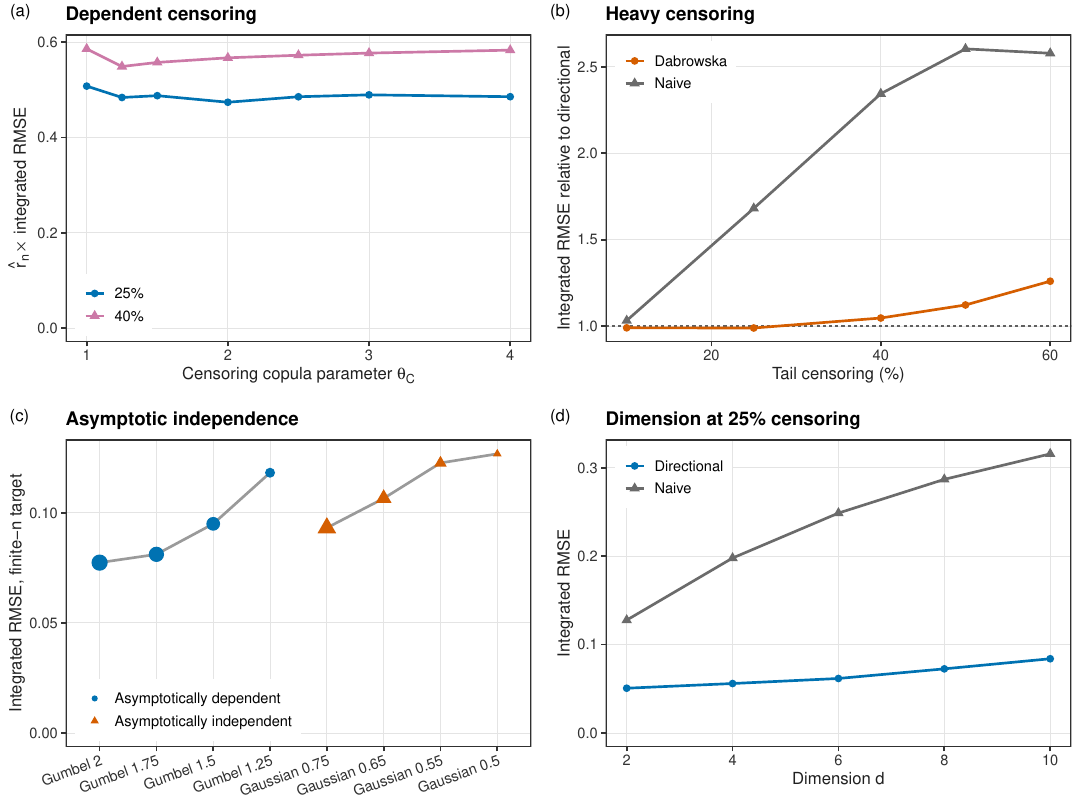}
 \caption{Departures from the prime case.  In panel~(a), the scaled error
 $\widehat r_n\times\operatorname{IRMSE}$ of the directional estimator is plotted as the
 censoring coordinates are made dependent through a Gumbel copula with parameter
 $\theta_C$, at two marginal censoring levels.  In panel~(b), the integrated RMSE of
 the Dabrowska-type and naive estimators are shown relative to the directional one as the tail
 censoring fraction increases.  In panel~(c), the integrated RMSE is plotted against the finite-$n$
 ratio for the eight dependence settings, ordered by decreasing tail dependence, with
 point size proportional to the realized joint count.  In panel~(d), the integrated RMSE of the directional
 and naive estimators are plotted against dimension, at $25\%$ marginal tail censoring and matched
 joint counts.}
 \label{fig:sim-robustness}
\end{figure}

\begin{table}[H]
 \centering
 \small
 \begin{tabular}{lrrrr}
 \hline
 & IRMSE & $\widehat r_n^2$ & $\widehat r_n\times$IRMSE & Coverage\\
 \hline
 \multicolumn{5}{l}{\textit{(a) Dependent censoring coordinates}}\\
 \quad $c=0.25$, $\theta_C=1$ & $0.0751$ & $45.8$ & $0.508$ & $0.945$\\
 \quad $c=0.25$, $\theta_C=1.25$ & $0.0607$ & $63.7$ & $0.484$ & $0.953$\\
 \quad $c=0.25$, $\theta_C=1.50$ & $0.0562$ & $75.3$ & $0.488$ & $0.953$\\
 \quad $c=0.25$, $\theta_C=2$ & $0.0499$ & $90.4$ & $0.474$ & $0.959$\\
 \quad $c=0.25$, $\theta_C=2.50$ & $0.0488$ & $99.0$ & $0.486$ & $0.956$\\
 \quad $c=0.25$, $\theta_C=3$ & $0.0480$ & $103.9$ & $0.490$ & $0.953$\\
 \quad $c=0.25$, $\theta_C=4$ & $0.0462$ & $110.4$ & $0.486$ & $0.954$\\
 \quad $c=0.40$, $\theta_C=1$ & $0.0824$ & $50.6$ & $0.586$ & $0.933$\\
 \quad $c=0.40$, $\theta_C=1.25$ & $0.0583$ & $88.7$ & $0.549$ & $0.958$\\
 \quad $c=0.40$, $\theta_C=1.50$ & $0.0527$ & $112.0$ & $0.558$ & $0.960$\\
 \quad $c=0.40$, $\theta_C=2$ & $0.0479$ & $140.3$ & $0.567$ & $0.950$\\
 \quad $c=0.40$, $\theta_C=2.50$ & $0.0458$ & $156.6$ & $0.573$ & $0.959$\\
 \quad $c=0.40$, $\theta_C=3$ & $0.0447$ & $166.8$ & $0.577$ & $0.954$\\
 \quad $c=0.40$, $\theta_C=4$ & $0.0436$ & $178.8$ & $0.584$ & $0.961$\\
 \multicolumn{5}{l}{\textit{(b) Heavier censoring}}\\
 \quad $c=0.10$ & $0.0742$ & $36.2$ & $0.446$ & $0.944$\\
 \quad $c=0.25$ & $0.0765$ & $44.9$ & $0.513$ & $0.954$\\
 \quad $c=0.40$ & $0.0847$ & $50.5$ & $0.602$ & $0.928$\\
 \quad $c=0.50$ & $0.0963$ & $54.2$ & $0.709$ & $0.908$\\
 \quad $c=0.60$ & $0.1179$ & $61.4$ & $0.923$ & $0.859$\\
 \multicolumn{5}{l}{\textit{(c) Asymptotic independence}}\\
 \quad Gumbel $2$ & $0.0774$ & $42.8$ & $0.506$ & $0.944$\\
 \quad Gumbel $1.75$ & $0.0811$ & $38.1$ & $0.501$ & $0.949$\\
 \quad Gumbel $1.50$ & $0.0950$ & $31.5$ & $0.533$ & $0.943$\\
 \quad Gumbel $1.25$ & $0.1182$ & $20.9$ & $0.541$ & $0.945$\\
 \quad Gaussian $0.75$ & $0.0932$ & $32.1$ & $0.528$ & $0.941$\\
 \quad Gaussian $0.65$ & $0.1065$ & $25.2$ & $0.535$ & $0.942$\\
 \quad Gaussian $0.55$ & $0.1226$ & $19.6$ & $0.543$ & $0.941$\\
 \quad Gaussian $0.50$ & $0.1268$ & $17.6$ & $0.532$ & $0.945$\\
 \multicolumn{5}{l}{\textit{(d) Dimension at $25\%$ censoring}}\\
 \quad $d=2$ & $0.0507$ & $101.7$ & $0.511$ & $0.950$\\
 \quad $d=4$ & $0.0560$ & $105.0$ & $0.574$ & $0.945$\\
 \quad $d=6$ & $0.0617$ & $108.9$ & $0.644$ & $0.942$\\
 \quad $d=8$ & $0.0726$ & $112.6$ & $0.770$ & $0.928$\\
 \quad $d=10$ & $0.0840$ & $116.3$ & $0.906$ & $0.892$\\
 \hline
 \end{tabular}
 \caption{Additional diagnostics for the four departures.  The integrated RMSE is
 measured against $F^\circ$, except in~(c) where no limit exists and it is measured
 against the finite-$n$ ratio $p_n(q)/p_n(\boldsymbol 1)$.  The joint count is the
 realized $\widehat r_n^2$ at the reference, and the coverage is that of
 \eqref{greenwoodCI} at $\alpha=0.05$ over the non-reference grid, each averaged over
 the $500$ replications.}
 \label{tab:sim-robustness}
\end{table}

\newpage
\section{Proofs of the main results}\label{sec:proofs}
\subsection{Preliminaries}
\emph{On the dimension.} The statements in the main paper are made in an arbitrary fixed
dimension $d$, while the proofs below are written for $d=2$, which keeps proofs more readable.  The substitutions that make them general are mostly
routine, apart from the points provided below.

First, put $I(\delta)=\{j:\delta^{(j)}=1\}$ and abbreviate
$A_j=A^{(j)}_{i,n}(q)$.  On every fixed indicator vector $\delta\neq\boldsymbol0$, the exact
failure-band identity is given by
\begin{align}
 &\{\delta_i=\delta,\ \Delta_{i,n}(q)=1,\ a<W_{i,n}(q)\leq b\}\nonumber\\
 &\quad=\{\delta_i=\delta,\ A_h>a\ \text{for all }h\}
 \cap\bigcup_{j\in I(\delta)}
 \{A_j\leq b,\ A_j\leq A_l\ \text{for all }l\notin I(\delta)\}.
 \label{failured}
\end{align}
Thus all directional ties are retained, and comparisons are needed only between
opposite indicators.  Each set in the finite union is cut out, on a fixed indicator, by finitely
many thresholds in the coordinates $\psi_{h,n}(Z^{(h)})$ and the ratios
$\psi_{j,n}(Z^{(j)})/\psi_{l,n}(Z^{(l)})$.  For fixed $d$, finite intersections and
finite unions of these one-parameter threshold classes have a VC index depending only
on $d$; hence the argument preceding \eqref{VC} and the maximal bound
\eqref{maximalVC} continue to hold with constants depending only on $d$.  Moreover,
because the ratio comparisons in \eqref{failured} are only between opposite indicators,
their limit boundaries compare an event coordinate with a censoring coordinate and
are null by the independence of $X$ and $C$ and continuity of the margins.

Second, the common envelope becomes
$E_\tau=\{\psi_{j,n}(Z^{(j)})\geq\tau/T^{d-1},\ j\leq d\}$, since
$\psi_{j,n}(Z^{(j)})>\tau$ together with $d-1$ ratio bounds by $T$ forces every other
coordinate past $\tau/T^{d-1}$, and \eqref{envelopemass} then holds with
$F^\circ G^\circ$ evaluated at $(\tau/T^{d-1})\boldsymbol 1$.

Third, assign, to each failure at a tied value, the smallest uncensored
coordinate attaining the minimum.  Marginal continuity makes the transformed values
distinct across observations within each coordinate, so $d_\ell(q)\leq d$; replacing
$2$ by the fixed number $d$ in \eqref{tietwo}--\eqref{KMNA} leaves the
$O_{\P}(r_n^{-2})$ remainder unchanged.

Fourth, in the corollaries of Section~\ref{supp:mda} the criterion
$\beta^\ast>\max_j\beta_j$ for the plug-in to be negligible is in terms of the
joint censoring index and therefore already dimension free.

\medskip

The following lemma bounds the $\nu_n$-mass of mixed-indicator events which are almost tied.

\begin{lemma}\label{lem:nearties}
Let $\eta\in(0,1]$, $q\in[1/4,16T]^2$ and $0<v\leq2$.  Then
\begin{align}
 \nu_n\Bigl\{\bigl|\log{A^{(1)}_{1,n}(q)/A^{(2)}_{1,n}(q)}\bigr|\leq\eta,\
 \delta^{(1)}_1\neq\delta^{(2)}_1,\ v\leq W_{1,n}(q)\leq2\Bigr\}
 \ \leq\ C\,\eta\,H_n(v;q),
 \label{bandmass}
\end{align}
with $C=2\rho e^{\rho}$ and $\rho$ as in Assumption~\ref{a:radialmodulus}.
\end{lemma}

\begin{proof}
Write $B$ for the set on the left of \eqref{bandmass}, and suppose $\delta^{(1)}_1=1$,
$\delta^{(2)}_1=0$; the opposite mark is symmetric.  Then
$A^{(1)}_{1,n}(q)=\psi_{1,n}(X_1^{(1)})/q_1\geq T_{1,n}(q)$ and
$A^{(2)}_{1,n}(q)=\psi_{2,n}(C_1^{(2)})/q_2\geq D_{1,n}(q)$ by \eqref{TD}.  If
$A^{(1)}_{1,n}(q)\leq A^{(2)}_{1,n}(q)$ then $W_{1,n}(q)=A^{(1)}_{1,n}(q)$ and, by
\eqref{reduction}, $T_{1,n}(q)=W_{1,n}(q)$ and $T_{1,n}(q)\leq D_{1,n}(q)\leq
A^{(2)}_{1,n}(q)$; otherwise $D_{1,n}(q)=W_{1,n}(q)=A^{(2)}_{1,n}(q)$ and
$D_{1,n}(q)\leq T_{1,n}(q)\leq A^{(1)}_{1,n}(q)$.  In both cases $T_{1,n}(q)$ and
$D_{1,n}(q)$ lie between the two coordinates $A^{(1)}_{1,n}(q)$ and $A^{(2)}_{1,n}(q)$, so on
$B$,
\begin{align*}
 \bigl|\log{T_{1,n}(q)/D_{1,n}(q)}\bigr|\leq\eta,
 \qquad W_{1,n}(q)\leq T_{1,n}(q)\vee D_{1,n}(q)\leq e^{\eta}W_{1,n}(q).
\end{align*}
Next, $p_n(\cdot\,q)$ and $c_n(\cdot\,q)$ are nonincreasing, so the two ratios forming the
left-hand side of the inequality in Assumption~\ref{a:radialmodulus} are each at least one; as their product is at most
$A^{\rho}$, each is at most $A^{\rho}$, and in particular $p_n(wq)\leq A^{\rho}p_n(Awq)$ on the
same range.  By \eqref{TD}, $\P\{T_{1,n}(q)\geq w\}=p_n(wq)$ and
$\P\{D_{1,n}(q)\geq w\}=c_n(wq)$, and $T_{1,n}(q)$ and $D_{1,n}(q)$ are independent.  On $B$
one has $v\leq D_{1,n}(q)\leq2e^{\eta}$, so conditioning on $D_{1,n}(q)$ and using
$1-e^{-x}\leq x$,
\begin{align*}
 \P(B)&\leq\int_{[v,2e^{\eta}]}\bigl\{p_n(we^{-\eta}q)-p_n(we^{\eta}q)\bigr\}
 \,\P\bigl\{D_{1,n}(q)\in dw\bigr\}\\
 &\leq2\rho\eta\int_{[v,\infty)}p_n(we^{-\eta}q)\,\P\bigl\{D_{1,n}(q)\in dw\bigr\}
 \ \leq\ 2\rho\eta\,p_n(ve^{-\eta}q)\,c_n(vq)\\
 &\leq2\rho e^{\rho}\eta\,p_n(vq)\,c_n(vq),
\end{align*}
every argument lying in $(0,16]$ because $\eta\leq1$ and
$w\leq2e^{\eta}$.  Divide by $\P\{W_{1,n}(q)\geq v\}=p_n(vq)c_n(vq)$ and multiply by
$n/r_n^2$.
\end{proof}

\medskip

We also use the scaling identity
\begin{align}
 W_{i,n}(cq)={W_{i,n}(q)\over c},\qquad \Delta_{i,n}(cq)=\Delta_{i,n}(q),
 \qquad c>0,
 \label{rayscaling}
\end{align}
immediate from \eqref{TD} and \eqref{reduction}: moving along a ray in $q$ merely
rescales the reduced sample. 

\subsection{Proof of Theorem~\ref*{thm:consistency}}
\emph{Plan.}  Steps~1--4 treat the oracle problem, in which the standardizations are
known.  Step~1 identifies the population risk and failure measures of the reduced
observations and their limits.  Step~2 proves uniform convergence of their empirical
counterparts on $v\in[\tau,1]$, a region bounded away from the origin, where the risk
function is bounded.  Step~3 extends this to $v\in(0,\tau]$, where the risk function is
unbounded, by cutting the interval into layers on which it has a fixed order of
magnitude.  Step~4 converts the resulting hazard convergence into convergence of the
product limit.  Step~5 then passes from the oracle to the estimated standardizations,
appealing to the results in Appendix~\ref{sec:appendix}.

\textit{Step 1: population quantities.}
We first work with the known standardizations $\psi_{j,n}$.  A tilde distinguishes
the resulting empirical quantities from those in \eqref{KM}.  Independence in
\eqref{TD} turns the definitions \eqref{HH1} into
\begin{align}
 H_n(v;q)&={p_n(vs,vt)\over p_n(1,1)}
             {c_n(vs,vt)\over c_n(1,1)},
 \label{Hn}\\
 H_{1,n}(dv;q)&={c_n(vs,vt)\over c_n(1,1)}
                  \bigl\{-d{p_n(vs,vt)\over p_n(1,1)}\bigr\}.
 \label{H1n}
\end{align}
Indeed, independence of $T_{1,n}(q)$ and $D_{1,n}(q)$ gives, as an identity between
measures on $(0,\infty)$,
\begin{align*}
 &\P\bigl(W_{1,n}(q)\in dv,\Delta_{1,n}(q)=1\bigr)\\
 &\hspace{2cm}=\P\bigl(D_{1,n}(q)\geq v\bigr)\P\bigl(T_{1,n}(q)\in dv\bigr)
 =c_n(vs,vt)\{-dp_n(vs,vt)\}.
\end{align*}
Note that \eqref{H1n} does not require densities.  By
Assumption~\ref{a:FG} we get, locally uniformly for $v,s,t>0$,
\begin{align}
 H_n(v;s,t)\longrightarrow
 H(v;s,t):=F^\circ(vs,vt)G^\circ(vs,vt).
 \label{Hlimit}
\end{align}

\textit{Step 2: uniform convergence of the empirical measures on $[\tau,1]$.}
The empirical versions, suitably normalized, are
\begin{align*}
 \widetilde H_n(v;q)&={1\over r_n^2}\sum_{i=1}^n
               \ind_{\{W_{i,n}(q)\geq v\}},\\
 \widetilde H_{1,n}((a,b];q)&={1\over r_n^2}\sum_{i=1}^n
 \Delta_{i,n}(q)\ind_{\{a<W_{i,n}(q)\leq b\}}.
\end{align*}
By \eqref{nu}, $\E[\widetilde H_n(v;q)]=H_n(v;q)$ and
$\E[\widetilde H_{1,n}((a,b];q)]=H_{1,n}((a,b];q)$, so these are unbiased.
We show first that, for every $\tau\in(0,1)$,
\begin{align}
 \sup_{q\in K}\sup_{\tau\leq v\leq1}
 |\widetilde H_n(v;q)-H_n(v;q)|&\overset{\P}{\longrightarrow}0,
 \label{empH}\\
 \sup_{q\in K}\sup_{\tau\leq v\leq1}
 |\widetilde H_{1,n}((\tau,v];q)-H_{1,n}((\tau,v];q)|
 &\overset{\P}{\longrightarrow}0.
 \label{empH1}
\end{align}
The argument is a pointwise-variance and then a monotonicity extension argument, similar to a Glivenko--Cantelli argument.  We
give the details because the directional ordering requires one additional
reparametrization.

Put $U_{i,n}=\psi_{1,n}(Z_i^{(1)})$ and $V_{i,n}=\psi_{2,n}(Z_i^{(2)})$.  The indicators in
\eqref{empH} are upper orthants,
\begin{align*}
 \ind_{\{W_{i,n}(s,t)\geq v\}}
 =\ind_{\{U_{i,n}\geq vs,V_{i,n}\geq vt\}},
\end{align*}
and are coordinatewise monotone in $(vs,vt)$.  For the failure subdistribution, split
the observations according to $(\delta_i^{(1)},\delta_i^{(2)})$.  If both coordinates
are uncensored, \eqref{Delta} equals one.  If the censoring mark is $(1,0)$ or $(0,1)$,
the failure is produced by the first or second coordinate, respectively.  Consequently,
for $\tau\leq a<b\leq1$,
\begin{align}
 \widetilde H_{1,n}((a,b];s,t)
={}&\widetilde K_n(as,at)-\widetilde K_n(bs,bt)\nonumber\\
 &+\widetilde J_{1,n}(bs,s/t)-\widetilde J_{1,n}(as,s/t)\nonumber\\
 &+\widetilde J_{2,n}(bt,t/s)-\widetilde J_{2,n}(at,t/s),
 \label{failuredecomp}
\end{align}
where
\begin{align*}
 \widetilde K_n(x,y)&={1\over r_n^2}\sum_{i=1}^n
 \ind_{\{U_{i,n}>x,V_{i,n}>y,\delta_i^{(1)}=\delta_i^{(2)}=1\}},\\
 \widetilde J_{1,n}(x,z)&={1\over r_n^2}\sum_{i=1}^n
 \ind_{\{\delta_i^{(1)}=1,\delta_i^{(2)}=0,\,
             \tau<U_{i,n}\leq x,\,U_{i,n}/V_{i,n}\leq z\}},\\
 \widetilde J_{2,n}(y,z)&={1\over r_n^2}\sum_{i=1}^n
 \ind_{\{\delta_i^{(1)}=0,\delta_i^{(2)}=1,\,
             \tau<V_{i,n}\leq y,\,V_{i,n}/U_{i,n}\leq z\}}.
\end{align*}
The three nonzero censoring configurations give the disjoint events
\begin{align*}
 \{\delta^{(1)}=\delta^{(2)}=1,\ a<W\leq b\}
 &=\{\delta^{(1)}=\delta^{(2)}=1,\ U>as,V>at\}
   \setminus\{U>bs,V>bt\},\\
 \{(\delta^{(1)},\delta^{(2)})=(1,0),\ \Delta=1,\ a<W\leq b\}
 &=\{(\delta^{(1)},\delta^{(2)})=(1,0),\ as<U\leq bs,\ U/V\leq s/t\},\\
 \{(\delta^{(1)},\delta^{(2)})=(0,1),\ \Delta=1,\ a<W\leq b\}
 &=\{(\delta^{(1)},\delta^{(2)})=(0,1),\ at<V\leq bt,\ V/U\leq t/s\}.
\end{align*}
These identities establish \eqref{failuredecomp}; note that the weak inequalities keep all
possible directional ties.
The population versions are obtained by replacing the empirical sums by their
expectations.

Each family of indicators appearing in \eqref{empH} and \eqref{failuredecomp} is
therefore indexed by a point of a compact rectangle and is monotone in that index.  Namely, writing
$O_{i,n}=(U_{i,n},V_{i,n},\delta_i^{(1)},\delta_i^{(2)})$ for the transformed
observation, the six families are
\begin{align}
 {\cal R}&=\bigl\{\ind_{\{U\geq x,V\geq y\}}:
      (x,y)\in[\tau,T]^2\bigr\},\nonumber\\
 {\cal K}&=\bigl\{\ind_{\{U>x,V>y,\ \delta^{(1)}=\delta^{(2)}=1\}}:
      (x,y)\in[\tau,T]^2\bigr\},\nonumber\\
 {\cal J}_1&=\bigl\{\ind_{\{\tau<U\leq x,\ U/V\leq z,\
      (\delta^{(1)},\delta^{(2)})=(1,0)\}}:
      (x,z)\in[\tau,T]\times[T^{-1},T]\bigr\},\nonumber\\
 {\cal J}_2&=\bigl\{\ind_{\{\tau<V\leq y,\ V/U\leq z,\
      (\delta^{(1)},\delta^{(2)})=(0,1)\}}:
      (y,z)\in[\tau,T]\times[T^{-1},T]\bigr\},\nonumber\\
 {\cal J}_1'&=\bigl\{\ind_{\{\tau<V\leq y,\ V/U\leq z,\
      (\delta^{(1)},\delta^{(2)})=(1,0)\}}:
      (y,z)\in[\tau,T]\times[T^{-1},T]\bigr\},\nonumber\\
 {\cal J}_2'&=\bigl\{\ind_{\{\tau<U\leq x,\ U/V\leq z,\
      (\delta^{(1)},\delta^{(2)})=(0,1)\}}:
      (x,z)\in[\tau,T]\times[T^{-1},T]\bigr\},
 \label{fourclasses}
\end{align}
the index ranges being those produced by $v\in[\tau,1]$ and $q=(s,t)\in K$, that is
$vs,vt\in[\tau,T]$ and $s/t,t/s\in[T^{-1},T]$.  Each member of \eqref{fourclasses} is
an indicator of an orthant in one of the three coordinate pairs $(U,V)$, $(U,U/V)$,
$(V,V/U)$, intersected with a fixed censoring mark, and is monotone in each of its two
index coordinates.  The last two families differ from ${\cal J}_2$ and ${\cal J}_1$
only in the censoring mark attached to them; they actually do not occur in
\eqref{failuredecomp}, but are listed here since they are required in Step~3 and in
Appendix~\ref{sec:appendix}.

Each family has a fixed envelope of bounded $\nu_n$-mass.  A member of ${\cal R}$ or
${\cal K}$ vanishes off $\{U\geq\tau,V\geq\tau\}$.  A member of ${\cal J}_1$ or of
${\cal J}_2'$ vanishes
off $\{U>\tau,V>\tau/T\}$, because $U>\tau$ and $U/V\leq z\leq T$ force $V\geq U/T>\tau/T$;
symmetrically, a member of ${\cal J}_2$ or of ${\cal J}_1'$ vanishes off
$\{U>\tau/T,V>\tau\}$.  All these
envelopes are contained in $E_\tau:=\{U\geq\tau/T,V\geq\tau/T\}$, and by \eqref{Hn},
\eqref{Hlimit},
\begin{align}
 \nu_n(E_\tau)=H_n\bigl(1;(\tau/T,\tau/T)\bigr)
 \longrightarrow F^\circ(\tau/T,\tau/T)\,G^\circ(\tau/T,\tau/T)<\infty,
 \label{envelopemass}
\end{align}
so that $\sup_n\nu_n(E_\tau)=:C_\tau<\infty$.

Each of the six families is a family of indicators of sets of
geometric type, namely of a set cut out by two inequalities, in one of the three fixed
coordinate pairs $(U,V)$, $(U,U/V)$, $(V,V/U)$, intersected with a fixed censoring mark.  For a
fixed real function $\xi$ the collection $\{\{\xi\leq x\}:x\in\mathbb R\}$ is a
Vapnik--Chervonenkis class of index $2$; intersections of two such collections, and
intersection with a fixed set, again give Vapnik--Chervonenkis classes of bounded index
\citep[Lemma~2.6.17]{vanderVaartWellner1996S}.  Hence each family in
\eqref{fourclasses} is a Vapnik--Chervonenkis class whose index is an absolute constant, and in particular it does not depend on $n$, on $\tau$, or on $T$.  For a class ${\cal G}$ whose members are
signed sums of at most six indicators of sets drawn from such classes, the covering
numbers satisfy
\begin{align}
 \sup_Q N\bigl(\varepsilon\|G\|_{Q,2},{\cal G},L^2(Q)\bigr)
 \leq C\varepsilon^{-\vartheta},\qquad 0<\varepsilon<1,
 \label{VC}
\end{align}
with $C$ and $\vartheta$ being fixed constants and the supremum being taken over \emph{all} probability measures $Q$
\citep[Theorem~2.6.7]{vanderVaartWellner1996S}.  The bound is therefore automatically
uniform in $n$, even if the law of $O_{1,n}$ changes with $n$.  
Also note that the families are indexed by compact subsets of Euclidean space and, under the strict and weak endpoint conventions, their indicators admit pointwise approximation from rational one-sided grids, so no measurability issues arise.

We now establish a consequence that is used here and in
Sections~\ref{sec:largesample} and~\ref{sec:appendix}.  Let ${\cal G}$ be any family of
functions as in \eqref{VC}, that is, signed sums of at most six indicators of sets
drawn from the classes underlying \eqref{fourclasses}, and suppose every member of
${\cal G}$ vanishes off a set $E$ with $\nu_n(E)\leq M$ (the envelope being a fixed
multiple of $\ind_E$, which only inflates the absolute constant below).  Then
\begin{align}
 \E^*\Biggl[\sup_{g\in{\cal G}}\bigl|{1\over r_n^2}\sum_{i=1}^ng(O_{i,n})-\nu_n(g)\bigr|\Biggr]
 \leq {C\sqrt M\over r_n} ,
 \label{maximalVC}
\end{align}
with $C$ an absolute constant.  Indeed, writing $P_n$ for the law of $O_{1,n}$ and
$\mathbb G_n=\sqrt n(\mathbb P_n-P_n)$ for the empirical process of the row, the
left-hand side equals $(\sqrt n/r_n^2)\,\E^*[\|\mathbb G_n\|_{\cal G}]$; the maximal
inequality \citep[Theorem~2.14.1]{vanderVaartWellner1996S} bounds
$\E^*[\|\mathbb G_n\|_{\cal G}]$ by $CJ(1,{\cal G})\|\ind_E\|_{P_n,2}$, the entropy
integral $J(1,{\cal G})=\int_0^1\{1+\log\sup_QN(\varepsilon\|\ind_E\|_{Q,2},{\cal G},
L^2(Q))\}^{1/2}d\varepsilon$ being finite and fixed by \eqref{VC}; and by
\eqref{nu},
$\|\ind_E\|_{P_n,2}=P_n(E)^{1/2}=\{r_n^2\nu_n(E)/n\}^{1/2}\leq(r_n^2M/n)^{1/2}$.
Multiplying the two factors gives \eqref{maximalVC}.

Applying \eqref{maximalVC} with $E=E_\tau$ and $M=C_\tau$ yields
\begin{align}
 \E^*\Biggl[\sup_{f\in{\cal F}}\bigl|{1\over r_n^2}\sum_{i=1}^nf(O_{i,n})-\nu_n(f)\bigr|\Biggr]
\longrightarrow0
 \qquad\text{for }
 {\cal F}\in\bigl\{{\cal R},{\cal K},{\cal J}_1,{\cal J}_2,
 {\cal J}_1',{\cal J}_2'\bigr\},
 \label{empsix}
\end{align}
and hence, through the decomposition \eqref{failuredecomp}, the same convergence holds for both \eqref{empH} and \eqref{empH1}.
With these ingredients in hand, let
\begin{align*}
 \widetilde\Lambda_n(q)=
 \int_{(0,1]}{\widetilde H_{1,n}(dv;q)\over\widetilde H_n(v;q)}.
\end{align*}
We establish a lower bound for $H$ that will be used repeatedly.  Put
\begin{align*}
 H(v;q)=F^\circ(vs,vt)G^\circ(vs,vt).
\end{align*}
By Assumption~\ref{a:FG}, this function is continuous and strictly positive on $(0,\infty)^3$,
and the set $\{(v,q):v\in[\tau,2],q\in K\}$ is compact.
Hence
\begin{align}
 \inf_{v\in[\tau,2],\,q\in K}H(v;q)=:2h_\tau>0 .
 \label{Hlower}
\end{align}
Since $H_n\to H$ uniformly on $[\tau,2]\times K$ by \eqref{Hlimit}, we also have
$H_n(v;q)\geq h_\tau$ on that set for all $n$ large, and the same holds with $K$
replaced by any compact subset of $(0,\infty)^2$, with $h_\tau$ depending on that set.
The difference of the two integrals over
$(\tau,1]$ equals
\begin{align*}
 &\int_{(\tau,1]}\bigl\{{1\over\widetilde H_n(v;q)}
                    -{1\over H_n(v;q)}\bigr\}\widetilde H_{1,n}(dv;q)\\
 &\quad+{\widetilde H_{1,n}((\tau,1];q)-H_{1,n}((\tau,1];q)\over H_n(1;q)}\\
 &\quad-\int_{(\tau,1]}\{\widetilde H_{1,n}((\tau,v];q)
                   -H_{1,n}((\tau,v];q)\}\,d\{H_n(v;q)^{-1}\}.
\end{align*}
By \eqref{Hlower}, we have $H_n\geq h_\tau$ on $[\tau,1]\times K$ for $n$ large; note that
\eqref{empH} gives
$\inf_{v\in[\tau,1],q\in K}\widetilde H_n(v;q)\geq h_\tau/2$ with probability tending to
one, and that
$\widetilde H_{1,n}((\tau,1];q)\leq\widetilde H_n(\tau;q)=O_{\P}(1)$ uniformly in $q$.
The three terms are then bounded as follows.  The first is at most
\begin{align*}
 {\sup_{v\in[\tau,1]}|\widetilde H_n(v;q)-H_n(v;q)|\over
  \widetilde H_n(v;q)H_n(v;q)}\,\widetilde H_{1,n}((\tau,1];q)
 \leq{2\over h_\tau^2}\sup_{v\in[\tau,1]}|\widetilde H_n-H_n|(v;q)\cdot O_{\P}(1)
 =o_{\P}(1)
\end{align*}
uniformly in $q$, by \eqref{empH}.  The second is at most
$h_\tau^{-1}\sup_{v}|\widetilde H_{1,n}-H_{1,n}|((\tau,v];q)=o_{\P}(1)$ by \eqref{empH1}.
The third is at most
$\sup_{v}|\widetilde H_{1,n}-H_{1,n}|((\tau,v];q)$ times the total variation of the
nondecreasing function $v\mapsto H_n(v;q)^{-1}$ on $[\tau,1]$, which equals
$H_n(1;q)^{-1}-H_n(\tau;q)^{-1}\leq h_\tau^{-1}$; it is therefore $o_{\P}(1)$ as well.
Hence
\begin{align}
 \sup_{q\in K}\Biggl|
 \int_{(\tau,1]}{\widetilde H_{1,n}(dv;q)\over\widetilde H_n(v;q)}
 -\int_{(\tau,1]}{H_{1,n}(dv;q)\over H_n(v;q)}
 \Biggr|\overset{\P}{\longrightarrow}0.
 \label{hazardtau}
\end{align}

\textit{Step 3: the interval $(0,\tau]$.}
We now treat the interval close to zero. For the full integral, use the exact decomposition
\begin{align*}
 \widetilde\Lambda_n(q)-\int_{(0,1]}{H_{1,n}(dv;q)\over H_n(v;q)}
={}&\int_{(0,1]}{(\widetilde H_{1,n}-H_{1,n})(dv;q)\over H_n(v;q)}\\
 &+\int_{(0,1]}{H_n(v;q)-\widetilde H_n(v;q)
       \over\widetilde H_n(v;q)H_n(v;q)}\,\widetilde H_{1,n}(dv;q).
\end{align*}

On $[\tau,1]$ the two integrals in the exact decomposition are controlled by
\eqref{hazardtau}.  On $(0,\tau]$ the envelope $E_\tau$ of Step~2 cannot be reused,
because it came, together with the bound \eqref{envelopemass}, from a fixed distance to the
coordinate axes, whereas here $H_n(v;q)\to\infty$ as $v\downarrow0$.  The plan is to cut
$(0,\tau]$ into layers on which $H_n$ has a fixed order of magnitude $2^mb_n$, and to
apply \eqref{maximalVC} on the $m$th layer under an envelope of $\nu_n$-mass
$O(2^mb_n)$, allowing the accuracy to deteriorate like $2^{(1-\alpha)m}b_n$.  Because the integrands carry the factor $H_n^{-1}$ which is of asymptotic order
$(2^mb_n)^{-1}$, the resulting layerwise error is $O(\varepsilon2^{-\alpha m})$, which
is summable in $m$; and because the layerwise deviation probabilities carry the factor
$2^{-(1/2-\alpha)m}$, so does a final union bound argument.

Take $\tau<1/(16T)$, and set
\begin{align*}
 b_n=\inf_{q\in K}H_n(\tau;q),\qquad
 I_{m,n}(q)=\{v\in(0,\tau]:2^m b_n\leq H_n(v;q)<2^{m+1}b_n\},\qquad m\geq0.
\end{align*}
Each $I_{m,n}(q)$ is an interval, possibly empty, and
$b_n\to\inf_{q\in K}H(\tau;q)>0$.

The localized classes on these intervals have a common finite envelope.  In fact, for
$q=(s,t)\in K=[1,T]^2$ the definition \eqref{TD} of $W_{1,n}$ gives
\begin{align*}
 {W_{1,n}(1,1)\over T}
 =\min\bigl\{{\psi_{1,n}(Z_1^{(1)})\over T},{\psi_{2,n}(Z_1^{(2)})\over T}\bigr\}
 \leq W_{1,n}(q)\leq W_{1,n}(1,1),\qquad q\in K,
\end{align*}
because $1\leq s,t\leq T$.  Consequently $\{W_{1,n}(1,1)\geq Tv\}
\subset\{W_{1,n}(q)\geq v\}\subset\{W_{1,n}(1,1)\geq v\}$, so that, if $(v,q)$ belongs
to the $m$th layer, Assumption~\ref{a:radialmodulus} with $A=T$ gives
\begin{align}
 \nu_n\{W_{1,n}(1,1)\geq v\}
 \leq C_T\,\nu_n\{W_{1,n}(1,1)\geq Tv\}
 \leq C_T H_n(v;q)\leq C_T2^{m+1}b_n,
 \label{layerenvelope}
\end{align}
the first inequality holding because $v\leq\tau<1/T$.
Let $\underline v_{m,n}$ be the infimum of the lower endpoints of the nonempty intervals
$I_{m,n}(q)$, $q\in K$, and put
\begin{align*}
 E_{m,n}=\{W_{1,n}(1,1)\geq \underline v_{m,n}\}.
\end{align*}
If an infimum is not attained, the event is defined by decreasing limits; if
$\underline v_{m,n}=0$, it is the whole sample space.  Monotone convergence in
\eqref{layerenvelope} gives
\begin{align}
 \nu_n(E_{m,n})\leq C_T2^{m+1} b_n.
 \label{layermass}
\end{align}
Every risk event $\{W_{1,n}(q)\geq v\}$ with $v\in I_{m,n}(q)$ is contained in
$E_{m,n}$, and so is each of the three mark-specific bands displayed after
\eqref{failuredecomp}, since each of them is contained in a risk event of the layer.
The cutoff $\tau$ in the index sets of ${\cal J}_1,{\cal J}_2,{\cal J}_1',{\cal J}_2'$
plays no part in the geometry underlying \eqref{VC} and may be lowered to the left
endpoint of the layer; i.e., the band formulas remain available below
$\tau$, whereas the functions $\widetilde J_1,\widetilde J_2$ of
\eqref{failuredecomp}, whose definition contains the fixed cutoff $\tau$, do not.

Fix $\varepsilon\in(0,1)$ and $\alpha\in(0,1/2)$.  On the $m$th layer we apply
\eqref{maximalVC} with $E=E_{m,n}$, so that $M=C2^mb_n$ by \eqref{layermass}: for each
of the classes of \eqref{fourclasses}, restricted to the layer,
\begin{align}
 \E^*\Biggl[\sup_{g}\bigl|{1\over r_n^2}\sum_{i=1}^n
 \bigl\{g(O_{i,n})-\E[g(O_{i,n})]\bigr\}\bigr|\Biggr]
 \leq{C\sqrt{2^mb_n}\over r_n} .
 \label{layermaximalVC}
\end{align}
Markov's inequality with the threshold $\varepsilon2^{(1-\alpha)m}b_n$ therefore gives,
for each layer, a failure probability at most
$C\sqrt{2^mb_n}/(r_n\varepsilon2^{(1-\alpha)m}b_n)
=C\varepsilon^{-1}2^{-(1/2-\alpha)m}/(r_n\sqrt{b_n})$, and summing over $m\geq0$ and
over the finitely many classes, we obtain an upper bound of the probability that the $m$th layer bound below fails, taking the union over $m$, given by
\begin{align}
{C\over\varepsilon\,r_n\sqrt{b_n}}\sum_{m\geq0}2^{-(1/2-\alpha)m}
 ={C\over\varepsilon\,r_n\sqrt{b_n}}=o(1),
 \label{layerprob}
\end{align}
because $\alpha<1/2$ makes the series converge and because $r_n\to\infty$ by
Assumption~\ref{a:rate} while $b_n\to\inf_{q\in K}H(\tau;q)>0$.  Outside that event we have,
uniformly in $q\in K$ and $m\geq0$,
\begin{align}
 \sup_{v\in I_{m,n}(q)}|\widetilde H_n(v;q)-H_n(v;q)|
 &\leq \varepsilon2^{(1-\alpha)m}b_n,\nonumber\\
 \sup_{v'<v;\ v',v\in I_{m,n}(q)}
 |(\widetilde H_{1,n}-H_{1,n})((v',v];q)|
 &\leq \varepsilon2^{(1-\alpha)m}b_n,
 \label{layerbounds}
\end{align}
and the same bound holds, with the same proof, for the empirical minus the population
mass of any band $\{W_{1,n}(q)\in(v',v]\}$ intersected with a fixed censoring mark and
with a ratio constraint $z_1<A^{(1)}_{1,n}(q)/A^{(2)}_{1,n}(q)\leq z_2$, since such a
band is obtained from the classes ${\cal J}_1$, ${\cal J}_2$, ${\cal J}_1'$,
${\cal J}_2'$ of \eqref{fourclasses} restricted to $E_{m,n}$ by intersections and
differences, and so again has uniform covering numbers of the form \eqref{VC}.

We now insert \eqref{layerbounds} into the two integrals.  Write the endpoints of
$I_{m,n}(q)$ as $v_m^-=v_{m,n}^-(q)\leq
v_m^+=v_{m,n}^+(q)$, suppressing $n$ and $q$ below, and let
$M(v)=(\widetilde H_{1,n}-H_{1,n})((v_m^-,v];q)$, so that $\sup_{v}|M(v)|\leq
C\varepsilon2^{(1-\alpha)m}b_n$ by \eqref{layerbounds}.  Since $H_n^{-1}$ is
nondecreasing and takes values in $(2^{-(m+1)}b_n^{-1},2^{-m}b_n^{-1}]$ on the layer,
its total variation there is at most $2^{-m}b_n^{-1}$, and integration by parts gives
\begin{align*}
 \Bigl|\int_{I_{m,n}(q)}{(\widetilde H_{1,n}-H_{1,n})(dv;q)\over H_n(v;q)}\Bigr|
 &=\Bigl|{M(v_m^+)\over H_n(v_m^+;q)}
       -\int_{I_{m,n}(q)}M(v)\,d\bigl\{H_n(v;q)^{-1}\bigr\}\Bigr|\\
 &\leq\sup_v|M(v)|\bigl\{{1\over H_n(v_m^+;q)}
       +\mathop{\rm TV}_{I_{m,n}(q)}\bigl(H_n^{-1}\bigr)\bigr\}\\
 &\leq C\varepsilon2^{(1-\alpha)m}b_n\cdot{2\over2^mb_n}
 =C\varepsilon2^{-\alpha m}.
\end{align*}
The first inequality in \eqref{layerbounds} also yields
$\widetilde H_n(v;q)\geq H_n(v;q)/2$ on the layer, after decreasing $\varepsilon$ by a
fixed constant, since $C\varepsilon2^{(1-\alpha)m}b_n\leq C\varepsilon
H_n(v;q)$ there.  Moreover, the observations counted by $\widetilde H_{1,n}(I_{m,n}(q);q)$
all satisfy $W_{i,n}(q)\geq v$ for any $v\in I_{m,n}(q)$ below them, so that, letting
$v\downarrow v_m^-$ inside the layer and using \eqref{layerbounds},
\begin{align*}
 \widetilde H_{1,n}(I_{m,n}(q);q)
 \leq \sup_{v\in I_{m,n}(q)}\widetilde H_n(v;q)
 \leq \sup_{v\in I_{m,n}(q)}H_n(v;q)+C\varepsilon2^{(1-\alpha)m}b_n\\
 \leq 2^{m+1}b_n+2^mb_n
 \leq C2^m b_n.
\end{align*}
Therefore
\begin{align*}
 &\Bigl|\int_{I_{m,n}(q)}{H_n(v;q)-\widetilde H_n(v;q)
       \over\widetilde H_n(v;q)H_n(v;q)}
       \widetilde H_{1,n}(dv;q)\Bigr|\\
 &\quad\leq{\sup_{v\in I_{m,n}(q)}|\widetilde H_n-H_n|(v;q)\over
       \bigl(2^mb_n/2\bigr)\bigl(2^mb_n\bigr)}
       \widetilde H_{1,n}(I_{m,n}(q);q)\\
 &\leq{C\varepsilon2^{(1-\alpha)m}b_n\over(2^mb_n)^2}\,C2^mb_n
 = C\varepsilon2^{-\alpha m}.
\end{align*}
Summing the two displays over $m\geq0$ bounds the contribution of $(0,\tau]$ to the
exact decomposition by $C\varepsilon\sum_{m\geq0}2^{-\alpha m}=C\varepsilon/(1-2^{-\alpha})$,
uniformly in $q\in K$, outside an event of probability $o(1)$.  Since $\varepsilon$ was
arbitrary, that contribution is $o_{\P}(1)$ uniformly in $q$.  Consequently,
\begin{align}
 \sup_{q\in K}\bigl|\widetilde\Lambda_n(q)
 -\int_{(0,1]}{H_{1,n}(dv;q)\over H_n(v;q)}\bigr|
 \overset{\P}{\longrightarrow}0.
 \label{hazardfull}
\end{align}

\textit{Step 4: from the hazard to the product limit.}
It remains to identify the population hazard and to compare the product limit with its
logarithm.  This is standard in other non-extreme contexts, but we give a full proof in
our more delicate setting under our assumptions.  Together,
\eqref{Hn}--\eqref{H1n} give the following equality of measures:
\begin{align*}
 {H_{1,n}(dv;q)\over H_n(v;q)}
 &={\{c_n(vs,vt)/c_n(1,1)\}
       [-d\{p_n(vs,vt)/p_n(1,1)\}]
    \over
    \{p_n(vs,vt)/p_n(1,1)\}\{c_n(vs,vt)/c_n(1,1)\}}\\
 &={-dp_n(vs,vt)\over p_n(vs,vt)}.
\end{align*}
The function $v\mapsto p_n(vs,vt)=\P\bigl(T_{1,n}(q)>v\bigr)$ is continuous, since an atom of
the random variable $T_{1,n}(q)=\min\{\psi_{1,n}(X^{(1)})/s,\psi_{2,n}(X^{(2)})/t\}$ at a
point $v_0$ would require $\P\bigl(\psi_{1,n}(X^{(1)})=v_0s\bigr)>0$ or
$\P\bigl(\psi_{2,n}(X^{(2)})=v_0t\bigr)>0$, because $\{T_{1,n}(q)=v_0\}\subset
\{\psi_{1,n}(X^{(1)})=v_0s\}\cup\{\psi_{2,n}(X^{(2)})=v_0t\}$, and the
margins of $X$ are continuous.  It is also clearly positive on $(0,1]$ and tends to
one as $v\downarrow0$.  Hence $v\mapsto\log p_n(vs,vt)$ is continuous and of bounded
variation on $[\varepsilon',1]$ for every $\varepsilon'>0$, and the chain rule for
continuous Stieltjes integrals gives
$-d\,p_n(vs,vt)/p_n(vs,vt)=-d\log p_n(vs,vt)$, so that
\begin{align}
 \int_{(0,1]}{H_{1,n}(dv;q)\over H_n(v;q)}
 =\int_{(0,1]}{-d p_n(vs,vt)\over p_n(vs,vt)}
 =-\bigl\{\log p_n(s,t)-\lim_{v\downarrow0}\log p_n(vs,vt)\bigr\}
 =-\log p_n(s,t).
 \label{populationhazard}
\end{align}

Let $\widetilde p_n(q)$ be the product limit formed from
$(W_{i,n}(q),\Delta_{i,n}(q))$.  Denote by
$1\geq v_1(q)>v_2(q)>\cdots$ its distinct failure times, listed in decreasing order,
and by
\begin{align*}
 d_\ell(q)=\sum_{i=1}^n\Delta_{i,n}(q)\ind_{\{W_{i,n}(q)=v_\ell(q)\}},
 \qquad
 Y_\ell(q)=\sum_{i=1}^n\ind_{\{W_{i,n}(q)\geq v_\ell(q)\}}
\end{align*}
the failure and risk counts there.  Because the four marginal distribution functions
are continuous, the variables $U_{1,n},\ldots,U_{n,n}$ are almost surely distinct, and
so are $V_{1,n},\ldots,V_{n,n}$.  Split the sample into
$P_1(q)=\{i:U_{i,n}/s\leq V_{i,n}/t\}$ and $P_2(q)=\{i:V_{i,n}/t<U_{i,n}/s\}$; then
$W_{i,n}(q)=U_{i,n}/s$ for $i\in P_1(q)$ and $W_{i,n}(q)=V_{i,n}/t$ for
$i\in P_2(q)$, so no two observations in the same part share a value of $W_{i,n}(q)$.
Hence
\begin{align}
 d_\ell(q)\leq2\qquad\forall\ell,q.
 \label{tietwo}
\end{align}
Write $Y_n^\ast(q)=\sum_{i=1}^n\ind_{\{W_{i,n}(q)\geq1\}}$ for the risk count at
$v=1$; since $v_\ell(q)\leq1$, we have $Y_\ell(q)\geq Y_n^\ast(q)$ for every $\ell$.
The uniform convergence of the empirical risk function established in Steps~2--3, the
identity $Y_n^\ast(q)=r_n^2\widetilde H_n(1;q)$ and the positivity of
$\inf_{q\in K}H(1;q)$, guaranteed by Assumption~\ref{a:FG}, show that
\begin{align}
 {Y_n^\ast(q)\over r_n^2}\stackrel{\mathbb{P}}{\longrightarrow}
 H(1;q)\quad\text{uniformly in }q\in K,\qquad
 \inf_{q\in K}H(1;q)>0.
 \label{riskatone}
\end{align}
We claim that, on the event $\{\inf_{q\in K}Y_n^\ast(q)>2\}$, whose probability tends
to one by \eqref{riskatone},
\begin{align}
 \sum_\ell {d_\ell(q)\over Y_\ell(q)^2}
 \leq{1\over Y_n^\ast(q)-2}
 \qquad\text{for all }q\in K.
 \label{ranksum}
\end{align}
Indeed, for $d\in\{1,2\}$ and $Y\geq d$ one has
$d/Y^2\leq\sum_{j=0}^{d-1}(Y-j)^{-2}$, so it suffices to show that the integers
\begin{align*}
 \bigl\{Y_\ell(q)-j:\ \ell\geq1,\ 0\leq j\leq d_\ell(q)-1\bigr\}
\end{align*}
are pairwise distinct and at least $Y_n^\ast(q)-1$.  For distinctness, note that
the $d_\ell(q)$ observations failing at $v_\ell(q)$ are counted in $Y_\ell(q)$ but not
in $Y_{\ell-1}(q)$, because $v_{\ell-1}(q)>v_\ell(q)$; hence
$Y_\ell(q)\geq Y_{\ell-1}(q)+d_\ell(q)$ and the integer blocks
$\{Y_\ell(q)-d_\ell(q)+1,\ldots,Y_\ell(q)\}$, $\ell\geq1$, are disjoint and increasing
in $\ell$.  For the lower bound, the smallest of all these integers is
$Y_1(q)-d_1(q)+1\geq Y_n^\ast(q)-1$ by \eqref{tietwo}.  Therefore the left-hand side
of \eqref{ranksum} is at most
$\sum_{m\geq Y_n^\ast(q)-1}m^{-2}\leq\int_{Y_n^\ast(q)-2}^\infty x^{-2}dx
=\{Y_n^\ast(q)-2\}^{-1}$, which is \eqref{ranksum}.
Next, for $0\leq x\leq1/2$,
\begin{align*}
 0\leq-\log(1-x)-x=\int_0^x{y\over1-y}\,dy
 \leq{1\over1-x}\int_0^xy\,dy={x^2\over2(1-x)}\leq x^2,
\end{align*}
and by \eqref{tietwo} and \eqref{riskatone} we have
$d_\ell(q)/Y_\ell(q)\leq2/Y_n^\ast(q)\leq1/2$ for all $\ell$ and $q$, with probability
tending to one.  On that event, by $d_\ell^2\leq2d_\ell$ and \eqref{ranksum},
\begin{align}
 0\leq\sum_\ell\Bigl[-\log\bigl\{1-{d_\ell(q)\over Y_\ell(q)}\bigr\}
             -{d_\ell(q)\over Y_\ell(q)}\Bigr]
 \leq \sum_\ell{d_\ell(q)^2\over Y_\ell(q)^2}
 \leq {2\over Y_n^\ast(q)-2}.
 \label{KMNA}
\end{align}
By \eqref{riskatone} the right-hand side of \eqref{KMNA} is $O_{\P}(r_n^{-2})$ uniformly
in $q\in K$.  Thus the logarithm of the product limit and the
Nelson--Aalen integral differ uniformly by an $O_{\P}(r_n^{-2})$ term.  On the same event no
factor of the product limit vanishes, so $\widetilde p_n(q)>0$ and, more precisely,
\begin{align*}
 \log{\widetilde p_n(q)\over p_n(q)}
 ={}&-\bigl\{\widetilde\Lambda_n(q)
       -\int_{(0,1]}{H_{1,n}(dv;q)\over H_n(v;q)}\bigr\}-\sum_\ell\Bigl[-\log\bigl\{1-{d_\ell(q)\over Y_\ell(q)}\bigr\}
             -{d_\ell(q)\over Y_\ell(q)}\Bigr].
\end{align*}
Using the above limit result and equations \eqref{hazardfull} and
\eqref{populationhazard}, we get that the uniform absolute value of the right-hand side
tends to zero.  Since
$|e^x-1|\leq e^{|x|}-1$, this also yields
\begin{align}
 \sup_{q\in K}\bigl|{\widetilde p_n(q)\over p_n(q)}-1\bigr|
 \overset{\P}{\longrightarrow}0.
 \label{oraclecons}
\end{align}
Steps~1--4 used no property of $K=[1,T]^2$ beyond its being a compact subset of
$[1/4,16T]^2$, the direction set of Assumption~\ref{a:radialmodulus}, and we state this for later use.  Let
$K^\ast=[\kappa,\kappa^\ast]^2$ with $1/4\leq\kappa\leq1\leq\kappa^\ast\leq16T$ and take
$\tau<\kappa/\kappa^\ast$.  In Step~2 the index ranges in \eqref{fourclasses} become
$[\tau\kappa,\kappa^\ast]$ for $x$ and $y$ and
$[\kappa/\kappa^\ast,\kappa^\ast/\kappa]$ for $z$, and the common envelope becomes
$\{U\geq\tau\kappa^2/\kappa^\ast,V\geq\tau\kappa^2/\kappa^\ast\}$, because
$U>\tau\kappa$ and $U/V\leq\kappa^\ast/\kappa$ force $V>\tau\kappa^2/\kappa^\ast$; it
is still of bounded $\nu_n$-mass by \eqref{Hn} and \eqref{Hlimit}.  In Step~3 the
inequalities
$W_{1,n}(1,1)/\kappa^\ast\leq W_{1,n}(q)\leq W_{1,n}(1,1)/\kappa$, valid for
$q\in K^\ast$, give
$\{W_{1,n}(1,1)\geq\kappa^\ast v\}\subset\{W_{1,n}(q)\geq v\}
\subset\{W_{1,n}(1,1)\geq\kappa v\}$, so that \eqref{layerenvelope} holds with
Assumption~\ref{a:radialmodulus} invoked at $A=\kappa^\ast/\kappa$ and at the argument
$\kappa v\leq\kappa\tau<\kappa/\kappa^\ast$.  Nothing else changes, and therefore
\eqref{empsix}, \eqref{empH}, \eqref{empH1}, \eqref{layerbounds} and
\eqref{oraclecons} all hold with their suprema taken over $K^\ast$.

\textit{Step 5: from the oracle to the estimated standardizations.}
Proposition~\ref{prop:stability}(i) of
Appendix~\ref{sec:appendix} applies with $K^\circ=K$ and gives
\begin{align}
 \sup_{q\in K}\bigl|\log\widehat p_n(q)-\log\widetilde p_n(q)\bigr|
 \overset{\P}{\longrightarrow}0 .
 \label{stabilitycons}
\end{align}
\textit{Conclusion.}
Combining \eqref{stabilitycons} with the oracle statement \eqref{oraclecons} of
Step~4, we obtain
\begin{align*}
 \sup_{q\in K}\bigl|\log\widehat p_n(q)-\log p_n(q)\bigr|
 \overset{\P}{\longrightarrow}0 ,
\end{align*}
and $\log\{p_n(s,t)/p_n(1,1)\}\to\log F^\circ(s,t)$ uniformly on $K$ by
Assumption~\ref{a:FG} and the positivity of $F^\circ$.  Subtracting the statement at $q=(1,1)$
gives $\sup_{q\in K}|\log\widehat F_n^\circ(q)-\log F^\circ(q)|\overset{\P}\to0$, and
since $F^\circ$ is bounded on $K$ the same statement holds without logarithms.  This proves the
theorem.

\subsection{Proof of Theorem~\ref*{thm:normality}}
Throughout, $\mathbb V$ denotes the centered Gaussian process indexed by $L^2(\nu)$
whose covariance is $\E[\mathbb V(f)\mathbb V(g)]=\nu(fg)$, and
\begin{align}
 \mathbb L(q)=-\mathbb V(\xi_q),\qquad q\in(0,\infty)^2 ,
 \label{Ldef}
\end{align}
so that $\mathbb G(q)=\mathbb L(\boldsymbol 1)-\mathbb L(q)$ has the covariance
\eqref{Glimit}.

\emph{Plan.}  The proof follows a linearization, empirical-process and
continuous-mapping structure.  Steps~1--5 treat the oracle
problem and Step~6 puts them together and plugs in the estimated standardizations.
Step~1 deals with the deterministic bias.  Step~2 identifies the limit of the
normalized law of a single observation; this is the key input from
which every covariance and continuity statement follows.
Step~3 proves weak convergence of the empirical risk and failure processes on the
region $[\tau,1]$, bounded away from the origin.  Step~4 linearizes the product limit
in that region.  Step~5 shows that the contribution of $(0,\tau]$ is negligible as
$\tau\downarrow0$, uniformly in $n$, which is what allows $\tau$ to be removed, and
constructs $\mathbb L$.  Step~6 passes from the cumulative hazard to the product limit,
inserts the estimated standardizations by using
Appendix~\ref{sec:appendix}, and finally takes the ratio.

Choose $\delta\in(0,1/2)$ and put
$\widetilde K=[1-\delta,T+\delta]^2$, with $K$ contained in its interior and
$\widetilde K\subset[1/4,4T]^2$.

\textit{Step 1: the deterministic approximation.}
Step~1 consists of Assumption~\ref{a:bias}, which we restate here in the form used
below: with $\widetilde K=[1-\delta,T+\delta]^2$,
\begin{align}
 \sup_{q\in\widetilde K}r_n\bigl|{p_n(q)\over p_n(1,1)}-F^\circ(q)\bigr|
 \longrightarrow0 .
 \label{bias2}
\end{align}
Since $\inf_{\widetilde K}F^\circ>0$ by Assumption~\ref{a:FG}, \eqref{bias2} also holds after
taking logarithms on both sides of the difference.

\textit{Step 2: the limit of the normalized law of one observation.}
Recall the normalization \eqref{nu} and abbreviate
$U=U_{1,n}=\psi_{1,n}(Z_1^{(1)})$, $V=V_{1,n}=\psi_{2,n}(Z_1^{(2)})$,
$\delta^{(j)}=\delta_1^{(j)}$.  Fix $\lambda\in(0,1)$ and define the finite measures
\begin{align*}
 \mu_n^X(B)&={\P\bigl((\psi_{1,n}(X^{(1)}),\psi_{2,n}(X^{(2)}))\in B\bigr)
                \over p_n(1,1)},\\
 \mu_n^C(B)&={\P\bigl((\psi_{1,n}(C^{(1)}),\psi_{2,n}(C^{(2)}))\in B\bigr)
                \over c_n(1,1)},
 \qquad B\subset[\lambda,\infty)^2\ \text{Borel}.
\end{align*}
Their survival functions are $\mu_n^X([x,\infty)\times[y,\infty))=p_n(x,y)/p_n(1,1)$ and
$\mu_n^C([x,\infty)\times[y,\infty))=c_n(x,y)/c_n(1,1)$, which by Assumption~\ref{a:FG}
converge, locally uniformly on $(0,\infty)^2$, to the continuous
functions $F^\circ$ and $G^\circ$.  Moreover the two families are uniformly tight on
$[\lambda,\infty)^2$, since for $L\geq1$,
\begin{align}
 \mu_n^X\bigl([\lambda,\infty)^2\setminus[\lambda,L]^2\bigr)
 \leq{p_n(L,\lambda)\over p_n(1,1)}+{p_n(\lambda,L)\over p_n(1,1)}
 \longrightarrow F^\circ(L,\lambda)+F^\circ(\lambda,L),
 \label{Ltail}
\end{align}
which tends to zero as $L\to\infty$ because $F^\circ(s,t)\to0$ when $s\vee t\to\infty$;
the same computation with $c_n$ and $G^\circ$ applies to $\mu_n^C$.  Convergence of
the survival functions to a continuous limit, together
with the tightness of \eqref{Ltail} and its counterpart for $\mu_n^C$, gives
weak convergence of finite measures on $[\lambda,\infty)^2$:
\begin{align}
 \mu_n^X\Longrightarrow\mu^X,\qquad \mu_n^C\Longrightarrow\mu^C,
 \label{muka}
\end{align}
where $\mu^X$ and $\mu^C$ are finite Borel measures on $[\lambda,\infty)^2$ with
survival functions $F^\circ$ and $G^\circ$.  Neither $\mu^X$ nor $\mu^C$ accumulates mass on a line
$\{x_1=c\}$, $\{x_2=c\}$, $\{c_1=c\}$ or
$\{c_2=c\}$, because $F^\circ$ and $G^\circ$ are continuous in each argument.

Now let
\begin{align*}
 \Psi(x_1,x_2,c_1,c_2)=
 \bigl(x_1\wedge c_1,\ x_2\wedge c_2,\
       \ind_{\{x_1\leq c_1\}},\ \ind_{\{x_2\leq c_2\}}\bigr).
\end{align*}
Since $X$ and $C$ are independent, since
$\P\bigl(\psi_{1,n}(Z^{(1)})>1,\psi_{2,n}(Z^{(2)})>1\bigr)=p_n(1,1)c_n(1,1)$ by \eqref{ZisXC},
and since
$\{U\geq\lambda,V\geq\lambda\}$ is exactly the event where all four scaled coordinates are at
least $\lambda$, we have the exact identity
\begin{align}
 \nu_n\bigl(\ \cdot\ \cap\{U\geq\lambda,V\geq\lambda\}\bigr)
 =\bigl(\mu_n^X\otimes\mu_n^C\bigr)\circ\Psi^{-1},
 \label{imagemeasure}
\end{align}
both sides being measures on $[\lambda,\infty)^2\times\{0,1\}^2$.  By \eqref{muka},
$ \mu_n^X\otimes\mu_n^C\Rightarrow\mu^X\otimes\mu^C$ on $[\lambda,\infty)^4$.  The map $\Psi$
is continuous off $\{x_1=c_1\}\cup\{x_2=c_2\}$, and this set is
$(\mu^X\otimes\mu^C)$-null, since by Fubini's theorem
$(\mu^X\otimes\mu^C)\{x_1=c_1\}=\int\mu^C\bigl(\{x_1\}\times[\lambda,\infty)\bigr)
\mu^X(dx_1\,dx_2)=0$, since $\mu^C$ accumulates no mass on vertical lines, and symmetrically for
$\{x_2=c_2\}$.  Combining \eqref{imagemeasure} with the continuous mapping theorem for
weak convergence therefore gives
\begin{align}
 \nu_n\bigl(\ \cdot\ \cap\{U\geq\lambda,V\geq\lambda\}\bigr)\Longrightarrow\nu^{(\lambda)}
 :=\bigl(\mu^X\otimes\mu^C\bigr)\circ\Psi^{-1}.
 \label{nulimit}
\end{align}
The measures $\nu^{(\lambda)}$ are restrictions of one another as $\lambda$ decreases, and we write
$\nu$ for the common extension.

We shall use \eqref{nulimit} only through the following derivation.  Call a set
\emph{admissible} if it is a finite combination, contained in
$\{U\geq\lambda,V\geq\lambda\}$, of sets of the form
\begin{align*}
 \{U\geq x\},\quad\{V\geq y\},\quad\{\delta^{(1)}=1\},\quad\{\delta^{(2)}=1\},\quad
 \{W_{1,n}(q)\geq v\},\quad\{\Delta_{1,n}(q)=1\},
\end{align*}
each with $\geq$ optionally replaced by $>$, together with $\{U\leq zV\}$ or
$\{V\leq zU\}$ intersected with a mark
$(\delta^{(1)},\delta^{(2)})\in\{(1,0),(0,1)\}$.  Note that all of these are fixed
subsets of $[\lambda,\infty)^2\times\{0,1\}^2$, not depending on $n$, since $W_{1,n}(s,t)=\min(U/s,V/t)$, and $\Delta_{1,n}(s,t)$ is determined by
$U,V,\delta^{(1)},\delta^{(2)}$ through \eqref{Delta}.  Then it is straightforward to
verify case by case that
\begin{align}
 \nu_n(A)\longrightarrow\nu(A)\qquad\text{for every admissible }A.
 \label{admissible}
\end{align}
Indeed, in the coordinates $(x_1,x_2,c_1,c_2)$ the boundary of such a set is contained
in a finite union of the sets
\begin{align*}
 \{x_j=c\},\quad\{c_j=c\},\quad\{x_j=c_j\},\quad
 \{x_1=zc_2\},\quad\{c_1=zx_2\},
\end{align*}
with $c,z$ constants.  The sets $\{U\geq x\}$ and $\{V\geq y\}$ contribute boundaries of
the first two types, since $U=x_1\wedge c_1$ and $V=x_2\wedge c_2$; the marks
$\{\delta^{(j)}=1\}=\{x_j\leq c_j\}$ contribute the third type; the set
$\{W_{1,n}(q)\geq v\}$ contributes the first two types, its boundary lying in
$\{x_1=sv\}\cup\{x_2=tv\}\cup\{c_1=sv\}\cup\{c_2=tv\}$; and the set
$\{\Delta_{1,n}(q)=1\}=\{\min(x_1/s,x_2/t)\leq\min(c_1/s,c_2/t)\}$ has boundary
contained in
$\{x_1=c_1\}\cup\{x_2=c_2\}\cup\{x_1=(s/t)c_2\}\cup\{c_1=(s/t)x_2\}$.  The ratio sets
contribute the last two types because on the mark $(1,0)$ one has $U=x_1$ and $V=c_2$,
and on the mark $(0,1)$ one has $U=c_1$ and $V=x_2$.  Each of these five sets is
$(\mu^X\otimes\mu^C)$-null, the first two by the line-null property of $\mu^X$ and
$\mu^C$ together with Fubini's theorem, the third as shown above, and the last two
because for fixed $x_1$ the section $\{c_2=x_1/z\}$ is $\mu^C$-null, and symmetrically
for $\{c_1=zx_2\}$.  This proves \eqref{admissible}.
Note that restricting to the two mixed marks is crucial because it keeps the ratio
sets admissible; a ratio set on the mark $(1,1)$ would require $\mu^X$ not to accumulate mass on the
ray $\{x_1=zx_2\}$, but no such set occurs in \eqref{failuredecomp}.

\textit{Step 3: the empirical tail process.}
Fix $\tau\in(0,1)$ and, for $v\in[\tau,1]$, set
\begin{align*}
 \phi_{0,n}(v;q)=\ind_{\{W_{1,n}(q)\geq v\}},\qquad
 \phi_{1,n}^\tau(v;q)=\Delta_{1,n}(q)\ind_{\{\tau<W_{1,n}(q)\leq v\}},
\end{align*}
with $\phi^\tau_{0,n}=\phi_{0,n}$, write $\phi^\tau_{a,i,n}(v;q)$ for the same
indicator evaluated at the $i$th observation, and write $\phi^\tau_a$ for the corresponding
indicator in the limiting coordinates of Step~2.  For $a,b\in\{0,1\}$ put
\begin{align}
 \Sigma_{ab}^\tau((v,q),(v',q'))=\lim_{n\to\infty}
 \nu_n\bigl\{\phi_{a,n}^\tau(v;q)\phi_{b,n}^\tau(v';q')\bigr\},
 \label{Sigma}
\end{align}
the argument of $\nu_n$ being an indicator, so that this is the normalized probability
of the intersection of the two underlying events; the existence of the limits, uniformly
in $v,v'\in[\tau,1]$ and in $q,q'\in\widetilde K$, is part of what is verified below.
Set $\mathbb V^\tau_a(v;q)=\mathbb V(\phi^\tau_a(v;q))$, with $\mathbb V$ the Gaussian
process of \eqref{Ldef}, and abbreviate $\mathbb V_0=\mathbb V^\tau_0$.  Being
integrals of one $\mathbb V$, these satisfy the compatibility relation
\begin{align}
 \mathbb V_1^\tau(v;q)=\mathbb V_1^{\tau'}(v;q)-\mathbb V_1^{\tau'}(\tau;q),
 \qquad 0<\tau'<\tau,
 \label{compatible}
\end{align}
automatically, it being the pathwise identity
$\phi_1^\tau(v;q)=\phi_1^{\tau'}(v;q)-\phi_1^{\tau'}(\tau;q)$ read under $\mathbb V$.
At the population thresholds consider
\begin{align*}
 \mathbb V_{0,n}(v;q)&=r_n\{\widetilde H_n(v;q)-H_n(v;q)\},\\
 \mathbb V_{1,n}((\tau,v];q)&=r_n
 \{\widetilde H_{1,n}((\tau,v];q)-H_{1,n}((\tau,v];q)\}.
\end{align*}
We prove that these two processes converge jointly to
$(\mathbb V_0,\mathbb V_1^\tau)$ on
$[\tau,1]\times\widetilde K$.  For this purpose, the five ingredients of the
triangular-array argument are verified next.

First consider the covariance limits.  Take
$\lambda=\tau(1-\delta)$.  Every product
$\phi^\tau_{a,n}(v;q)\phi^\tau_{b,n}(v';q')$ with
$v,v'\in[\tau,1]$ and $q,q'\in\widetilde K$ is the indicator of an admissible set in
the sense of Step~2.
Hence \eqref{admissible} gives the limits in
\eqref{Sigma}, with
$\Sigma^\tau_{ab}((v,q),(v',q'))=\nu\{\phi_a^\tau(v;q)\phi_b^\tau(v';q')\}$, where
$\phi^\tau_a$ denotes the corresponding indicator in the limiting coordinates.
By \eqref{reduction} and the independence of $T_{1,n}$ and $D_{1,n}$, the censoring
weight in these limits is evaluated at the larger of the two directional failure levels,
which is what makes the covariance depend on the censoring distribution.

Finally, the limits of the uncentered products in \eqref{Sigma} are the covariance
limits of the centered processes, because each $\nu_n(\phi^\tau_{a,n})$ is bounded, so
\begin{align*}
 {\E[\phi_{a,n}^{\tau}]\,\E[\phi_{b,n}^{\tau}]
  \over p_n(1,1)\,c_n(1,1)}
 ={r_n^2\over n}\,\nu_n\bigl(\phi^\tau_{a,n}\bigr)\nu_n\bigl(\phi^\tau_{b,n}\bigr)
 =O\bigl(r_n^2/n\bigr)=o(1),
\end{align*}
the last step because $r_n^2/n=p_n(1,1)c_n(1,1)\to0$.

We now establish finite-dimensional convergence.  Fix finitely many indices
$\theta_j=(v_j,q_j)\in[\tau,1]\times\widetilde K$, components
$a_j\in\{0,1\}$ and coefficients $c_j$, and let
\begin{align*}
 \zeta_{i,n}={1\over r_n}\sum_jc_j
 \bigl\{\phi_{a_j,n}^{\tau}(\theta_j)(O_{i,n})
       -\E[\phi_{a_j,n}^{\tau}(\theta_j)(O_{i,n})]\bigr\},
 \qquad i=1,\ldots,n,
\end{align*}
be the summands of a linear combination of terms from
$(\mathbb V_{0,n},\mathbb V_{1,n})$, which are independent and centered within the
row.  Since each $\phi$ takes values in $\{0,1\}$, $|\zeta_{i,n}|\leq2\sum_j|c_j|/r_n$,
which tends to zero; hence for each fixed $\eta>0$ the Lindeberg criterion sum
$\sum_{i=1}^n\E\bigl[\zeta_{i,n}^2\ind_{\{|\zeta_{i,n}|>\eta\}}\bigr]$ is identically zero for
all $n$ large.  At the same time
$\sum_{i=1}^n\E[\zeta_{i,n}^2]
=\sum_{j,l}c_jc_l\{\nu_n(\phi^\tau_{a_j,n}(\theta_j)\phi^\tau_{a_l,n}(\theta_l))
-(r_n^2/n)\nu_n(\phi^\tau_{a_j,n}(\theta_j))\nu_n(\phi^\tau_{a_l,n}(\theta_l))\}$
converges, by \eqref{Sigma} and the preceding display, to the quadratic form
$\sum_{j,l}c_jc_l\Sigma^\tau_{a_ja_l}(\theta_j,\theta_l)$.  The Lindeberg--Feller
theorem and the Cram\'er--Wold device then give the required finite-dimensional
convergence; this is the first condition.

Next, we establish an entropy bound.  Let $\mathcal F_n$ be the class of the two
indicator components divided by $r_n$, and equip it with
\begin{align*}
 \varrho_n(f,g)^2=\sum_{i=1}^n\E[\|f(O_{i,n})-g(O_{i,n})\|_2^2],
 \qquad O_{i,n}=(U_{i,n},V_{i,n},\delta_i^{(1)},\delta_i^{(2)}).
\end{align*}
Note that since the observations
are identically distributed within the row and $f\in\mathcal F_n$ is an indicator
divided by $r_n$, we have, for $f=g_1/r_n$ and $g=g_2/r_n$ with $g_1,g_2$ indicators,
\begin{align}
 \varrho_n(f,g)^2={n\over r_n^2}\E\bigl[(g_1(O_{1,n})-g_2(O_{1,n}))^2\bigr]
 =\nu_n\bigl(|g_1-g_2|\bigr);
 \label{rhonu}
\end{align}
that is, squared $\varrho_n$-distances are $\nu_n$-masses, exactly the quantity that the
maximal inequality \eqref{maximalVC} controls.  Further, the triangular-array central
limit theorem used below is stated for the uncentered class $\mathcal F_n$ and
performs the centering itself, so no centering of the class is required.

The class $\mathcal F_n$ has the same orthant structure as in the consistency proof.
The at-risk indicators are upper orthants in $(U_{i,n},V_{i,n})$.  By
\eqref{failuredecomp}, with the fixed lower cutoff changed from $\tau$ to
$\tau(1-\delta)$ so that the transformed region remains bounded away from the axes,
the failure indicators are signed sums of six orthant indicators from the classes
${\cal K}$, ${\cal J}_1$, ${\cal J}_2$ of \eqref{fourclasses}, with index ranges now
$x,y\in[\tau(1-\delta),T+\delta]$ and
$z\in[z_{\min},z_{\max}]$, $z_{\max}=z_{\min}^{-1}=(T+\delta)/(1-\delta)$.  All these
indicators vanish off the event
$E:=\{U_{1,n}\geq\tau(1-\delta)/z_{\max},V_{1,n}\geq\tau(1-\delta)/z_{\max}\}$. Indeed, for a
${\cal J}_1$ index, $U_{1,n}>\tau(1-\delta)$ and $U_{1,n}/V_{1,n}\leq z_{\max}$ force
$V_{1,n}\geq\tau(1-\delta)/z_{\max}$, and symmetrically for ${\cal J}_2$.  The mass of
that envelope is bounded uniformly in $n$, since by writing
$\theta=\tau(1-\delta)/z_{\max}$, \eqref{Hn} and \eqref{Hlimit} we get
\begin{align}
 \nu_n(E)=H_n\bigl(1;(\theta,\theta)\bigr)
 \longrightarrow F^\circ(\theta,\theta)G^\circ(\theta,\theta)<\infty,
 \label{envelopemass2}
\end{align}
so $\sup_n\nu_n(E)=:C_\tau<\infty$.

Each of the two families making up ${\cal F}_n$ is built from the same geometric sets
as in \eqref{fourclasses}, since every one of its members is a signed sum of at most six
indicators of the sets described there, all of them vanishing off $E$.  By the discussion preceding \eqref{VC}, those sets form
Vapnik--Chervonenkis classes of fixed index; sums of finitely many such classes, and
pairs of them, again have uniform covering numbers of the form \eqref{VC}
\citep[Lemma~2.6.17 and Theorem~2.6.7]{vanderVaartWellner1996S}.  Hence there are constants $C,\vartheta$ with
\begin{align}
 \sup_QN\bigl(\varepsilon\|G\|_{Q,2},{\cal F}_n,L^2(Q)\bigr)
 \leq C\varepsilon^{-\vartheta},\qquad0<\varepsilon<1 ,
 \label{arraycovering}
\end{align}
$G=\ind_E/r_n$ being the common envelope, the supremum running over all probability
measures $Q$, and the bound being free of $n$.  In particular the entropy integral
\begin{align*}
 \int_0^{\varsigma_n}\Bigl\{\log\sup_QN\bigl(\varepsilon\|G\|_{Q,2},{\cal F}_n,
 L^2(Q)\bigr)\Bigr\}^{1/2}d\varepsilon
 \leq\int_0^{\varsigma_n}\{\log(C\varepsilon^{-\vartheta})\}^{1/2}\,d\varepsilon
 \longrightarrow0
\end{align*}
for every $\varsigma_n\downarrow0$, because the integrand is integrable at zero.
Thus the random-entropy hypothesis of the triangular-array central limit
theorem is satisfied; for its random semimetric take $Q=\mathbb P_n$, while
$\sqrt n\|G\|_{\mathbb P_n,2}=O_{\P}(1)$ because
$\E[n\mathbb P_nG^2]=\nu_n(E)\leq C_\tau$.
More explicitly, writing $R_n=\sqrt n\|G\|_{\mathbb P_n,2}$, the substitution
$u=\varepsilon R_n$ bounds the theorem's entropy integral by
$R_n\int_0^{(\delta_n/R_n)\wedge1}\{\log(C\varepsilon^{-\vartheta})\}^{1/2}
\,d\varepsilon=o_{\P}(1)$ for every $\delta_n\downarrow0$; this is the second condition.

We turn to the envelope condition, which requires that
$\sum_{i=1}^n\E^*[\|f\|^2_{\mathcal F_n}\ind_{\{\|f\|_{\mathcal F_n}>\eta\}}]\to0$ for
every $\eta>0$, where $\|f\|_{\mathcal F_n}$ denotes the pointwise supremum over
$f\in\mathcal F_n$.  Both components of an $f\in\mathcal F_n$ are
indicators divided by $r_n$, so
\begin{align*}
 \|f\|_{\mathcal F_n}
 =\sup_{\substack{\tau\leq v\leq1\\q\in\widetilde K}}
 {1\over r_n}\Biggl\|
 \begin{pmatrix}
  \ind_{\{W_{1,n}(q)\geq v\}}\\
  \Delta_{1,n}(q)\ind_{\{\tau<W_{1,n}(q)\leq v\}}
 \end{pmatrix}\Biggr\|_2
 \leq {\sqrt2\over r_n}.
\end{align*}
Since $r_n\to\infty$, for each fixed $\eta>0$ we have $\sqrt2/r_n\leq\eta$ for all $n$
large, so the indicator $\ind_{\{\|f\|_{\mathcal F_n}>\eta\}}$ vanishes identically and
the displayed sum is exactly zero for all sufficiently large $n$.  Thus the envelope
condition is satisfied; this is the third condition.

We next upgrade the pointwise covariance convergence to uniform convergence.
Here discretization is needed, but it can be built under the fixed
limit measure $\nu$ of Step~2, and then applied to $\nu_n$.  Write $\theta=(v,q)$ and $\Sigma^\tau_{ab,n}(\theta,\theta')
=\nu_n\{\phi^\tau_{a,n}(\theta)\phi^\tau_{b,n}(\theta')\}$ for the finite-$n$
expression in \eqref{Sigma}.  Given $\varepsilon>0$ and one of the six orthant classes of \eqref{fourclasses},
partition the compact index range of each of its two index coordinates into finitely
many intervals whose endpoints are chosen so that the $\nu$-mass of the corresponding
open strip is at most $\varepsilon/2$; this is possible because $\nu(E)<\infty$, and it
requires at most $2\nu(E)/\varepsilon+1$ cut points per coordinate.  Taking the
resulting indicators at the cut points as endpoint functions produces, for every index
$\theta$, a bracket
$\phi^-_a(\theta)\leq\phi^\tau_{a,n}(\theta)\leq\phi^+_a(\theta)$ built from them,
clipped to $[0,1]$, with $\nu(\phi_a^+-\phi_a^-)\leq6\varepsilon$.  Every endpoint
function is the indicator of an admissible set, so \eqref{admissible} applies to it;
since there are finitely many of them, for all $n$ large we also have
$\nu_n(\phi^+_a(\theta)-\phi^-_a(\theta))\leq7\varepsilon$ simultaneously for all
$\theta$.  Using $|ff'-gg'|\leq|f-g|+|f'-g'|$ for numbers in $[0,1]$ and
$|\phi^\tau_{a,n}(\theta)-\phi^-_a(\theta)|\leq\phi^+_a(\theta)-\phi^-_a(\theta)$,
\begin{align*}
 \bigl|\Sigma^\tau_{ab,n}(\theta,\theta')
       -\nu_n\bigl\{\phi^-_a(\theta)\phi^-_b(\theta')\bigr\}\bigr|
 \leq\nu_n\bigl(\phi^+_a(\theta)-\phi^-_a(\theta)\bigr)
 +\nu_n\bigl(\phi^+_b(\theta')-\phi^-_b(\theta')\bigr)\leq14\varepsilon,
\end{align*}
and the same bound with $\nu$ in place of $\nu_n$ and $12\varepsilon$ in place of
$14\varepsilon$.  Each product $\phi^-_a(\theta)\phi^-_b(\theta')$ is a fixed simple
function, measurable with respect to the finite algebra generated by the endpoint sets;
every atom of that algebra is a combination of admissible sets and hence
admissible, so \eqref{admissible} gives
$\nu_n\{\phi^-_a(\theta)\phi^-_b(\theta')\}\to\nu\{\phi^-_a(\theta)\phi^-_b(\theta')\}$,
and the convergence is uniform over the finitely many possible products.  Combining
the two inequalities and one convergence gives
\begin{align}
 \limsup_{n\to\infty}
 \sup_{\substack{\tau\leq v,v'\leq1\\q,q'\in\widetilde K}}
 \bigl|\Sigma_{ab,n}^\tau(\theta,\theta')
       -\Sigma_{ab}^\tau(\theta,\theta')\bigr|
 \leq 26\,\varepsilon,\qquad a,b\in\{0,1\}.
 \label{uniformcov}
\end{align}
Letting $\varepsilon\downarrow0$ proves the uniformity asserted after \eqref{Sigma}; this is the fourth condition.

Let $\varrho$ be the limiting intrinsic semimetric, that is
$\varrho(\theta,\theta')^2=\Sigma^\tau_{aa}(\theta,\theta)
+\Sigma^\tau_{bb}(\theta',\theta')-2\Sigma^\tau_{ab}(\theta,\theta')$ for the
corresponding components, and let $\varrho_n$ denote the same expression with
$\Sigma^\tau_{ab}$ replaced by $\Sigma^\tau_{ab,n}$, consistent with the definition of
$\varrho_n$ above.  By \eqref{uniformcov}, $\|\varrho_n-\varrho\|_\infty\to0$.  Moreover
$\varrho$ is continuous because as $\theta''\to\theta$, the symmetric difference of the
underlying sets shrinks to a $\nu$-null boundary of the type listed in Step~2, so
$\nu|\phi^\tau_a(\theta'')-\phi^\tau_a(\theta)|\to0$ by dominated convergence.  A
continuous semimetric on a compact set is totally bounded, and for every
$\varsigma_n\downarrow0$,
\begin{align*}
 \sup_{\varrho(\theta,\theta')<\varsigma_n}\varrho_n(\theta,\theta')
 \leq\varsigma_n+\|\varrho_n-\varrho\|_\infty\longrightarrow0;
\end{align*}
this is the fifth and final condition.

Consequently we may apply the triangular-array central
limit theorem under the random-entropy condition
\citep[Theorem~2.11.1]{vanderVaartWellner1996S}. The entropy condition there is
stated with a random measure, and \eqref{arraycovering}, being uniform over
all probability measures $Q$, dominates it.  The same finite entropy integral is Dudley's sufficient
condition for the limiting Gaussian process to have a version with uniformly
$\varrho$-continuous sample paths; since $\varrho$ is itself continuous with respect to
the ordinary topology on $[\tau,1]\times\widetilde K$, as just shown, this upgrades to
continuity in the ordinary sense.  We therefore obtain the desired conclusion:
\begin{align}
 (\mathbb V_{0,n},\mathbb V_{1,n})
 \leadsto(\mathbb V_0,\mathbb V_1^\tau)
 \quad\text{in }[\ell^\infty([\tau,1]\times\widetilde K)]^2,
 \label{processconv}
\end{align}
with $(\mathbb V_0,\mathbb V_1^\tau)$ having a version with continuous sample paths on
$[\tau,1]\times\widetilde K$.

Two extensions of \eqref{processconv} will be needed, which we state for future reference.  First, the
argument used no property of $\widetilde K$ beyond compactness, so \eqref{processconv}
holds with $\widetilde K$ replaced by any compact subset of $[1/8,8T]^2$; combined
with the scaling identity \eqref{rayscaling}, which turns a statement for
$v\in[\tau/2,1]$ and $q\in2\widetilde K$ into one for $v\in[\tau,2]$ and
$q\in\widetilde K$, this extends \eqref{processconv} to $[\tau,2]\times\widetilde K$.
Second, the classes ${\cal J}_1'$ and ${\cal J}_2'$ of \eqref{fourclasses} differ from
${\cal J}_2$ and ${\cal J}_1$ only in the censoring mark attached to them, so they
satisfy \eqref{arraycovering} with the same constants and enter the envelope and
semimetric computations identically.  The same array central limit theorem therefore also applies, and in particular the family
\begin{align}
 \Bigl\{{1\over r_n}\sum_{i=1}^n\bigl[g(O_{i,n})-\E[g(O_{i,n})]\bigr]:
 g\in{\cal F}\Bigr\}\
 \label{aecsix}
\end{align}
is asymptotically equicontinuous in $\varrho_n$ for each of the six families ${\cal F}$ of \eqref{fourclasses}, with $v$ ranging over
$[\tau,2]$ and $q$ over $\widetilde K$.

\textit{Step 4: linearization on $[\tau,1]$.}
Let
\begin{align*}
 \widetilde\Lambda_n(q)=
 \int_{(0,1]}{\widetilde H_{1,n}(dv;q)\over\widetilde H_n(v;q)}.
\end{align*}
For fixed $\tau\in(0,1)$, define the truncated linear term
\begin{align}
 \mathbb L_{n,\tau}(q)={}&
 \int_{(\tau,1]}{\mathbb V_{1,n}(dv;q)\over H_n(v;q)}
 -\int_{(\tau,1]}{\mathbb V_{0,n}(v;q)\over H_n(v;q)^2}
 H_{1,n}(dv;q).
 \label{Lntau}
\end{align}
and write
\begin{align}
 \mathbb L_\tau(q)={\mathbb V_1^\tau(1;q)\over H(1;q)}
 -\int_{(\tau,1]}\mathbb V_1^\tau(v;q)\,d\{H(v;q)^{-1}\}
 -\int_{(\tau,1]}{\mathbb V_0(v;q)\over H(v;q)^2}H_1(dv;q)
 \label{Ltau}
\end{align}
for the corresponding limiting expression.
On $[\tau,1]$, both $H_n$ and, by \eqref{processconv}, $\widetilde H_n$ are uniformly
bounded away from zero with probability tending to one.  We claim that, uniformly in
$q\in\widetilde K$,
\begin{align}
 r_n\Bigl\{
 \int_{(\tau,1]}{\widetilde H_{1,n}(dv;q)\over\widetilde H_n(v;q)}
 -\int_{(\tau,1]}{H_{1,n}(dv;q)\over H_n(v;q)}
 \Bigr\}
={}&\mathbb L_{n,\tau}(q)+o_{\P}(1).
 \label{linearization}
\end{align}
To verify this, we use the exact
decomposition
\begin{align}
 \int_{(\tau,1]}{\widetilde H_{1,n}(dv)\over\widetilde H_n(v)}
 -\int_{(\tau,1]}{H_{1,n}(dv)\over H_n(v)}
 =\Pi_{1,n}-\Pi_{2,n}+\Pi_{3,n}-\Pi_{4,n},
 \label{fourterms}
\end{align}
with
\begin{align*}
 \Pi_{1,n}&=\int_{(\tau,1]}{(\widetilde H_{1,n}-H_{1,n})(dv)\over H_n(v)},
 &\Pi_{2,n}&=\int_{(\tau,1]}{\widetilde H_n(v)-H_n(v)\over H_n(v)^2}\,H_{1,n}(dv),\\
 \Pi_{3,n}&=\int_{(\tau,1]}{\{\widetilde H_n(v)-H_n(v)\}^2\over
                         \widetilde H_n(v)H_n(v)^2}\,H_{1,n}(dv),
 &\Pi_{4,n}&=\int_{(\tau,1]}{\widetilde H_n(v)-H_n(v)\over
                         \widetilde H_n(v)H_n(v)}\,
          (\widetilde H_{1,n}-H_{1,n})(dv).
\end{align*}
By the definitions of $\mathbb V_{0,n}$ and $\mathbb V_{1,n}$ we have the exact
identity $r_n(\Pi_{1,n}-\Pi_{2,n})=\mathbb L_{n,\tau}$, so \eqref{linearization}
is implied by showing that $r_n\Pi_{3,n}$ and $r_n\Pi_{4,n}$ are $o_{\P}(1)$ uniformly in
$q$.

The term $\Pi_{3,n}$ is $O_{\P}(r_n^{-2})$ uniformly, since \eqref{processconv} gives
$\sup_{v,q}|\widetilde H_n-H_n|=O_{\P}(r_n^{-1})$, the denominators are bounded away
from zero on $[\tau,1]$, and $H_{1,n}((\tau,1];q)$ is bounded uniformly by
$H_n(\tau;q)$, hence uniformly in $n$ and $q$.  Therefore $r_n\Pi_{3,n}=O_{\P}(r_n^{-1})
=o_{\P}(1)$.

For $\Pi_{4,n}$ write
\begin{align*}
 f_n(v;q)={\mathbb V_{0,n}(v;q)\over
                  \widetilde H_n(v;q)H_n(v;q)}.
\end{align*}
Since $\widetilde H_n-H_n=\mathbb V_{0,n}/r_n$ and
$(\widetilde H_{1,n}-H_{1,n})(dv;q)=\mathbb V_{1,n}(dv;q)/r_n$, we obtain
\begin{align*}
 r_n\Pi_{4,n}
 =r_n\int_{(\tau,1]}{\mathbb V_{0,n}(v;q)/r_n\over
       \widetilde H_n(v;q)H_n(v;q)}\,{\mathbb V_{1,n}(dv;q)\over r_n}
 ={1\over r_n}\int_{(\tau,1]}f_n(v;q)\,\mathbb V_{1,n}(dv;q).
\end{align*}
Fix a partition $\pi:\tau=t_0<\cdots<t_J=1$ and let $f_n^\pi(v;q)$ equal
$f_n(t_{j-1};q)$ on $(t_{j-1},t_j]$.  For fixed $\pi$,
\begin{align*}
 \sup_{q\in\widetilde K}\bigl|{1\over r_n}
       \int_{(\tau,1]}f_n^\pi(v;q)\,\mathbb V_{1,n}(dv;q)\bigr|
 \leq {1\over r_n}\sum_{j=1}^J&
       \sup_{q\in\widetilde K}|f_n(t_{j-1};q)|\times\sup_{q\in\widetilde K}
       |\mathbb V_{1,n}((t_{j-1},t_j];q)|=o_{\P}(1).
\end{align*}
Every factor in the finite sum is tight by \eqref{processconv}.  The approximation
error satisfies
\begin{align*}
 \sup_{q\in\widetilde K}\bigl|{1\over r_n}
       \int_{(\tau,1]}(f_n-f_n^\pi)(v;q)\,
                   \mathbb V_{1,n}(dv;q)\bigr|
 \leq \|f_n-f_n^\pi\|_\infty
       \sup_{q\in\widetilde K}
       {\mathop{\rm TV}_{(\tau,1]}\{\mathbb V_{1,n}(\,\cdot\,;q)\}\over r_n}.
\end{align*}
The second factor is bounded by
\begin{align*}
 \sup_{q\in\widetilde K}
 \{\widetilde H_{1,n}((\tau,1];q)+H_{1,n}((\tau,1];q)\}=O_{\P}(1).
\end{align*}
For the first factor, note that
\begin{align*}
\|f_n-f_n^\pi\|_\infty&\leq
C\sup_{|v-v'|\leq|\pi|}\sup_q|\mathbb V_{0,n}(v;q)-\mathbb V_{0,n}(v';q)|
\\
&\quad+C\|\mathbb V_{0,n}\|_\infty\sup_{|v-v'|\leq|\pi|}\sup_q
|(\widetilde H_nH_n)(v;q)-(\widetilde H_nH_n)(v';q)|,
\end{align*}
where $|\pi|$ is the mesh and $C$ is a bound for the reciprocal of
$\widetilde H_nH_n$ on $[\tau,1]\times\widetilde K$.  The first supremum tends to zero
in probability as $|\pi|\downarrow0$ after $n\to\infty$, by the asymptotic
equicontinuity contained in \eqref{processconv} together with the continuity of
$\varrho$ established above; the second does so because $\widetilde H_n\to H$ and $H_n\to H$
uniformly and $H$ is uniformly continuous, while $\|\mathbb V_{0,n}\|_\infty=O_{\P}(1)$.
Hence $\|f_n-f_n^\pi\|_\infty=o_{\P}(1)$.  Thus, $r_n\Pi_{4,n}=o_{\P}(1)$ uniformly in $q$, which together with the bound for
$\Pi_{3,n}$ proves \eqref{linearization}.

Uniformly on $[\tau,1]\times\widetilde K$, \eqref{Hn}--\eqref{H1n} give
$H_n\to H$ and
$H_{1,n}((\tau,v];q)\to H_1((\tau,v];q)$.  Integration by parts gives
\begin{align*}
 \int_{(\tau,1]}{\mathbb V_{1,n}(dv;q)\over H_n(v;q)}
 ={\mathbb V_{1,n}((\tau,1];q)\over H_n(1;q)}
 -\int_{(\tau,1]}\mathbb V_{1,n}((\tau,v];q)
                  \,d\{H_n(v;q)^{-1}\}.
\end{align*}
The increasing Stieltjes measures induced by $H_n^{-1}$ converge weakly to the measure
induced by $H^{-1}$.  Likewise, $H_{1,n}(dv)/H_n(v)^2$ converges weakly to
$H_1(dv)/H(v)^2$.  Uniform convergence of the deterministic functions, together with the
asymptotic equicontinuity of \eqref{processconv}, makes both integral maps continuous
at the limiting processes.  The continuous mapping theorem can therefore be applied and gives
\begin{align}
 \mathbb L_{n,\tau}\leadsto\mathbb L_\tau
 \qquad\text{in }\ell^\infty(\widetilde K).
 \label{Lntauconv}
\end{align}

\textit{Step 5: the interval $(0,\tau]$ and the construction of $\mathbb L$.}
It remains to show that the part of the hazard integral on $(0,\tau]$ is
negligible, uniformly in $n$, as $\tau\downarrow0$; this is the most delicate and arduous step.  Hence we do this in three parts.  We first establish an
exact variance computation for the linear part, then we obtain a layerwise maximal inequality that
turns that computation into an appropriate bound for the supremum over $q$, and finally treat the two nonlinear remainder terms $\Pi_{3,n}$ and $\Pi_{4,n}$ below $\tau$. 

Define the lower linear term
\begin{align*}
 D_{n,\tau}(q)={}&
 \int_{(0,\tau]}{\mathbb V_{1,n}(dv;q)\over H_n(v;q)}
 -\int_{(0,\tau]}{\mathbb V_{0,n}(v;q)\over H_n(v;q)^2}
                 H_{1,n}(dv;q).
\end{align*}
The deterministic parts in the two empirical processes cancel, and we may therefore write
\begin{align*}
 D_{n,\tau}(q)={1\over r_n}\sum_{i=1}^n\ell_{i,n,\tau}(q),
\end{align*}
where
\begin{align*}
 \ell_{i,n,\tau}(q)={}&
{\Delta_{i,n}(q)\ind_{\{W_{i,n}(q)\leq\tau\}}\over H_n(W_{i,n}(q);q)}-\int_{(0,\tau]}{\ind_{\{W_{i,n}(q)\geq v\}}
        \over H_n(v;q)^2}H_{1,n}(dv;q)=A_i-B_i.
\end{align*}
The latter terms have the same expectation. Indeed, by \eqref{nu},
\begin{align*}
 \E[A_i]&={r_n^2\over n}\int_{(0,\tau]}{H_{1,n}(dv;q)\over H_n(v;q)},
 \\
 \E[B_i]&=\int_{(0,\tau]}{\E[\ind_{\{W_{i,n}(q)\geq v\}}]\over H_n(v;q)^2}H_{1,n}(dv;q)
 ={r_n^2\over n}\int_{(0,\tau]}{H_{1,n}(dv;q)\over H_n(v;q)},
\end{align*}
so $\ell_{i,n,\tau}(q)$ has mean zero.  Its second moment is computed as
follows.  We have
\begin{align*}
 \E[A_i^2]=\E\Biggl[{\Delta_{i,n}(q)\ind_{\{W_{i,n}(q)\leq\tau\}}\over
       H_n(W_{i,n}(q);q)^2}\Biggr]
 ={r_n^2\over n}\int_{(0,\tau]}{H_{1,n}(dv;q)\over H_n(v;q)^2}.
\end{align*}
For the remaining two moments use
$\E\bigl[\ind_{\{W_{i,n}(q)\geq u\}}\ind_{\{W_{i,n}(q)\geq v\}}\bigr]
=(r_n^2/n)H_n(u\vee v;q)$, so that
\begin{align*}
 \E[B_i^2]={r_n^2\over n}\int_{(0,\tau]}\int_{(0,\tau]}
 {H_n(u\vee v;q)\over H_n(u;q)^2H_n(v;q)^2}\,H_{1,n}(du;q)H_{1,n}(dv;q),
\end{align*}
while
\begin{align*}
 \E[A_iB_i]
 &=\int_{(0,\tau]}\int_{(0,\tau]}\ind_{\{v\leq u\}}
 {\P\bigl(W_{i,n}(q)\in du,\Delta_{i,n}(q)=1\bigr)\over H_n(u;q)H_n(v;q)^2}H_{1,n}(dv;q)\\
 &={r_n^2\over n}\int_{(0,\tau]}\int_{(0,\tau]}\ind_{\{v\leq u\}}
 {H_{1,n}(du;q)H_{1,n}(dv;q)\over H_n(u;q)H_n(v;q)^2}.
\end{align*}
Split the double integral in $\E[B_i^2]$ into $\{v<u\}$, $\{u<v\}$, and the diagonal
$\{u=v\}$.  The diagonal is null because $H_{1,n}(\cdot;q)$ is atomless.
On $\{v<u\}$ one has $H_n(u\vee v;q)=H_n(u;q)$ and the integrand becomes
$\{H_n(u;q)H_n(v;q)^2\}^{-1}$, which is the integrand of $\E[A_iB_i]$; on $\{u<v\}$
the same expression is obtained after interchanging the names of $u$ and $v$.  Hence
$\E[B_i^2]=2\E[A_iB_i]$ and
\begin{align*}
 \E[\ell_{i,n,\tau}(q)^2]=\E[A_i^2]-2\E[A_iB_i]+\E[B_i^2]=\E[A_i^2] .
\end{align*}
Since the $\ell_{i,n,\tau}(q)$, $i=1,\ldots,n$, are independent and centered,
\begin{align}
 \E[D_{n,\tau}(q)^2]
 ={n\over r_n^2}\E[\ell_{1,n,\tau}(q)^2]
 =\int_{(0,\tau]}{H_{1,n}(dv;q)\over H_n(v;q)^2}
 \leq {1\over H_n(\tau;q)},
 \label{exactlower}
\end{align}
the last inequality holding because the failure measure is dominated by
the distribution measure of $W_{1,n}(q)$, and so
\begin{align*}
 \int_{(0,\tau]}{H_{1,n}(dv;q)\over H_n(v;q)^2}
 \leq\int_{(0,\tau]}{-dH_n(v;q)\over H_n(v;q)^2}
 ={1\over H_n(\tau;q)}-\lim_{v\downarrow0}{1\over H_n(v;q)}
 \leq {1\over H_n(\tau;q)}.
\end{align*}
The same cancellation and the same domination, carried out under $\nu$ on $(0,1]$
rather than under $\nu_n$ on $(0,\tau]$, give
$\nu(\xi_q^2)=\int_{(0,1]}H_1(dv;q)/H(v;q)^2\leq1/H(1;q)$, which is
\eqref{limitlower}.

The bound \eqref{exactlower} is pointwise in $q$; to make it uniform we again cut
$(0,\tau]$ into layers.  Take $\tau<1/(16T)$ and put
\begin{align*}
 b_{n,\tau}&=\inf_{q\in\widetilde K}H_n(\tau;q),\\
 I_{m,n}(q)&=\{v\in(0,\tau]:
        2^m b_{n,\tau}\leq H_n(v;q)<2^{m+1}b_{n,\tau}\},\qquad m\geq0.
\end{align*}
For $q=(s,t)\in\widetilde K=[1-\delta,T+\delta]^2$,
\begin{align*}
 {W_{1,n}(1,1)\over T+\delta}\leq W_{1,n}(q)
 \leq {W_{1,n}(1,1)\over1-\delta}.
\end{align*}
Let $\underline v_{m,n}$ be the infimum of the lower endpoints of the nonempty intervals
$I_{m,n}(q)$ and put
$E_{m,n}=\{W_{1,n}(1,1)\geq(1-\delta)\underline v_{m,n}\}$.  The preceding display and
Assumption~\ref{a:radialmodulus}, with $A=(T+\delta)/(1-\delta)$, show by monotone limits that
\begin{align}
 \nu_n(E_{m,n})\leq C2^m b_{n,\tau}.
 \label{layermass2}
\end{align}
Every risk or failure band with endpoints in $I_{m,n}(q)$ vanishes off $E_{m,n}$.

For later use define the two centered indicator processes
\begin{align*}
 A_{m,n}=\sup_{q\in\widetilde K}\sup_{v\in\overline{I_{m,n}(q)}}
              |\mathbb V_{0,n}(v;q)|,\quad B_{m,n}=\sup_{q\in\widetilde K}
       \sup_{\substack{v'<v\\v',v\in\overline{I_{m,n}(q)}}}
              |\mathbb V_{1,n}((v',v];q)|.
\end{align*}
The class for $A_{m,n}$ consists of upper-orthant indicators, and by
\eqref{failuredecomp} the class for $B_{m,n}$ consists of signed sums of at most six
indicators from \eqref{fourclasses}.  Both classes vanish off $E_{m,n}$, whose
$\nu_n$-mass is at most $C2^mb_{n,\tau}$ by \eqref{layermass2}.  Thus
\eqref{maximalVC}, multiplied by $r_n$, gives
\begin{align}
 \E^*[A_{m,n}]\leq C\sqrt{2^m b_{n,\tau}},
 \qquad
 \E^*[B_{m,n}]\leq C\sqrt{2^m b_{n,\tau}}.
 \label{ABmoments}
\end{align}

Define the layer contribution directly by
\begin{align*}
 D_{n,\tau}^{(m)}(q)={}&
 \int_{I_{m,n}(q)}{\mathbb V_{1,n}(dv;q)\over H_n(v;q)}
 -\int_{I_{m,n}(q)}{\mathbb V_{0,n}(v;q)\over H_n(v;q)^2}H_{1,n}(dv;q).
\end{align*}
Write the endpoints of $I_{m,n}(q)$ as
$v_m^-=v_{m,n}^-(q)\leq v_m^+=v_{m,n}^+(q)$, suppressing $n$ and $q$ below.
Integration by parts and the fact that
$H_n^{-1}$ is nondecreasing with total variation at most $(2^mb_{n,\tau})^{-1}$
on the layer give
\begin{align*}
 \Bigl|\int_{I_{m,n}(q)}{\mathbb V_{1,n}(dv;q)\over H_n(v;q)}\Bigr|
 &\leq \sup_{v\in I_{m,n}(q)}|\mathbb V_{1,n}((v_m^-,v];q)|
 \Bigl\{{1\over H_n(v_m^+;q)}+
 \mathop{\rm TV}_{I_{m,n}(q)}(H_n^{-1})\Bigr\}\\
 &\leq {2B_{m,n}\over2^mb_{n,\tau}}.
\end{align*}
For the compensator, the domination $H_{1,n}(dv;q)\leq-dH_n(v;q)$ used in
\eqref{exactlower} yields
\begin{align*}
 \Bigl|\int_{I_{m,n}(q)}{\mathbb V_{0,n}(v;q)\over H_n(v;q)^2}
 H_{1,n}(dv;q)\Bigr|
 \leq A_{m,n}\int_{I_{m,n}(q)}{H_{1,n}(dv;q)\over H_n(v;q)^2}
 \leq {C A_{m,n}\over2^mb_{n,\tau}}.
\end{align*}
Consequently, \eqref{ABmoments} gives directly
\begin{align}
 \E^*\Bigl[\sup_{q\in\widetilde K}|D_{n,\tau}^{(m)}(q)|\Bigr]
 \leq {C\over\sqrt{2^m b_{n,\tau}}}.
 \label{layermaximal}
\end{align}

Since $D_{n,\tau}=\sum_{m\geq0}D_{n,\tau}^{(m)}$, \eqref{layermaximal} gives
\begin{align*}
 \E^*\bigl[\sup_{q\in\widetilde K}|D_{n,\tau}(q)|\bigr]
 \leq {C\over\sqrt{b_{n,\tau}}}\sum_{m\geq0}2^{-m/2}
 \leq {C\over\sqrt{b_{n,\tau}}}.
\end{align*}
Here $b_{n,\tau}\to b_\tau:=\inf_{q\in\widetilde K}H(\tau;q)$ by the uniform
convergence $H_n\to H$, and Assumption~\ref{a:divergence} applied with $K'=\widetilde K$ gives
\begin{align}
 \lim_{\tau\downarrow0}b_\tau=\infty.
 \label{btaudiverge}
\end{align}
Then an application of Markov's inequality gives
\begin{align}
 \lim_{\tau\downarrow0}\limsup_{n\to\infty}
 \P\Bigl(\sup_{q\in\widetilde K}|D_{n,\tau}(q)|>\varepsilon\Bigr)=0.
 \label{lowerlinear}
\end{align}

We turn to the two nonlinear remainders below $\tau$.  Denote by
$\Pi_{3,n}^{(m)}$ and $\Pi_{4,n}^{(m)}$ the contributions of
$I_{m,n}(q)$ to $\Pi_{3,n}$ and $\Pi_{4,n}$ of
\eqref{fourterms}, now integrated over $(0,\tau]$ instead of $(\tau,1]$.  The
quantities $A_{m,n}$ and $B_{m,n}$ were defined and bounded above precisely for these
layers.

Consequently,
\begin{align*}
 \E^*\Biggl[\sup_{q\in\widetilde K}\sup_{0<v\leq\tau}
 {|\widetilde H_n(v;q)-H_n(v;q)|\over H_n(v;q)}
 \Biggr]
 &\leq {1\over r_n}\sum_{m\geq0}
       {\E^*[A_{m,n}]\over2^m b_{n,\tau}}\leq {C\over r_n\sqrt{b_{n,\tau}}}\sum_{m\geq0}2^{-m/2}=o(1).
\end{align*}
By Markov's inequality, the \emph{relative-risk event}
\begin{align}
 {\cal E}_n=\Bigl\{\sup_{q\in\widetilde K}\sup_{0<v\leq\tau}
 {|\widetilde H_n(v;q)-H_n(v;q)|\over H_n(v;q)}\leq{1\over2}\Bigr\}
 \label{relrisk}
\end{align}
has probability tending to one.  On ${\cal E}_n$ we have $\widetilde H_n\geq H_n/2$ and, for
$v$ in the $m$th layer, $|\widetilde H_n-H_n|(v;q)\leq H_n(v;q)/2\leq2^mb_{n,\tau}$, that
is
\begin{align}
 A_{m,n}\leq r_n\,2^mb_{n,\tau}\qquad\text{on }{\cal E}_n.
 \label{Abounded}
\end{align}
On ${\cal E}_n$ the third remainder on the $m$th layer satisfies
\begin{align*}
 r_n\sup_{q\in\widetilde K}|\Pi_{3,n}^{(m)}(q)|
 \leq {C A_{m,n}^2\over r_n(2^mb_{n,\tau})^2}
 \leq {CA_{m,n}\over2^mb_{n,\tau}},
\end{align*}
where the first inequality uses $H_{1,n}(I_{m,n}(q);q)\leq C\,2^mb_{n,\tau}$ together
with $\widetilde H_nH_n^2\geq(2^mb_{n,\tau})^3/2$ on the layer, and the second is
\eqref{Abounded}.

For the fourth remainder, put
\begin{align*}
 g_n(v;q)={(\widetilde H_n-H_n)(v;q)\over\widetilde H_n(v;q)H_n(v;q)}.
\end{align*}
One-dimensional integration by parts gives
\begin{align*}
 \Bigl|\int_{I_{m,n}(q)}g_n(v;q)
       (\widetilde H_{1,n}-H_{1,n})(dv;q)\Bigr|
 \leq {B_{m,n}\over r_n}
 \{2\sup_{v\in I_{m,n}(q)}|g_n(v;q)|
       +\mathop{\rm TV}_{I_{m,n}(q)}(g_n)\}.
\end{align*}
Moreover,
\begin{align*}
 \sup_{v\in I_{m,n}(q)}|g_n(v;q)|
 &\leq {CA_{m,n}\over r_n(2^mb_{n,\tau})^2},\\
 \mathop{\rm TV}_{I_{m,n}(q)}(g_n)
 &\leq {C\over2^mb_{n,\tau}}+
       {CA_{m,n}\over r_n(2^mb_{n,\tau})^2}.
\end{align*}
For the variation bound, both risk functions are nonincreasing.  On the relative-risk
event, their variations on the layer satisfy
\begin{align*}
 \mathop{\rm TV}(\widetilde H_n)
 \leq\mathop{\rm TV}(H_n)+2\sup|\widetilde H_n-H_n|\leq C\,2^mb_{n,\tau},
\end{align*}
uniformly in $m$ and $q$, while $(\widetilde H_nH_n)^{-1}$ is nondecreasing with
variation at most $C(2^mb_{n,\tau})^{-2}$.  The inequality
$\mathop{\rm TV}(fg)\leq\|f\|_\infty\mathop{\rm TV}(g)+\|g\|_\infty\mathop{\rm TV}(f)$,
applied to $g_n=(\widetilde H_n-H_n)\cdot(\widetilde H_nH_n)^{-1}$, gives the stated
result.  It
follows that, on ${\cal E}_n$ and using \eqref{Abounded} once more,
\begin{align*}
 r_n\sup_{q\in\widetilde K}|\Pi_{4,n}^{(m)}(q)|
 \leq {CB_{m,n}\over2^mb_{n,\tau}}+
       {CA_{m,n}B_{m,n}\over r_n(2^mb_{n,\tau})^2}
 \leq {CB_{m,n}\over2^mb_{n,\tau}}.
\end{align*}
Summing the two layer bounds over $m$ and taking outer expectations, the
first-moment bounds \eqref{ABmoments} give
\begin{align*}
 \E^*\Biggl[\sum_{m\geq0}\bigl\{
       {CA_{m,n}\over2^mb_{n,\tau}}
       +{CB_{m,n}\over2^mb_{n,\tau}}\bigr\}\Biggr]
 \leq C\sum_{m\geq0}{\sqrt{2^mb_{n,\tau}}\over2^mb_{n,\tau}}
 ={C\over\sqrt{b_{n,\tau}}}\sum_{m\geq0}2^{-m/2}
 \leq {C\over\sqrt{b_{n,\tau}}}.
\end{align*}
Since $\P({\cal E}_n)\to1$, Markov's inequality applied to the preceding display, and
using \eqref{btaudiverge}, prove that the two nonlinear lower-endpoint contributions
vanish uniformly after letting $n\to\infty$ and then $\tau\downarrow0$.

Combining this conclusion with \eqref{lowerlinear} proves
\begin{align}
 \lim_{\tau\downarrow0}\limsup_{n\to\infty}
 \P\Bigl(\sup_{q\in\widetilde K}\Bigl|
 r_n\{\widetilde\Lambda_n(q)-\log p_n(q)^{-1}\}
 -\mathbb L_{n,\tau}(q)\Bigr|>\varepsilon\Bigr)=0.
 \label{lowertogether}
\end{align}

Next we show that $\mathbb L_\tau$ converges, and identify its limit with the process
$\mathbb L$ of \eqref{Ldef}.  For $0<\tau'<\tau$ the compatibility relation
\eqref{compatible} gives, from \eqref{Ltau},
\begin{align*}
 \mathbb L_{\tau'}(q)-\mathbb L_\tau(q)
 =\int_{(\tau',\tau]}{\mathbb V_1^{\tau'}(dv;q)\over H(v;q)}
 -\int_{(\tau',\tau]}{\mathbb V_0(v;q)\over H(v;q)^2}H_1(dv;q),
\end{align*}
which is exactly the limiting form of $\mathbb L_{n,\tau'}-\mathbb L_{n,\tau}$, from the contribution of $(\tau',\tau]$ to $D_{n,\tau}$.  Taking the covariance limit
in \eqref{exactlower}, restricted to $(\tau',\tau]$,
gives the pointwise identity
\begin{align*}
 \E\bigl[(\mathbb L_\tau(q)-\mathbb L_{\tau'}(q))^2\bigr]
 =\int_{(\tau',\tau]}{H_1(dv;q)\over H(v;q)^2}\leq{1\over H(\tau;q)},
\end{align*}
the inequality by the domination of the failure measure used in \eqref{limitlower}.
For the uniform statement we do not need a separate Gaussian argument.  Since
$\mathbb L_{n,\tau}$ and $\mathbb L_{n,\tau'}$ are, by \eqref{Lntauconv} applied at
level $\tau'$ and the continuous mapping theorem, jointly weakly convergent to
$\mathbb L_\tau$ and $\mathbb L_{\tau'}$ in $\ell^\infty(\widetilde K)$, the
supremum norm is a continuous functional, and the portmanteau theorem for nonnegative
lower semicontinuous functionals gives
\begin{align*}
 \E^*[\|\mathbb L_\tau-\mathbb L_{\tau'}\|_{\widetilde K}]
 \leq\liminf_{n\to\infty}
 \E^*[\|\mathbb L_{n,\tau}-\mathbb L_{n,\tau'}\|_{\widetilde K}]
 \leq\liminf_{n\to\infty}{C\over\sqrt{b_{n,\tau}}}
 ={C\over\sqrt{b_\tau}},
 \qquad b_\tau=\inf_{q\in\widetilde K}H(\tau;q),
\end{align*}
where the middle inequality is the bound \eqref{layermaximal} summed over
$m\geq0$, applied to the layers contained in $(\tau',\tau]$, and is uniform in
$0<\tau'<\tau$.  Since $b_\tau\to\infty$ by \eqref{btaudiverge}, the bound just obtained shows that the
family $\{\mathbb L_\tau:\tau>0\}$ is Cauchy in probability.
Two further facts turn this into convergence to an actual limit.  First,
$\ell^\infty(\widetilde K)$ is complete under the sup norm.  Second, each $\mathbb
L_\tau$ is already tight, being the weak limit \eqref{Lntauconv} of $\mathbb
L_{n,\tau}$.  A Cauchy (in probability) family of tight random elements of a complete
space converges in outer probability to a tight limit \citep{vanderVaartWellner1996S}.
Defined, as here, on one common probability space through \eqref{compatible}, the
family therefore converges in outer probability to a tight limit.

That limit is $\mathbb L$.  Substituting $\mathbb V^\tau_a(v;q)=\mathbb V(\phi^\tau_a(v;q))$
into \eqref{Ltau} and evaluating the second integral, which contributes
$\{H(1;q)^{-1}-H(W(q);q)^{-1}\}$ on $\{\tau<W(q)\leq1\}$ and cancels the first term
there, gives $\mathbb L_\tau(q)=-\mathbb V(\xi_{q,\tau})$, where $\xi_{q,\tau}$ is
\eqref{scorelimit} with the integral restricted to $(\tau,1]$ and the indicator in the
second term replaced by $\ind_{\{\tau<W(q)\leq1\}}$.  The computation behind
\eqref{limitlower}, applied on $(0,\tau]$, gives
$\nu\{(\xi_q-\xi_{q,\tau})^2\}=\int_{(0,\tau]}H_1(dv;q)/H(v;q)^2$, which tends to zero
as $\tau\downarrow0$, so $\mathbb L_\tau(q)\to\mathbb L(q)$ in $L^2$ for every $q$.  A
tight limit in $\ell^\infty(\widetilde K)$ agreeing pointwise in $L^2$ with $\mathbb L$
is $\mathbb L$.  Together with \eqref{Lntauconv} and \eqref{lowertogether}, the
converging-together lemma concludes that
\begin{align}
 r_n\{\widetilde\Lambda_n-\log p_n^{-1}\}
 \leadsto\mathbb L
 \qquad\text{in }\ell^\infty(\widetilde K).
 \label{NAconv}
\end{align}
Since $(\mathbb V_0,\mathbb V_1^\tau)$ has a version with continuous sample paths, by
\eqref{processconv}, and $\mathbb L_\tau$ is built from it through the continuous
integral maps in \eqref{Ltau}, each $\mathbb L_\tau$ has a version with continuous sample paths on
$\widetilde K$; since $\mathbb L_\tau\to\mathbb L$ uniformly on $\widetilde K$, $\mathbb L$ also inherits a version with continuous sample paths.

\textit{Step 6: from Nelson--Aalen to Kaplan--Meier, and the ratio.}
The main identity that allows passing from the cumulative hazard to the product limit is \eqref{KMNA}, proved already in
the consistency argument with $K$ replaced by $\widetilde K$. On the event
$\{\inf_{q\in\widetilde K}Y^\ast_n(q)>2\}$,
\begin{align*}
 0\leq-\log\widetilde p_n(q)-\widetilde\Lambda_n(q)
 \leq {2\over Y_n^\ast(q)-2},
 \qquad Y_n^\ast(q)=\sum_{i=1}^n\ind_{\{W_{i,n}(q)\geq1\}}.
\end{align*}
Uniform convergence on $\widetilde K$, which follows from \eqref{processconv} at
$v=1$, gives
\begin{align*}
 {Y_n^\ast(q)\over r_n^2}
 \overset{\P}{\longrightarrow}H(1;q)\ \text{uniformly},
 \qquad\inf_{q\in\widetilde K}H(1;q)>0.
\end{align*}
Consequently,
\begin{align*}
 \sup_{q\in\widetilde K}r_n
 |\log\widetilde p_n(q)+\widetilde\Lambda_n(q)|=o_{\P}(1).
\end{align*}
Together with \eqref{NAconv}, we obtain
\begin{align}
 r_n\{\log\widetilde p_n-\log p_n\}
 \leadsto-\mathbb L
 \qquad\text{in }\ell^\infty(\widetilde K).
 \label{oracleCLT}
\end{align}

We now plug in the estimated standardizations.  Proposition~\ref{prop:stability}(ii)
gives $r_n\sup_{q\in\widetilde K}|\log\widehat p_n(q)-\log\widetilde p_n(q)|\overset{\P}\to0$,
which with \eqref{oracleCLT} gives
\begin{align}
 r_n\bigl\{\log\widehat p_n-\log p_n\bigr\}\leadsto-\mathbb L
 \qquad\text{in }\ell^\infty(\widetilde K).
 \label{plugCLT}
\end{align}

It remains to pass to the ratio.  By \eqref{bias2} and $\inf_KF^\circ>0$,
\begin{align*}
 \sup_{q\in K}r_n\bigl|\log{p_n(q)\over p_n(1,1)}
 -\log F^\circ(q)\bigr|\overset{\P}{\longrightarrow}0 ,
\end{align*}
so applying the continuous mapping theorem to \eqref{plugCLT} through the bounded
linear functional $z\mapsto z(\cdot)-z(1,1)$ from $\ell^\infty(\widetilde K)$ to
$\ell^\infty(K)$ gives
\begin{align*}
 r_n\bigl\{\log\widehat F_n^\circ-\log F^\circ\bigr\}
 \leadsto\mathbb L(1,1)-\mathbb L
 \qquad\text{in }\ell^\infty(K),
\end{align*}
the logarithms on the left being well defined with probability tending to one by
\eqref{plugCLT} and \eqref{oraclecons}.
Finally, to get rid of the logarithms, write $\Upsilon_n=\log\widehat F_n^\circ-\log F^\circ$, so that
$\widehat F_n^\circ-F^\circ=F^\circ(e^{\Upsilon_n}-1)$.  Since
$r_n\Upsilon_n=O_{\P}(1)$ uniformly on $K$ and $|e^x-1-x|\leq x^2e^{|x|}$, we get
\begin{align*}
 \sup_{q\in K}r_n\bigl|\widehat F_n^\circ-F^\circ-F^\circ\,\Upsilon_n\bigr|
 \leq\sup_{q\in K}F^\circ\cdot r_n\Upsilon_n^2e^{|\Upsilon_n|}
 =O_{\P}\bigl(r_n^{-1}\bigr)=o_{\P}(1),
\end{align*}
so $r_n(\widehat F_n^\circ-F^\circ)$ and $F^\circ r_n\Upsilon_n$ have the same weak
limit.
This yields $r_n(\widehat F_n^\circ-F^\circ)\leadsto F^\circ\mathbb G$, with
$\mathbb G$ as in \eqref{Glimit}. The proof is complete.

\subsection{Proof of Proposition~\ref*{prop:rateest}}
For $v>0$, write
\begin{align*}
 Y_n^{\rm ref}(v)=\sum_{i=1}^n
 \ind_{\{\psi_{1,n}(Z_i^{(1)})\geq v,\,
          \psi_{2,n}(Z_i^{(2)})\geq v\}}.
\end{align*}
Continuity of the margins and \eqref{ZisXC} give
\begin{align*}
 {\E[Y_n^{\rm ref}(v)]\over r_n^2}
 ={p_n(v,v)c_n(v,v)\over p_n(1,1)c_n(1,1)}.
\end{align*}
Let $v_n^{\rm hi}=(1-\eta_n)^{-1}$ and
$v_n^{\rm lo}=(1+\eta_n)^{-1}$.  Since both thresholds tend to one,
the locally uniform convergence in Assumption~\ref{a:FG} yields
\begin{align}
 {\E[Y_n^{\rm ref}(v_n^{\rm hi})]\over r_n^2}\longrightarrow1,
 \qquad
 {\E[Y_n^{\rm ref}(v_n^{\rm lo})]\over r_n^2}\longrightarrow1.
 \label{ratecountmean}
\end{align}
For $a\in\{{\rm hi},{\rm lo}\}$, $Y_n^{\rm ref}(v_n^a)$ is a sum of independent Bernoulli
variables, and hence, for either value of $a$,
\begin{align*}
 \operatorname{Var}\Bigl\{{Y_n^{\rm ref}(v_n^a)\over r_n^2}\Bigr\}
 \leq {\E[Y_n^{\rm ref}(v_n^a)]\over r_n^4}\longrightarrow0
\end{align*}
by Assumption~\ref{a:rate}, \eqref{ratecountmean}, and Chebyshev's inequality.  Consequently,
\begin{align}
 {{Y_n^{\rm ref}(v_n^{\rm hi})}\over r_n^2}\overset{\P}\longrightarrow1,
 \qquad
 {{Y_n^{\rm ref}(v_n^{\rm lo})}\over r_n^2}\overset{\P}\longrightarrow1.
 \label{ratecounts}
\end{align}

It remains to replace the true standardizations by their estimates.  Work on the
event in Assumption~\ref{a:plugin}, whose probability tends to one.  If
$\psi_{j,n}(Z_i^{(j)})\geq v_n^{\rm hi}$ and this value is at most $16T$, then
$\widehat\psi_{j,n}(Z_i^{(j)})\geq
(1-\eta_n)v_n^{\rm hi}=1$.  If it exceeds $16T$, the same conclusion follows from
\eqref{pluginabove} for all sufficiently large $n$.  Conversely, if
$\widehat\psi_{j,n}(Z_i^{(j)})\geq1$ and
$\psi_{j,n}(Z_i^{(j)})\leq16T$, then the same event gives
$\psi_{j,n}(Z_i^{(j)})\geq v_n^{\rm lo}$; if the latter value exceeds $16T$, this inequality
is automatic.  Applying these implications to both coordinates gives
\begin{align*}
 Y_n^{\rm ref}(v_n^{\rm hi})
 \leq \widehat r_n^2
 \leq Y_n^{\rm ref}(v_n^{\rm lo})
\end{align*}
with probability tending to one.  The squeeze theorem and \eqref{ratecounts} imply
$\widehat r_n^2/r_n^2\overset{\P}\longrightarrow1$.  Taking square roots proves the
claim.

\subsection{Proof of Proposition~\ref*{prop:variance}}
Put $q_0=(1,1)$. For fixed
$\tau\in(0,1)$, let $\widehat\xi_{i,n}^\tau(q)$ be \eqref{greenwoodscore} with the
integral restricted to $(\tau,1]$ and its second numerator replaced by
$n\widehat\Delta_{i,n}(q)
\ind_{\{\tau<\widehat W_{i,n}(q)\leq1\}}$.  Define
\begin{align*}
 \widehat V_n^\tau(q)
 ={\widehat r_n^2\over n(n-1)}\sum_{i=1}^n
 \bigl\{\widehat\xi_{i,n}^\tau(q)
       -\widehat\xi_{i,n}^\tau(q_0)\bigr\}^2.
\end{align*}

\emph{Step 1: a fixed lower endpoint.}
For $a=0,1$, introduce the estimated-sample indicators
\begin{align*}
 \widehat\phi_{0,i,n}(v;q)&=\ind_{\{\widehat W_{i,n}(q)\geq v\}},\\
 \widehat\phi_{1,i,n}^\tau(v;q)&=\widehat\Delta_{i,n}(q)
       \ind_{\{\tau<\widehat W_{i,n}(q)\leq v\}},
\end{align*}
and, with $\widehat\phi_{0,i,n}^\tau=\widehat\phi_{0,i,n}$, put
\begin{align*}
 \widehat\Sigma_{ab,n}^\tau((v,q),(v',q'))
 ={1\over\widehat r_n^2}\sum_{i=1}^n
 \widehat\phi_{a,i,n}^\tau(v;q)
 \widehat\phi_{b,i,n}^\tau(v';q').
\end{align*}
We claim that, for $a,b\in\{0,1\}$ and $q',q''\in\{q,q_0\}$,
\begin{align}
 \sup_{v,v'\in[\tau,1]}
 \bigl|\widehat\Sigma_{ab,n}^\tau((v,q'),(v',q''))
       -\Sigma_{ab}^\tau((v,q'),(v',q''))\bigr|
 \overset{\P}\longrightarrow0.
 \label{greenwoodcov}
\end{align}
For the oracle indicators, the corresponding statement with normalization $r_n^2$
follows from the covariance calculation in Step~3 of the proof of
Theorem~\ref{thm:normality}.  More explicitly, every product of two such indicators
belongs to one of the VC classes used there, and \eqref{maximalVC} bounds the supremum
of the centered, $r_n^{-2}$-normalized sums by $C_\tau/r_n=o(1)$.  Their expectations
converge uniformly to \eqref{Sigma}, by \eqref{admissible}.

It remains to compare the estimated and oracle indicators.  By Lemma~\ref{lem:sandwich}, a risk indicator at
$v\in[\tau,1]$ can change only when
\begin{align*}
 {v\over1+\eta_n}\leq W_{i,n}(q)\leq{v\over1-\eta_n}.
\end{align*}
Uniformly in $v\in[\tau,1]$ and $q'\in\{q,q_0\}$, the $\nu_n$-mass of this set tends
to zero by \eqref{Hn}, the locally uniform convergence \eqref{Hlimit}, and continuity
of $H$ on $[\tau/2,2]\times\widetilde K$.  A failure indicator can additionally change
when its mark flips.  By \eqref{flipset}, such observations belong to an angular band
of width $3\eta_n$, and, on the present interval, Lemma~\ref{lem:nearties} bounds its $\nu_n$-mass
by $C_\tau\eta_n=o(1)$.  The set and angular-band classes have the VC bounds used in
Steps~C--D of the proof of Proposition~\ref{prop:stability}; hence \eqref{maximalVC}
shows that their empirical $r_n^{-2}$-normalized masses differ from their expectations
by $O_{\P}(r_n^{-1})$.  It follows that, uniformly in $v\in[\tau,1]$,
\begin{align}
 {1\over r_n^2}\sum_{i=1}^n
 \bigl|\widehat\phi_{a,i,n}^\tau(v;q)-\phi_{a,i,n}^\tau(v;q)\bigr|
 \overset{\P}\longrightarrow0,\qquad a=0,1.
 \label{greenwoodindicators}
\end{align}
The difference of two products of indicators is bounded by the sum of the two
individual differences, so \eqref{greenwoodindicators}, Proposition~\ref{prop:rateest}
and the oracle result prove \eqref{greenwoodcov}.

For later use, define on $[\tau,1]$, for Borel sets $B\subset(\tau,1]$,
\begin{align*}
 \widehat H_n(v;q)={1\over\widehat r_n^2}\sum_{i=1}^n
       \widehat\phi_{0,i,n}(v;q),\quad \widehat H_{1,n}(B;q)={\widehat N_n(B;q)\over\widehat r_n^2}.
\end{align*}
Combining \eqref{greenwoodcov} and \eqref{Hlimit} gives, for
$q'\in\{q,q_0\}$, that $\widehat H_n\to H$ uniformly and that the Stieltjes measures
$\widehat H_{1,n}(dv;q)$ converge weakly to $H_1(dv;q)$.  In particular, their total
masses are bounded in probability and
\begin{align}
 \inf_{q'\in\{q,q_0\}}\inf_{v\in[\tau,1]}\widehat H_n(v;q')
 \overset{\P}\longrightarrow
 \inf_{q'\in\{q,q_0\}}H(1;q')>0.
 \label{greenwoodrisk}
\end{align}

Multiplying \eqref{greenwoodscore} by $\widehat r_n^2/n$ gives
\begin{align}
 {\widehat r_n^2\over n}\widehat\xi_{i,n}^\tau(q)
 ={}&\int_{(\tau,1]}
 {\widehat\phi_{0,i,n}(v;q)\over
  \widehat H_n(v;q)
  \{\widehat H_n(v;q)-\widehat H_{1,n}(\{v\};q)\}}
 \,\widehat H_{1,n}(dv;q)
 \nonumber\\
 &-{\widehat\Delta_{i,n}(q)
       \ind_{\{\tau<\widehat W_{i,n}(q)\leq1\}}\over
  \widehat H_n(\widehat W_{i,n}(q);q)
  -\widehat H_{1,n}(\{\widehat W_{i,n}(q)\};q)}.
 \label{greenwoodz}
\end{align}
Every failure multiplicity is at most two, by the injectivity condition on the
estimated standardizations and the argument preceding \eqref{KMNAhat}.  Hence
\eqref{greenwoodrisk} and $\widehat r_n\to\infty$ in probability show that the
exceptional event in the definition of $\widehat\sigma_n^2$ has probability tending
to zero, and that replacing each denominator in \eqref{greenwoodz} by
the same denominator without its failure-count term changes the displayed quantity by
$O_{\P}(\widehat r_n^{-2})$, uniformly in $i$ and $q'\in\{q,q_0\}$.  The quantities are
uniformly bounded on the observations for which they are nonzero, and their number is
$O_{\P}(\widehat r_n^2)$.

Expanding the product of two expressions in \eqref{greenwoodz}, summing over $i$ and
dividing by $\widehat r_n^2$ gives, up to the negligible denominator correction, the
following four terms.  For fixed $q,q'$, a failure coordinate written as $dv$ or
$dv'$ denotes the Stieltjes measure generated by that coordinate of
$\widehat\Sigma_{ab,n}^\tau$.
\begin{align*}
 &\int_{(\tau,1]}\int_{(\tau,1]}
 {\widehat\Sigma_{00,n}^\tau((v,q),(v',q'))\over
  \widehat H_n(v;q)^2\widehat H_n(v';q')^2}
 \,\widehat H_{1,n}(dv;q)\widehat H_{1,n}(dv';q'),\\
 &-\int_{(\tau,1]}\int_{(\tau,1]}
 {\widehat H_{1,n}(dv;q)\over
  \widehat H_n(v;q)^2\widehat H_n(v';q')}
 \,\widehat\Sigma_{01,n}^\tau((v,q),dv';q'),\\
 &-\int_{(\tau,1]}\int_{(\tau,1]}
 {\widehat H_{1,n}(dv';q')\over
  \widehat H_n(v;q)\widehat H_n(v';q')^2}
 \,\widehat\Sigma_{10,n}^\tau(dv;q,(v',q')),\\
 &\int_{(\tau,1]}\int_{(\tau,1]}
 {\widehat\Sigma_{11,n}^\tau(dv,dv';q,q')\over
  \widehat H_n(v;q)\widehat H_n(v';q')}.
\end{align*}
By \eqref{greenwoodcov} we obtain weak convergence of the corresponding finite Stieltjes
measures.  Together with \eqref{greenwoodrisk} and the weak convergence of
$\widehat H_{1,n}$, it therefore shows that the four terms converge to the expansion
of $\E[\mathbb L_\tau(q)\mathbb L_\tau(q')]$ obtained from \eqref{Ltau}.  Since
$\widehat V_n^\tau$ is $n/(n-1)$ times the corresponding quadratic combination for
$q',q''\in\{q,q_0\}$, it follows that
\begin{align}
 \widehat V_n^\tau(q)
 \overset{\P}\longrightarrow
 \E\Bigl[\bigl\{\mathbb L_\tau(q_0)-\mathbb L_\tau(q)\bigr\}^2\Bigr].
 \label{greenwoodtruncated}
\end{align}

\emph{Step 2: removal of the lower endpoint.}
Let $\widehat\xi_{i,n}^{0,\tau}(q)$ be \eqref{greenwoodscore} with its integral and
failure term restricted to $(0,\tau]$.  At an atom $v$ of
$\widehat N_n(\cdot\,;q)$, the contribution to
$\widehat\xi_{i,n}^{0,\tau}(q)/n$ is
\begin{align*}
 {\widehat N_n(\{v\};q)\ind_{\{\widehat W_{i,n}(q)\geq v\}}
  \over \widehat Y_n(v;q)
  \{\widehat Y_n(v;q)-\widehat N_n(\{v\};q)\}}
 -{\widehat\Delta_{i,n}(q)
       \ind_{\{\widehat W_{i,n}(q)=v\}}
  \over \widehat Y_n(v;q)-\widehat N_n(\{v\};q)}.
\end{align*}
The sum of this contribution over $i$ is zero.  Contributions from two distinct
atoms are orthogonal under summation over $i$. Indeed, on the smaller risk set the
contribution from the earlier failure time is constant, whereas the contribution from
the later failure time sums to zero.  The sum of squares at an atom $v$ equals
\begin{align*}
 {\widehat N_n(\{v\};q)\over
  \widehat Y_n(v;q)
  \{\widehat Y_n(v;q)-\widehat N_n(\{v\};q)\}}.
\end{align*}
Consequently, the exact identity
\begin{align}
 {1\over n^2}\sum_{i=1}^n
 \bigl\{\widehat\xi_{i,n}^{0,\tau}(q)\bigr\}^2
 =\int_{(0,\tau]}
 {\widehat N_n(dv;q)\over
  \widehat Y_n(v;q)
  \{\widehat Y_n(v;q)-\widehat N_n(\{v\};q)\}}
 \label{greenwoodidentity}
\end{align}
holds.  Since every atom of $\widehat N_n(\cdot\,;q)$ has mass at most two, the
argument of \eqref{ranksum}, with endpoint $\tau$, gives on the event
$\widehat Y_n(\tau;q)>4$,
\begin{align}
 \int_{(0,\tau]}
 {\widehat N_n(dv;q)\over
  \widehat Y_n(v;q)
  \{\widehat Y_n(v;q)-\widehat N_n(\{v\};q)\}}
 \leq2\int_{(0,\tau]}{\widehat N_n(dv;q)\over\widehat Y_n(v;q)^2}
 \leq {2\over\widehat Y_n(\tau;q)-2}.
 \label{greenwoodlower}
\end{align}
For each fixed $\tau$, Step~1 gives
$\widehat Y_n(\tau;q')/\widehat r_n^2\to H(\tau;q')$ in probability for
$q'\in\{q,q_0\}$.  Using $(x-y)^2\leq2x^2+2y^2$ together with
\eqref{greenwoodidentity}--\eqref{greenwoodlower}, we obtain
\begin{align}
 &{\widehat r_n^2\over n(n-1)}\sum_{i=1}^n
 \Bigl[\bigl\{\widehat\xi_{i,n}^{0,\tau}(q)
                  -\widehat\xi_{i,n}^{0,\tau}(q_0)\bigr\}^2\Bigr]
 \nonumber\\
 &\hspace{2cm}=O_{\P}\Bigl({1\over H(\tau;q)}+{1\over H(\tau;q_0)}\Bigr).
 \label{greenwoodlower2}
\end{align}
Both terms on the right tend to zero as $\tau\downarrow0$, by Assumption~\ref{a:divergence}.

The full score is the sum of its $(0,\tau]$ and $(\tau,1]$ parts.  Cauchy--Schwarz,
\eqref{greenwoodtruncated} and \eqref{greenwoodlower2} therefore yield, for every
$\varepsilon>0$,
\begin{align}
 \lim_{\tau\downarrow0}\limsup_{n\to\infty}
 \P\bigl(\bigl|\widehat\sigma_n^2(q)-\widehat V_n^\tau(q)\bigr|>
 \varepsilon\bigr)=0.
 \label{greenwoodtogether}
\end{align}
On the other hand, the exact limiting second-moment calculation following
\eqref{compatible} gives
\begin{align*}
 \E\bigl[\{\mathbb L(q')-\mathbb L_\tau(q')\}^2\bigr]
 =\int_{(0,\tau]}{H_1(dv;q')\over H(v;q')^2}
 \leq {1\over H(\tau;q')},
 \qquad q'\in\{q,q_0\}.
\end{align*}
Thus
\begin{align*}
 \E\Bigl[\bigl\{\mathbb L_\tau(q_0)-\mathbb L_\tau(q)\bigr\}^2\Bigr]
 \longrightarrow
 \E\Bigl[\bigl\{\mathbb L(q_0)-\mathbb L(q)\bigr\}^2\Bigr]
 =\E[\mathbb G(q)^2].
\end{align*}
Combining this limit with \eqref{greenwoodtruncated} and
\eqref{greenwoodtogether} proves the variance consistency.

Finally, Theorem~\ref{thm:normality} and the delta method give
$r_n\{\log\widehat F_n^\circ(q)-\log F^\circ(q)\}\leadsto\mathbb G(q)$.
Together, Proposition~\ref{prop:rateest} and the first part of the present proposition permit
replacement of $r_n$ and $\E[\mathbb G(q)^2]^{1/2}$ by their estimators.  Slutsky's
lemma proves the studentized limit and hence \eqref{greenwoodCI}.

\subsection{Proof of Proposition~\ref*{prop:adaptive}}
Since $A(k_1)=0\leq\kappa$, the count $\widehat k$ is well defined, and
$k_1\leq\widehat k\leq k_L$ gives the first two conclusions.  For each of the
finitely many $q\in{\cal Q}_0$, Theorem~\ref{thm:normality} and
Propositions~\ref{prop:rateest} and~\ref{prop:variance} give
$\widehat F^\circ_{n,k_1}(q)\overset{\P}\to F^\circ(q)>0$,
$\widehat r_{n,k_1}\overset{\P}\to\infty$ and
$\widehat\sigma^2_{n,k_1}(q)\overset{\P}\to\E[\mathbb G(q)^2]$, the last being positive
because $F^\circ(q)<1$, as noted in the remark following
Proposition~\ref{prop:variance}; hence $\widehat{\rm sd}_{n,k_1}(q)>0$ for every
$q\in{\cal Q}_0$ with probability tending to one.  Taking $l'=1$ in
\eqref{stabilityrule} then gives
\begin{align*}
 \max_{q\in{\cal Q}_0}
 \bigl|\widehat F^\circ_{n,\widehat k}(q)-F^\circ(q)\bigr|
 \leq
 \max_{q\in{\cal Q}_0}
 \bigl|\widehat F^\circ_{n,k_1}(q)-F^\circ(q)\bigr|
 +\kappa\max_{q\in{\cal Q}_0}\widehat{\rm sd}_{n,k_1}(q).
\end{align*}
The two terms on the right are $O_{\P}(r_n^{-1})$ by
Theorem~\ref{thm:normality} and Propositions~\ref{prop:rateest}
and~\ref{prop:variance}, respectively.

\newpage

\section{Stability of the directional product limit under coordinate
distortion}\label{supp:distortion}\label{sec:appendix}
Replacing $\psi_{j,n}$ by $\widehat\psi_{j,n}$ distorts each coordinate.  This
section describes that distortion, shows that it acts on the reduced sample in exactly
two ways, and proves the stability of the product limit under both. It aligns with Assumption~\ref{a:plugin} and provides an example on how to satisfy the stability result used in
Section~\ref{sec:proofs}.

It is convenient to describe the marginal replacement by its effect on a standardized axis.
Put
\begin{align}
 \varphi_{j,n}=\widehat\psi_{j,n}\circ\psi_{j,n}^{-1},
 \label{distortion}
\end{align}
a random nondecreasing map of the range of $\psi_{j,n}$ into that of
$\widehat\psi_{j,n}$.  The estimated sample is then exactly the oracle sample pushed
through $(\varphi_{1,n},\ldots,\varphi_{d,n})$, so the plug-in estimation component involves
 understanding how the product limit \eqref{KM} reacts to a distortion of
each coordinate.  Since $\psi_{j,n}$ and $\widehat\psi_{j,n}$ are both nondecreasing, on the
event in Assumption~\ref{a:plugin} every $x$ with $\psi_{j,n}(x)>16T$ satisfies
\begin{align}
 \widehat\psi_{j,n}(x)\ \geq\ \sup_{x':\ \psi_{j,n}(x')\leq16T}\widehat\psi_{j,n}(x')
 \ \geq\ 16T(1-\eta_n) ,
 \label{pluginabove}
\end{align}
which is the only property of the large values that is used.
Such a distortion affects the reduced sample in exactly two ways,
and they are of a different nature.

\emph{Shift.}  As long as $W_{i,n}(q)\leq2$, which by
Lemma~\ref{lem:sandwich} is the only range in which the two samples are compared,
$W_{i,n}(q)$ is multiplied by a
factor in $[1-\eta_n,1+\eta_n]$.  A failure time therefore moves slightly, and the risk
set is read at a displaced argument.  What controls this is a modulus of continuity for
$H_n$ along rays, and that is precisely Assumption~\ref{a:radialmodulus} with
$A=\theta_n:=(1+\eta_n)/(1-\eta_n)\downarrow1$.

\emph{Flip.}  The distortion may also change which coordinate attains the
minimum in $W_{i,n}(q)$.  By \eqref{Delta} such an exchange alters the mark only when
two coordinates have different censoring indicators, and it requires their
standardized values to be within a factor $\theta_n$ of each other, that is, the
observed angle to lie in a band of logarithmic width $O(\eta_n)$ around one of the
$\binom d2$ diagonals of the direction $q$.  On such a band the reduced failure and
censoring times are themselves within that factor of each other, so Assumption~\ref{a:radialmodulus}
bounds its mass through Lemma~\ref{lem:nearties}; no further condition is required.

For the multiplicative standardization of Section~\ref{sec:RV} the distortion is a pure
rescaling, $\widehat\psi_{j,n}=(u_{j,n}/\widehat u_{j,n})\psi_{j,n}$, so the estimated sample is
the oracle sample at the random direction $q_j\widehat u_{j,n}/u_{j,n}$, and it suffices
that the marginal quantiles be estimated at a rate faster than $r_n$, which is what
Assumption~\ref{a:pluginnorm} requires.  In general the argument uses
Lemma~\ref{lem:nearties} instead, and only through the following
consequence: integrating \eqref{bandmass} by parts and using
$-\!\int_{(0,2]}dH_n/H_n=\log\{H_n(0+;q)/H_n(2;q)\}$ gives
\begin{align}
& \nu_n\Biggl[{\ind\bigl\{|\log(A^{(1)}_{1,n}(q)/A^{(2)}_{1,n}(q))|
 \leq\eta,\ \delta^{(1)}_1\neq\delta^{(2)}_1,\ W_{1,n}(q)\leq2\bigr\}
 \over H_n(W_{1,n}(q);q)}\Biggr]\nonumber\\
 &\quad\leq C\,\eta\bigl\{1+\log{n\over r_n^2\,H_n(2;q)}\bigr\} .
 \label{angleint}
\end{align}
The logarithm on the right is $O\{\log(n/r_n^2)\}$ uniformly for $q$ in a compact set,
by Assumption~\ref{a:FG} and $H_n(0+;q)=n/r_n^2$.
The proposition below studies these two effects separately.

\begin{proposition}\label{prop:stability}
Assume Assumptions~\ref{ass:structure} and~\ref{a:plugin},
and let $K^\circ$ be a compact subset of $[1/4,4T]^d$.
\begin{itemize}
\item[\rm(i)] If $\eta_n\log(n/r_n^2)\to0$, then
 $\sup_{q\in K^\circ}|\log\widehat p_n(q)-\log\widetilde p_n(q)|\overset{\P}\to0$.
\item[\rm(ii)] If moreover $r_n\eta_n\log(n/r_n^2)\to0$, if Assumption~\ref{a:divergence} holds,
 and if the empirical processes indexed by the finite families in the
 $d$-dimensional failure decomposition described at the start of
 Section~\ref{sec:proofs} are asymptotically equicontinuous on $[\tau,2]\times K^\circ$ for
 every $\tau$, as in \eqref{aecsix}, then
 $r_n\sup_{q\in K^\circ}|\log\widehat p_n(q)-\log\widetilde p_n(q)|\overset{\P}\to0$.
\end{itemize}
\end{proposition}


Assumption~\ref{a:plugin} is required only below the level $16T$.
The following lemma is what converts it into the comparison of the two
samples that the proof uses; it is the only place where that condition is used.  Write
throughout $\theta_n=(1+\eta_n)/(1-\eta_n)\downarrow1$, and
$\widetilde Y_n(v;q)=\#\{i:W_{i,n}(q)\geq v\}=r_n^2\widetilde H_n(v;q)$.

\begin{lemma}\label{lem:sandwich}
Let $K^\circ\subset[1/4,4T]^2$ and let $\eta_n\leq1/8$.  On the event in
Assumption~\ref{a:plugin}, the following hold for every $i\leq n$ and every $q\in K^\circ$.
\begin{itemize}
\item[\rm(i)] If $W_{i,n}(q)\leq2$, then every coordinate attaining $W_{i,n}(q)$ or
 $\widehat W_{i,n}(q)$ has standardized value at most $16T$, and
 \begin{align*}
  (1-\eta_n)W_{i,n}(q)\ \leq\ \widehat W_{i,n}(q)\ \leq\ (1+\eta_n)W_{i,n}(q).
 \end{align*}
\item[\rm(ii)] If $W_{i,n}(q)>2$, then $\widehat W_{i,n}(q)\geq2(1-\eta_n)\geq7/4$.
\item[\rm(iii)] If $\widehat W_{i,n}(q)\leq1$, then $W_{i,n}(q)\leq8/7$, so that
 \rm{(i)} applies to $i$.
\item[\rm(iv)] For every $0<v\leq3/2$,
 \begin{align}
  \widetilde Y_n(\theta_nv;q)\ \leq\ \widehat Y_n\bigl((1+\eta_n)v;q\bigr),
  \qquad
  \widehat Y_n\bigl((1-\eta_n)v;q\bigr)\ \leq\ \widetilde Y_n(v/\theta_n;q).
  \label{Ynested}
 \end{align}
\end{itemize}
Since $\widehat Y_n(\cdot\,;q)$ is nonincreasing, \eqref{Ynested} implies in particular
$\widetilde Y_n(\theta_nv;q)\leq\widehat Y_n((1-\eta_n)v;q)$ and
$\widehat Y_n((1+\eta_n)v;q)\leq\widetilde Y_n(v/\theta_n;q)$.
\end{lemma}

\begin{proof}
Fix $i$ and $q=(q_1,q_2)\in K^\circ$ and abbreviate
$y_j=\psi_{j,n}(Z^{(j)}_i)$, $a_j=A^{(j)}_{i,n}(q)=y_j/q_j$ and
$\widehat y_j=\widehat\psi_{j,n}(Z^{(j)}_i)$,
$\widehat a_j=\widehat A^{(j)}_{i,n}(q)=\widehat y_j/q_j$, so that
$W_{i,n}(q)=a_1\wedge a_2$ and $\widehat W_{i,n}(q)=\widehat a_1\wedge\widehat a_2$.
Call the coordinate $j$ \emph{low} if $y_j\leq16T$ and \emph{high} otherwise.  For a low
$j$, the event in Assumption~\ref{a:plugin} applied at $x=Z^{(j)}_i$ gives
\begin{align}
 (1-\eta_n)a_j\ \leq\ \widehat a_j\ \leq\ (1+\eta_n)a_j .
 \label{lowbound}
\end{align}
For a high $j$ we have $a_j=y_j/q_j>16T/(4T)=4$, and by \eqref{pluginabove}
\begin{align}
 \widehat a_j\ \geq\ {16T(1-\eta_n)\over q_j}\ \geq\ {16T(1-\eta_n)\over4T}
 \ =\ 4(1-\eta_n)\ \geq\ {7\over2} ,
 \label{highbound}
\end{align}
both using $1/4\leq q_j\leq4T$ and $\eta_n\leq1/8$.

(i) Let $W_{i,n}(q)\leq2$ and let $j$ attain the minimum, $a_j=W_{i,n}(q)$.  Then
$y_j=q_ja_j\leq4T\cdot2=8T<16T$, so $j$ is low; in particular at least one coordinate is
low and $\min\{a_l:l\text{ low}\}=W_{i,n}(q)$.  By \eqref{lowbound},
\begin{align*}
 (1-\eta_n)W_{i,n}(q)\ \leq\ \min\{\widehat a_l:l\text{ low}\}
 \ \leq\ (1+\eta_n)W_{i,n}(q)\ \leq\ {9\over8}\cdot2={9\over4} ,
\end{align*}
while by \eqref{highbound} any high $l$ has $\widehat a_l\geq7/2>9/4$.  The minimum
over all coordinates is therefore the minimum over the low ones, which is the displayed
bound, and no high coordinate attains $\widehat W_{i,n}(q)$.

(ii) Let $W_{i,n}(q)>2$.  A low $l$ has $\widehat a_l\geq(1-\eta_n)a_l>2(1-\eta_n)$ by
\eqref{lowbound}, and a high $l$ has $\widehat a_l\geq7/2>2(1-\eta_n)$ by
\eqref{highbound}; taking the minimum gives $\widehat W_{i,n}(q)\geq2(1-\eta_n)\geq7/4$.

(iii) If $W_{i,n}(q)>2$ then $\widehat W_{i,n}(q)\geq7/4>1$ by (ii); so
$\widehat W_{i,n}(q)\leq1$ forces $W_{i,n}(q)\leq2$, and then (i) gives
$W_{i,n}(q)\leq\widehat W_{i,n}(q)/(1-\eta_n)\leq8/7$.

(iv) For the first inequality let $l$ satisfy $W_{l,n}(q)\geq\theta_nv$.  If
$W_{l,n}(q)\leq2$ then (i) gives
$\widehat W_{l,n}(q)\geq(1-\eta_n)\theta_nv=(1+\eta_n)v$; if
$W_{l,n}(q)>2$ then (ii) gives
$\widehat W_{l,n}(q)\geq7/4>(9/8)(3/2)\geq(1+\eta_n)v$.
Either way $\widehat W_{l,n}(q)\geq(1+\eta_n)v$, so every index counted on the left is
counted on the right.  For the second inequality let $l$ satisfy
$\widehat W_{l,n}(q)\geq(1-\eta_n)v$ and suppose $W_{l,n}(q)<v/\theta_n$.  Since
$\theta_n\geq1$ gives $v/\theta_n\leq v\leq3/2<2$, part (i) applies to $l$ and gives
$\widehat W_{l,n}(q)\leq(1+\eta_n)W_{l,n}(q)<(1+\eta_n)v/\theta_n=(1-\eta_n)v$,
a contradiction; hence $W_{l,n}(q)\geq v/\theta_n$.
\end{proof}

\begin{proof}[of Proposition~\ref{prop:stability}]
Work throughout on the event of Assumption~\ref{a:plugin}, whose probability tends to one, and
assume $n$ large enough that $\eta_n\leq1/8$.

\emph{Step A: the transfer to cumulative hazards.}  Lemma~\ref{lem:sandwich} holds.  Applying
\eqref{Ynested} at $v=1$ and using \eqref{Hlower} gives
$\widehat Y_n(1;q)\geq\widetilde Y_n(\theta_n;q)\geq\widetilde Y_n(2;q)
\geq r_n^2h_\tau/2$ with probability tending
to one, uniformly on $K^\circ$.  Moreover the failure times of the estimated sample are
distinct in the sense of \eqref{tietwo}, since if $\widehat\Delta_{i,n}(q)=1$ then, by
\eqref{Delta}, there is a $j$ with $\widehat A^{(j)}_{i,n}(q)=\widehat W_{i,n}(q)$ and
$\delta^{(j)}_i=1$, so an index $i$ with $\widehat\Delta_{i,n}(q)=1$ and
$\widehat W_{i,n}(q)=w$ satisfies $\widehat\psi_{j,n}(Z^{(j)}_i)=wq_j$ for some $j$
with $\delta^{(j)}_i=1$; since $\widehat\psi_{j,n}$ is injective on the uncensored
observations of the $j$th coordinate, at most one index does so for each of the two
values of $j$, whence at most two in total.  The argument of Step~4 of the proof of
Theorem~\ref{thm:consistency} therefore applies verbatim to the estimated sample and
gives
\begin{align}
 \sup_{q\in K^\circ}\bigl|\log\widehat p_n(q)+\widehat\Lambda_n(q)\bigr|
 \leq\sup_{q\in K^\circ}{2\over\widehat Y_n(1;q)-2}=O_{\P}\bigl(r_n^{-2}\bigr),
 \label{KMNAhat}
\end{align}
with $\widehat\Lambda_n(q)=\sum_{i:\widehat\Delta_{i,n}(q)=1,\ \widehat W_{i,n}(q)\leq1}
\widehat Y_n(\widehat W_{i,n}(q);q)^{-1}$ being the Nelson--Aalen integral of the estimated
sample.  The same statement for $\widetilde p_n$ and $\widetilde\Lambda_n$ is Step~4
itself.  It therefore suffices to compare $\widehat\Lambda_n$ with
$\widetilde\Lambda_n$.

\emph{Step B: decomposing into flip and shift terms.}  A mark can change only if the coordinate attaining
the minimum changes, which by \eqref{Delta} matters only when
$\delta^{(1)}_i\neq\delta^{(2)}_i$.  We claim that, when
$\log\theta_n\leq3\eta_n$ and $\eta_n\leq1/2$,
\begin{align}
 \bigl\{i:\widehat W_{i,n}(q)\leq1,\
 \widehat\Delta_{i,n}(q)\neq\Delta_{i,n}(q)\bigr\}\subset
 {\cal B}_n(q):=\bigl\{i:\bigl|\log{A^{(1)}_{i,n}(q)\over A^{(2)}_{i,n}(q)}\bigr|
 \leq3\eta_n,\ \delta^{(1)}_i\neq\delta^{(2)}_i\bigr\} ,
 \label{flipset}
\end{align}
Indeed, let $\widehat W_{i,n}(q)\leq1$, so that
$W_{i,n}(q)\leq8/7$ by Lemma~\ref{lem:sandwich}(iii).  If one coordinate, say the
second, were high in the sense of the proof of Lemma~\ref{lem:sandwich}, then
$A^{(2)}_{i,n}(q)>4>8/7\geq A^{(1)}_{i,n}(q)$ and, by \eqref{lowbound} and
\eqref{highbound}, $\widehat A^{(1)}_{i,n}(q)\leq(9/8)(8/7)=9/7<7/2\leq
\widehat A^{(2)}_{i,n}(q)$; the first coordinate would then attain both minima and
\eqref{Delta} would give $\widehat\Delta_{i,n}(q)=\delta^{(1)}_i=\Delta_{i,n}(q)$.  So
both coordinates are low, and \eqref{lowbound} holds for both, and an exchange of the
coordinate attaining the minimum forces
$A^{(1)}_{i,n}(q)/A^{(2)}_{i,n}(q)\in[\theta_n^{-1},\theta_n]$, which is
\eqref{flipset}.
Splitting according to whether the mark changes gives the exact decomposition
\begin{align}
 \widehat\Lambda_n(q)-\widetilde\Lambda_n(q)=\Phi_n(q)+\Psi_n(q),
 \label{Lamsplit}
\end{align}
\begin{align*}
 \Phi_n(q)&=\sum_{i\in{\cal B}_n(q)}
 {\widehat\Delta_{i,n}(q)-\Delta_{i,n}(q)\over
  \widehat Y_n(\widehat W_{i,n}(q);q)}\,\ind_{\{\widehat W_{i,n}(q)\leq1\}},\\
 \Psi_n(q)&=\sum_{i:\ \Delta_{i,n}(q)=1}\bigl\{
 {\ind_{\{\widehat W_{i,n}(q)\leq1\}}\over\widehat Y_n(\widehat W_{i,n}(q);q)}
 -{\ind_{\{W_{i,n}(q)\leq1\}}\over\widetilde Y_n(W_{i,n}(q);q)}\bigr\} ,
\end{align*}
the flip term and the shift term.  We bound them one by one.  Throughout we use the
relative-risk bound of Steps~2 and~3 of the proof of
Theorem~\ref{thm:consistency}; Step~3 gives
$\widetilde H_n\geq H_n/2$ on the layers below $\tau$, and Step~2 gives the uniform
convergence of $\widetilde H_n$ to $H_n$ above $\tau$, where $H_n$ is bounded below by
\eqref{Hlower}. These results are applied on
the enlarged set $[1,16T]^2$ and transported by the scaling identity
$\widetilde H_n(v;q)=\widetilde H_n(cv;q/c)$. Namely, outside an event of probability $o(1)$,
\begin{align}
 \sup_{q\in[1/4,4T]^2}\ \sup_{0<v\leq4}
 \bigl|{\widetilde H_n(v;q)\over H_n(v;q)}-1\bigr|\leq{1\over2},
 \label{relrisk0}
\end{align}
while Assumption~\ref{a:radialmodulus} at $A=\theta_n$ gives
\begin{align}
 H_n(\theta_nv;q)\geq\theta_n^{-\rho}H_n(v;q),\qquad 0<v\leq4 .
 \label{Hmodulus}
\end{align}

\emph{Remark:} Throughout Steps~C and~D below, $b_{n,\tau}=\inf_{q\in K^\circ}H_n(\tau;q)$, and $I_{m,n}(q)$
and $E_{m,n}$ are the layers and envelopes of Step~5 of the proof of
Theorem~\ref{thm:normality} formed with $K^\circ$ in place of $\widetilde K$; since
$K^\circ\subset[1/4,4T]^2$ gives $W_{1,n}(1,1)/(4T)\leq W_{1,n}(q)\leq4W_{1,n}(1,1)$, the
bound \eqref{layermass2} holds verbatim, with Assumption~\ref{a:radialmodulus} at $A=16T$.

\emph{Step C: the flip term.}  If $\widehat W_{i,n}(q)\leq1$ then $W_{i,n}(q)\leq8/7$
by Lemma~\ref{lem:sandwich}(iii), and
$\widehat W_{i,n}(q)\leq(1+\eta_n)W_{i,n}(q)$ by Lemma~\ref{lem:sandwich}(i), so
$\widehat Y_n(\widehat W_{i,n}(q);q)\geq\widehat Y_n((1+\eta_n)W_{i,n}(q);q)\geq
\widetilde Y_n(\theta_nW_{i,n}(q);q)$, the last step being the first inequality of
\eqref{Ynested} at $v=W_{i,n}(q)\leq8/7\leq3/2$; combining with
\eqref{relrisk0} and \eqref{Hmodulus},
\begin{align}
 \bigl|\Phi_n(q)\bigr|\ \leq\ {4\over r_n^2}
 \sum_{i\in{\cal B}_n(q),\ W_{i,n}(q)\leq2}{1\over H_n(W_{i,n}(q);q)} .
 \label{Phibound}
\end{align}

For $0\leq a<b\leq2$, define the empirical and population flip measures by
\begin{align*}
 \widetilde M_n((a,b];q)={1\over r_n^2}\sum_{i=1}^n
 \ind_{\{i\in{\cal B}_n(q),\ a<W_{i,n}(q)\leq b\}},\quad 
 M_n((a,b];q)=\E[\widetilde M_n((a,b];q)].
\end{align*}
Thus the right-hand side of \eqref{Phibound} is
$4\int_{(0,2]}H_n(v;q)^{-1}\widetilde M_n(dv;q)$.  Its expectation is four times
the left-hand side of \eqref{angleint} with $\eta=3\eta_n$, and is therefore at most
$C\eta_n\log(n/r_n^2)$ uniformly in $q$.  On each unequal censoring mark, the bands
defining $\widetilde M_n$ are finite intersections and differences of the six indicator
families in \eqref{fourclasses}; hence their centered empirical processes satisfy the
same bound \eqref{maximalVC}.

Split the centered integral at $\tau$.  On $I_{m,n}(q)$, write
$v_m^-=v_{m,n}^-(q)$ for its lower endpoint.
The envelope bound \eqref{layermass2} and \eqref{maximalVC} give
\begin{align*}
 \E^*\Bigl[\sup_{q\in K^\circ}\sup_{v\in I_{m,n}(q)}
 |\{\widetilde M_n-M_n\}((v_m^-,v];q)|\Bigr]
 \leq{C\sqrt{2^mb_{n,\tau}}\over r_n}.
\end{align*}
Integration by parts, using the monotonicity of $H_n^{-1}$, now gives
\begin{align*}
 \Bigl|\int_{I_{m,n}(q)}{(\widetilde M_n-M_n)(dv;q)\over H_n(v;q)}\Bigr|
 \leq {2\over2^mb_{n,\tau}}
 \sup_{v\in I_{m,n}(q)}|\{\widetilde M_n-M_n\}((v_m^-,v];q)|.
\end{align*}
Consequently the centered lower-interval contribution has outer expectation at most
\begin{align}
 \sum_{m\geq0}{2\over2^mb_{n,\tau}}\cdot
 {C\sqrt{2^mb_{n,\tau}}\over r_n}
 \leq{C\over r_n\sqrt{b_{n,\tau}}}.
 \label{layerdev}
\end{align}

On $[\tau,2]$, a second integration by parts
and \eqref{Hlower} yield
\begin{align*}
 \Bigl|\int_{(\tau,2]}{(\widetilde M_n-M_n)(dv;q)\over H_n(v;q)}\Bigr|
 \leq {2\over h_\tau}\sup_{\tau\leq v\leq2}
 |\{\widetilde M_n-M_n\}((\tau,v];q)|.
\end{align*}
In case~(i), \eqref{maximalVC} with the fixed envelope $E_\tau$ makes the last
supremum $O_{\P}(r_n^{-1})=o_{\P}(1)$ uniformly in $q$.  In case~(ii), each such band
is a finite sum of increments of ${\cal J}_1,{\cal J}_2,{\cal J}_1',{\cal J}_2'$ whose
two angular indices have, by \eqref{rhonu} and \eqref{bandmass}, distance at most
\begin{align*}
 \bigl\{3C\eta_nH_n(\tau;q)\bigr\}^{1/2}\longrightarrow0
\end{align*}
uniformly in $q\in K^\circ$ and $v\in[\tau,2]$.  The assumed asymptotic
equicontinuity \eqref{aecsix} therefore gives
\begin{align*}
 \sup_{q\in K^\circ}\sup_{\tau\leq v\leq2}
 |\{\widetilde M_n-M_n\}((\tau,v];q)|=o_{\P}(r_n^{-1}).
\end{align*}

Altogether, we get
\begin{align}
 \sup_{q\in K^\circ}|\Phi_n(q)|
 \leq C\eta_n\log{n\over r_n^2}+{C\over r_n\sqrt{b_{n,\tau}}}
 +\begin{cases} o_{\P}(1) &\text{in case (i)},\\
                o_{\P}(r_n^{-1}) &\text{in case (ii)},\end{cases}
 \label{Phifinal}
\end{align}
outside an event of vanishing probability.

\emph{Step D: the shift term.}  Fix $i$ with $\Delta_{i,n}(q)=1$ and write
$v=W_{i,n}(q)$.  If both $v\leq1$ and $\widehat W_{i,n}(q)\leq1$, then
Lemma~\ref{lem:sandwich}(i) gives
$(1-\eta_n)v\leq\widehat W_{i,n}(q)\leq(1+\eta_n)v$, so that, $\widehat Y_n(\cdot\,;q)$
being nonincreasing and \eqref{Ynested} applying at this $v$,
\begin{align*}
 \widetilde Y_n(\theta_nv;q)\leq\widehat Y_n\bigl((1+\eta_n)v;q\bigr)
 \leq\widehat Y_n(\widehat W_{i,n}(q);q)
 \leq\widehat Y_n\bigl((1-\eta_n)v;q\bigr)\leq\widetilde Y_n(v/\theta_n;q);
\end{align*}
both $\widehat Y_n(\widehat W_{i,n}(q);q)$ and $\widetilde Y_n(v;q)$ therefore
lie between $\widetilde Y_n(\theta_nv;q)$ and $\widetilde Y_n(v/\theta_n;q)$, so
$|a^{-1}-b^{-1}|=|a-b|/(ab)$ bounds the corresponding summand of $\Psi_n(q)$ by
\begin{align}
 {1\over r_n^2}\cdot
 {\widetilde H_n(v/\theta_n;q)-\widetilde H_n(\theta_nv;q)
 \over\widetilde H_n(\theta_nv;q)\,\widetilde H_n(v;q)} .
 \label{shiftsummand}
\end{align}
By \eqref{relrisk0} and \eqref{Hmodulus} the denominator is at least
$\tfrac18H_n(v;q)^2$ for $n$ large.  For the numerator write
\begin{align}
 &\widetilde H_n\bigl({v\over\theta_n};q\bigr)-\widetilde H_n(\theta_nv;q) \label{shiftsplit}\\
 &=\bigl\{H_n\bigl({v\over\theta_n};q\bigr)-H_n(\theta_nv;q)\bigr\}
 +\bigl\{\widetilde H_n-H_n\bigr\}\bigl({v\over\theta_n};q\bigr)
 -\bigl\{\widetilde H_n-H_n\bigr\}(\theta_nv;q).
\nonumber
\end{align}
By Assumption~\ref{a:radialmodulus} the first term in curly braces is at most
$(\theta_n^{2\rho}-1)H_n(\theta_nv;q)\leq C\rho(\theta_n-1)H_n(v;q)$ for $\theta_n$
close to one, so its contribution to $\sup_{q}|\Psi_n(q)|$ is at most
\begin{align*}
 C\rho\bigl(\theta_n-1\bigr){1\over r_n^2}
 \sum_{i:\Delta_{i,n}(q)=1,\ W_{i,n}(q)\leq2}{1\over H_n(W_{i,n}(q);q)} ,
\end{align*}
whose expectation equals
\begin{align*}
 C\rho(\theta_n-1)\int_{(0,2]}{H_{1,n}(dv;q)\over H_n(v;q)}.
\end{align*}
By \eqref{populationhazard}, that integral is
$-\log p_n(2s,2t)\leq\log(n/r_n^2)+C$
uniformly on $K^\circ$, and $\theta_n-1\leq3\eta_n$, so this contribution is at most
$C\eta_n\log(n/r_n^2)$ in expectation; its absolute deviation is bounded exactly as in
Step~C.

The two remaining terms of \eqref{shiftsplit} are increments of
$\widetilde H_n-H_n$ over $\{\theta_nv\leq W_{1,n}(q)<v/\theta_n\}$,
whose indicator is the difference of two risk indicators.  On $(0,\tau)$, summing over
the layers with the weight $H_n^{-2}$, the layer masses satisfy
$\nu_n\{W_{1,n}(q)\in I_{m,n}(q)\}\leq2^{m+1}b_{n,\tau}$ and together with \eqref{maximalVC} give
\begin{align*}
 \sum_{m\geq0}{2^{m+1}b_{n,\tau}\over(2^mb_{n,\tau})^2}
 \cdot{C\sqrt{2^mb_{n,\tau}}\over r_n}
 \leq{C\over r_n\sqrt{b_{n,\tau}}} ,
\end{align*}
again with $C$ a fixed constant.  On $[\tau,2]$ the two risk indicators are taken at indices
whose $\varrho_n$-distance is, by \eqref{rhonu} and Assumption~\ref{a:radialmodulus},
$\bigl[\nu_n\{\theta_nv\leq W_{1,n}(q)<v/\theta_n\}\bigr]^{1/2}$
$\leq\{(\theta_n^{2\rho}-1)H_n(\theta_nv;q)\}^{1/2}
\leq\{C\rho\eta_nH_n(\tau;q)\}^{1/2}\to0$ uniformly, so the same two cases as in
Step~C apply here, with the absolute deviation being $o_{\P}(1)$ in case~(i), by \eqref{maximalVC}, and
$o_{\P}(r_n^{-1})$ in case~(ii), by asymptotic equicontinuity.

Finally, the observations for which exactly one of $W_{i,n}(q)\leq1$ or
$\widehat W_{i,n}(q)\leq1$ holds have $W_{i,n}(q)\in[(1+\eta_n)^{-1},(1-\eta_n)^{-1}]$.  Indeed, if $W_{i,n}(q)\leq1<\widehat W_{i,n}(q)$ then
Lemma~\ref{lem:sandwich}(i) applies and gives
$1<\widehat W_{i,n}(q)\leq(1+\eta_n)W_{i,n}(q)$; and if
$\widehat W_{i,n}(q)\leq1<W_{i,n}(q)$ then Lemma~\ref{lem:sandwich}(iii) gives
$W_{i,n}(q)\leq8/7$, so that (i) applies and gives
$1\geq\widehat W_{i,n}(q)\geq(1-\eta_n)W_{i,n}(q)$.  Its $\nu_n$-mass is at most
$H_n((1+\eta_n)^{-1};q)-H_n((1-\eta_n)^{-1};q)\leq C\rho\,\eta_n$ by
Assumption~\ref{a:radialmodulus}, and each summand is at most $2\{r_n^2H_n(2;q)\}^{-1}\leq
Cr_n^{-2}$ by \eqref{Hlower}; so this contribution is at most $C\rho\eta_n$ plus some of the terms derived above.  Altogether $\Psi_n$ obeys the same bound
\eqref{Phifinal} as $\Phi_n$, with a possibly different fixed constant.

\emph{Conclusion.}  Combining \eqref{Lamsplit}, \eqref{Phifinal}, and its counterpart
for $\Psi_n$ with \eqref{KMNAhat} and the corresponding bound for $\widetilde p_n$
established in Step~4 of the proof of Theorem~\ref{thm:consistency}, we obtain
\begin{align*}
 \sup_{q\in K^\circ}\bigl|\log\widehat p_n(q)-\log\widetilde p_n(q)\bigr|
 \leq C\eta_n\log{n\over r_n^2}+{C\over r_n\sqrt{b_{n,\tau}}}
 +\begin{cases} o_{\P}(1),\\ o_{\P}(r_n^{-1}),\end{cases}
\end{align*}
outside an event of probability $o(1)$, with $C$ absolute and the two cases as in
\eqref{Phifinal}.  In case~(i) the first term vanishes by assumption, the second is
$o(1)$ because $r_n\to\infty$ and $b_{n,\tau}\to\inf_{q}H(\tau;q)>0$, and the third is
$o_{\P}(1)$; this is part~(i), and $\tau$ may be kept fixed throughout.  In case~(ii),
multiply by $r_n$; the first term vanishes by assumption, the third is $o_{\P}(1)$, and
the second becomes $C/\sqrt{b_{n,\tau}}$, whose limsup in $n$ is
$C/\sqrt{\inf_qH(\tau;q)}$ with $C$ fixed; letting $\tau\downarrow0$ and using
$\inf_{q\in K^\circ}H(\tau;q)\to\infty$, which is Assumption~\ref{a:divergence}, gives part~(ii).
\end{proof}

\newpage
\section{Asymptotic theory for the Kaplan--Meier standardization}\label{sec:appendixKM}

This section proves the crucial theory on which Section~\ref{sec:KM} rests, namely the
standardization \eqref{KMpsi} built from the Kaplan--Meier estimator satisfies
Assumption~\ref{a:plugin} at the rate $\zeta_n$.  Each coordinate is involved separately, so throughout
this section $j\in\{1,\ldots,d\}$ is fixed and suppressed.  Locally, we simplify notation and
write $X=X^{(j)}$, $C=C^{(j)}$,
$Z=X\wedge C$, $\delta=\ind_{\{X\leq C\}}$, $G=G_j$, $\beta=\beta_j$ and
$x^*=\sup\{x:F_X(x)<1\}$, and $(Z_i,\delta_i)$, $i\leq n$, are the observations of that
coordinate.  Recall $\psi_n:=\psi_{j,n}$ and $\widehat\psi_n:=\widehat\psi_{j,n}$ from
\eqref{KMpsi}, and set
\begin{align}
 \kappa_n=k\,G(n/k).
 \label{kappadef}
\end{align}

\begin{proposition}\label{prop:KM}
Let $X$ and $C$ be independent, assume that $G$ satisfies
Assumption~\ref{a:Cbeta}, let $F_X$ and $F_C$ be continuous with $F_X$
strictly increasing on $(0,x^*)$, let $k=k_n$ be an intermediate sequence, $k\to\infty$
and $k/n\to0$, and let $\kappa_n\to\infty$.  Then, for every
fixed $y^\ast>1$,
\begin{align*}
 \sup_{x\,:\ \psi_n(x)\leq y^\ast}
 \bigl|{\widehat\psi_n(x)\over\psi_n(x)}-1\bigr|=O_{\P}\bigl(\kappa_n^{-1/2}\bigr).
\end{align*}
\end{proposition}

The proof follows Steps~1--4 of the proof of
Theorem~\ref{thm:consistency} but with the direction $q$ absent and the radius replaced by
the standardized value $y=\psi_n(x)$; in particular it uses the same
Vapnik--Chervonenkis classes, the same maximal inequality, the same decomposition of a
hazard integral, the same layering near the divergence of the risk function, and the
same passage \eqref{ranksum}--\eqref{KMNA} from a product limit to its Nelson--Aalen
integral.

\begin{proof}
\emph{Step 1: the reduced sample and its population quantities.}
Since $F_X$ is continuous and strictly increasing on $(0,x^*)$, the map $\psi_n$ is
continuous and strictly increasing there, with range $[k/n,\infty)$, and
$Z\leq X<x^*$ almost surely, so $\psi_n(Z)$ is well defined.  Put
\begin{align*}
 \widetilde h_n(y)={1\over\kappa_n}\#\bigl\{i\leq n:\psi_n(Z_i)\geq y\bigr\},
 \qquad
 \widetilde h_{1,n}(B)={1\over\kappa_n}
 \#\bigl\{i\leq n:\psi_n(Z_i)\in B,\ \delta_i=1\bigr\},
\end{align*}
for $y>0$ and Borel $B\subset(0,\infty)$, and let $h_n=\E[\widetilde h_n]$ and
$h_{1,n}=\E[\widetilde h_{1,n}]$.  These are the one-dimensional counterparts of
\eqref{HH1}.

Because $F_X$ is continuous, $\overline F_X(X)$ is uniform on $(0,1)$, so
\begin{align}
 \P\bigl(\psi_n(X)>y\bigr)
={k/n\over y},\qquad y\geq k/n ;
 \label{exactpareto}
\end{align}
and so the standardized risk is exactly standard Pareto above the level $k/n$.  Next, $\overline F_X$ is
nonincreasing, so $\overline F_X(Z)=\overline F_X(X)\vee\overline F_X(C)$, and for
$0<a<1$ continuity of $F_X$ gives $\overline F_X(c)\leq a$ if and only if $c\geq U_X(1/a)$.  Hence,
by the independence of $X$ and $C$ and the definition $G=\overline F_C\circ U_X$ of
Assumption~\ref{a:Cbeta},
\begin{align*}
 \P\bigl(\overline F_X(Z)\leq a\bigr)
 =\P\bigl(\overline F_X(X)\leq a\bigr)\P\bigl(C\geq U_X(1/a)\bigr)
 =a\,G(1/a) .
\end{align*}
Taking $a=k/n/y$ and dividing by $\kappa_n/n$ gives
\begin{align}
 h_n(y)={G(ny/k)\over y\,G(n/k)},\qquad y>k/n .
 \label{hnKM}
\end{align}
On $\{\delta=1\}$ one has $Z=X$, so by independence again, and by
\eqref{exactpareto} and $\psi_n^{-1}(y)=U_X(ny/k)$,
\begin{align*}
 h_{1,n}(dy)={n\over\kappa_n}\,\P\bigl(C\geq U_X(ny/k)\bigr)
 \P\bigl(\psi_n(X)\in dy\bigr)
 ={n\over\kappa_n}\,G\bigl({ny\over k}\bigr){k/n\over y^2}\,dy
 =h_n(y)\,{dy\over y}
\end{align*}
for $y>k/n$, while $h_{1,n}$ puts no mass on $(0,k/n]$.  Consequently
\begin{align}
 \int_{(0,y]}{h_{1,n}(du)\over h_n(u)}=\int_{k/n}^y{du\over u}=\log{ny\over k}
 =-\log\overline F_X\bigl(\psi_n^{-1}(y)\bigr),\qquad y>k/n ,
 \label{hazardKM}
\end{align}
the last equality because $\overline F_X(\psi_n^{-1}(y))=k/n/y$.  This is the
one-dimensional form of \eqref{populationhazard}.

\emph{Step 2: layers and a maximal inequality.}
Fix $y^\ast>1$ and put
\begin{align*}
 b_n=h_n(y^\ast)={G(ny^\ast/k)\over y^\ast\,G(n/k)} .
\end{align*}
By Assumption~\ref{a:Cbeta},
$b_n\to (y^\ast)^{-1-\beta}>0$; fix $n$ large enough that $b_n\geq\tfrac12(y^\ast)^{-1-\beta}$.  By
\eqref{hnKM} the function $h_n$ is continuous and strictly decreasing on $(k/n,y^\ast]$,
with $h_n(y^\ast)=b_n$ and $h_n(y)\uparrow G(1)\,n/\kappa_n$ as $y\downarrow k/n$, a
bound that diverges with $n$ because $\kappa_n\leq k=o(n)$ and $G(1)>0$, $G$ being
regularly varying.  The sets
\begin{align*}
 I_m=\bigl\{y\in(k/n,y^\ast]:2^mb_n\leq h_n(y)<2^{m+1}b_n\bigr\},\qquad m\geq0,
\end{align*}
are therefore intervals, and they partition $(k/n,y^\ast]$; only finitely many of them
are nonempty, namely those with $2^mb_n\leq G(1)n/\kappa_n$, but the bounds below are
summed over all $m\geq0$ and no count of the layers is needed.  Write
$\alpha_m=\inf I_m$ for a nonempty $I_m$, and let
\begin{align*}
 E_m=\begin{cases}
  \bigl\{\psi_n(Z)\geq\alpha_m\bigr\} &\text{if }\alpha_m\in I_m,\\
  \bigl\{\psi_n(Z)>\alpha_m\bigr\} &\text{otherwise},
 \end{cases}
\end{align*}
the smallest set containing $\{\psi_n(Z)\geq y\}$ for every $y\in I_m$.  In the first
case $(n/\kappa_n)\P(E_m)=h_n(\alpha_m)<2^{m+1}b_n$ because $\alpha_m\in I_m$; in the
second, $h_n$ being continuous on $(k/n,y^\ast]$ and nonincreasing,
$(n/\kappa_n)\P(E_m)=\lim_{y\downarrow\alpha_m}h_n(y)\leq2^{m+1}b_n$.  Either way
\begin{align}
 \sup_{y\in I_m}h_n(y)\ \leq\ {n\over\kappa_n}\P(E_m)\ \leq\ 2^{m+1}b_n .
 \label{layersupKM}
\end{align}
The distinction matters only for the bottom layer, where $\alpha_m=k/n\notin I_m$ and
$(n/\kappa_n)\P(E_m)=G(1)n/\kappa_n$, which is smaller than the mass $n/\kappa_n$ of
$\{\psi_n(Z)\geq k/n\}$ whenever $G(1)<1$.

Let $O_i=(\psi_n(Z_i),\delta_i)$ and consider, for a nonempty $I_m$, the two families
of functions of $O_1$
\begin{align*}
 {\cal R}_m=\bigl\{\ind_{\{\psi_n(Z)\geq y\}}:y\in I_m\bigr\},
 \qquad
 {\cal K}_m=\bigl\{\ind_{\{a<\psi_n(Z)\leq y,\ \delta=1\}}:
 a\in\{\alpha_m\}\cup I_m,\ y\in I_m\bigr\} .
\end{align*}
The index $a=\alpha_m$ is included because the partial sums used below start at the
left endpoint of the layer, which need not belong to it.  Each member of either family
is the indicator of a set cut out by at most two
inequalities on the single real variable $\psi_n(Z)$, intersected in the second case
with the fixed set $\{\delta=1\}$; exactly as in the discussion preceding \eqref{VC},
these are Vapnik--Chervonenkis classes whose index is an absolute constant, so their
uniform covering numbers obey \eqref{VC}.  Every member of either family vanishes off
$E_m$, whose mass satisfies \eqref{layersupKM}.  We claim that, for
${\cal G}\in\{{\cal R}_m,{\cal K}_m\}$,
\begin{align}
 \E^*\Biggl[\sup_{g\in{\cal G}}\bigl|{1\over\kappa_n}\sum_{i=1}^ng(O_i)-{n\over\kappa_n}\E[g(O_1)]
 \bigr|\Biggr]\ \leq\ C\,{\sqrt{2^mb_n}\over\sqrt{\kappa_n}} ,
 \label{maximalKM}
\end{align}
with $C$ absolute.  The derivation is that of \eqref{maximalVC} with $r_n^2$ replaced
by $\kappa_n$; we provide the details for completeness. Writing $P_n$ for the law of $O_1$ and
$\mathbb G_n=\sqrt n(\mathbb P_n-P_n)$, the left-hand side equals
$(\sqrt n/\kappa_n)\E^*[\|\mathbb G_n\|_{\cal G}]$; the maximal inequality
\citep[Theorem~2.14.1]{vanderVaartWellner1996S} bounds $\E^*[\|\mathbb G_n\|_{\cal G}]$ by
$CJ(1,{\cal G})\|\ind_{E_m}\|_{P_n,2}$, the entropy integral being finite and absolute
by \eqref{VC}; and
$\|\ind_{E_m}\|_{P_n,2}=P_n(E_m)^{1/2}=\{\kappa_n(n/\kappa_n)\P(E_m)/n\}^{1/2}
\leq(\kappa_n2^{m+1}b_n/n)^{1/2}$.  Multiplying the two factors gives
\eqref{maximalKM}.

\emph{Step 3: the risk function.}
Applying \eqref{maximalKM} to ${\cal R}_m$ and dividing by
$\inf_{I_m}h_n\geq2^mb_n$, we get
\begin{align*}
 \E^*\Biggl[\sup_{y\in I_m}{\bigl|\widetilde h_n(y)-h_n(y)\bigr|\over h_n(y)}\Biggr]
 \ \leq\ {C\sqrt{2^mb_n}\over 2^mb_n\sqrt{\kappa_n}}
 \ =\ {C\,2^{-m/2}\over\sqrt{b_n\kappa_n}} .
\end{align*}
Since the $I_m$ partition $(k/n,y^\ast]$, summing the geometric series over $m\geq0$ gives
\begin{align}
 \E^*\Biggl[\sup_{k/n<y\leq y^\ast}
 {\bigl|\widetilde h_n(y)-h_n(y)\bigr|\over h_n(y)}\Biggr]\ \leq\ {C\over\sqrt{b_n\kappa_n}}
 \ \longrightarrow\ 0 ,
 \label{relriskKM}
\end{align}
because $b_n$ is bounded away from zero and $\kappa_n\to\infty$.  By Markov's
inequality the event
\begin{align*}
 {\cal E}_n=\Bigl\{\sup_{k/n<y\leq y^\ast}
 \bigl|{\widetilde h_n(y)\over h_n(y)}-1\bigr|\leq{1\over2}\Bigr\}
\end{align*}
has probability tending to one, and on it $\tfrac12h_n\leq\widetilde h_n\leq\tfrac32h_n$
on $(k/n,y^\ast]$.

\emph{Step 4: the hazard integral.}
Let $\widetilde\lambda_n(y)=\int_{(0,y]}\widetilde h_{1,n}(du)/\widetilde h_n(u)$.  Exactly as in Step~3 of the proof of
Theorem~\ref{thm:consistency}, and using \eqref{hazardKM},
\begin{align}
 \widetilde\lambda_n(y)-\log{ny\over k}
 =\int_{(k/n,y]}{\bigl(\widetilde h_{1,n}-h_{1,n}\bigr)(du)\over h_n(u)}
 +\int_{(k/n,y]}{h_n(u)-\widetilde h_n(u)\over\widetilde h_n(u)h_n(u)}\,
 \widetilde h_{1,n}(du) ,
 \label{splitKM}
\end{align}
both measures putting no mass on $(0,k/n]$.

For the first term of \eqref{splitKM}, fix $m$.  On $I_m$,
$h_n\geq2^mb_n$, while $u\mapsto h_n(u)^{-1}$ is nondecreasing and has total
variation at most $(2^mb_n)^{-1}$.  Integration by parts therefore gives
\begin{align*}
 \sup_{k/n<y\leq y^\ast}
 \Bigl|\int_{I_m\cap(k/n,y]}
 {\bigl(\widetilde h_{1,n}-h_{1,n}\bigr)(du)\over h_n(u)}\Bigr|
 &\leq {2\over2^mb_n}
 \sup_{u\in I_m}\bigl|(\widetilde h_{1,n}-h_{1,n})((\alpha_m,u])\bigr|.
\end{align*}
Applying \eqref{maximalKM} to ${\cal K}_m$ and summing over $m\geq0$,
\begin{align}
 \E^*\Biggl[\sup_{k/n<y\leq y^\ast}
 \Bigl|\int_{(k/n,y]}{\bigl(\widetilde h_{1,n}-h_{1,n}\bigr)(du)\over h_n(u)}\Bigr|\Biggr]
 \ \leq\ \sum_{m\geq0}{2\over2^mb_n}\cdot{C\sqrt{2^mb_n}\over\sqrt{\kappa_n}}
 \ \leq\ {C\over\sqrt{b_n\kappa_n}} .
 \label{firstKM}
\end{align}

For the second term of \eqref{splitKM}, work on ${\cal E}_n$, where
$\widetilde h_n\geq h_n/2$.  Every index counted in $\widetilde h_{1,n}(I_m)$ has
$\psi_n(Z_i)\in I_m$, hence is counted in $\widetilde h_n(y)$ for every
$y\in I_m$ with $y\leq\psi_n(Z_i)$; as $\widetilde h_n$ is nonincreasing this gives
$\widetilde h_{1,n}(I_m)\leq\sup_{y\in I_m}\widetilde h_n(y)
\leq\tfrac32\sup_{y\in I_m}h_n(y)\leq3\cdot2^mb_n$ by \eqref{layersupKM}.  Therefore,
on ${\cal E}_n$,
\begin{align*}
 \int_{I_m}{\bigl|h_n-\widetilde h_n\bigr|\over\widetilde h_nh_n}\,
 \widetilde h_{1,n}(du)
 \ \leq\ {2\sup_{I_m}\bigl|h_n-\widetilde h_n\bigr|\over(2^mb_n)^2}\,
 \widetilde h_{1,n}(I_m)
 \ \leq\ {6\over2^mb_n}\sup_{I_m}\bigl|h_n-\widetilde h_n\bigr| ,
\end{align*}
whose outer expectation is at most $C2^{-m/2}(b_n\kappa_n)^{-1/2}$ by
\eqref{maximalKM}; summing over $m\geq0$ and combining with \eqref{firstKM} and
\eqref{splitKM}, gives
\begin{align}
 \sup_{k/n<y\leq y^\ast}\bigl|\widetilde\lambda_n(y)-\log{ny\over k}\bigr|
 \ =\ O_{\P}\bigl({1\over\sqrt{b_n\kappa_n}}\bigr)
 \ =\ O_{\P}\bigl(\kappa_n^{-1/2}\bigr).
 \label{lambdaKM}
\end{align}

\emph{Step 5: from the hazard integral to the product limit.}
The Kaplan--Meier estimator of $\overline F_X$ is
\begin{align*}
 \widehat{\overline F}_n(x)=\prod_{i:\ Z_i\leq x,\ \delta_i=1}
 \bigl\{1-{1\over\#\{l\leq n:Z_l\geq Z_i\}}\bigr\} .
\end{align*}
Since $\psi_n$ is strictly increasing on $(0,x^*)$ and all $Z_i$ lie there almost
surely, $Z_l\geq Z_i$ if and only if $\psi_n(Z_l)\geq\psi_n(Z_i)$, so with
$y=\psi_n(x)$
\begin{align*}
 \widehat{\overline F}_n(x)=\prod_{i:\ \psi_n(Z_i)\leq y,\ \delta_i=1}
 \bigl\{1-{1\over\kappa_n\widetilde h_n(\psi_n(Z_i))}\bigr\} ,
\end{align*}
a product limit whose failure times are the values $\psi_n(Z_i)$ at uncensored $i$ and
whose Nelson--Aalen integral is $\widetilde\lambda_n(y)$.  Those values are almost
surely distinct, because $F_Z$ is continuous, so every failure multiplicity equals one;
and every risk count in the product is at least
$\kappa_n\widetilde h_n(y)\geq\kappa_n\widetilde h_n(y^\ast)$, because
$\widetilde h_n$ is nonincreasing and the product runs over
$\psi_n(Z_i)\leq y\leq y^\ast$.  These are the
only two properties used in the derivation of \eqref{ranksum} and \eqref{KMNA}, which
therefore give, on ${\cal E}_n\cap\{\kappa_n\widetilde h_n(y^\ast)>2\}$ and uniformly for
$y\in(k/n,y^\ast]$,
\begin{align}
 0\ \leq\ -\log\widehat{\overline F}_n(x)-\widetilde\lambda_n(y)
 \ \leq\ {2\over\kappa_n\widetilde h_n(y)-2}
 \ \leq\ {2\over\kappa_n\widetilde h_n(y^\ast)-2}
 \ \leq\ {4\over\kappa_nb_n-4} ,
 \label{KMNAmarg}
\end{align}
the last step because $\widetilde h_n(y^\ast)\geq h_n(y^\ast)/2=b_n/2$ on ${\cal E}_n$.  The
right-hand side is $O(\kappa_n^{-1})$, and $\kappa_n\widetilde h_n(y^\ast)>2$ with
probability tending to one because $\kappa_nb_n\to\infty$.  On the same event every
factor of the product limit is at least
$1-\{\kappa_n\widetilde h_n(y^\ast)\}^{-1}>0$, so $\widehat{\overline F}_n(x)>0$ and
its logarithm is finite.

\emph{Step 6: conclusion.}
Let $x$ satisfy $\psi_n(x)\leq y^\ast$ and put $y=\psi_n(x)$.  If $y=k/n$ then
$\overline F_X(x)=1$, no $Z_i$ with $\delta_i=1$ satisfies $Z_i\leq x$ almost surely,
so $\widehat{\overline F}_n(x)=1$ and the ratio in the proposition equals one.  If
$y>k/n$, then by \eqref{hazardKM}, \eqref{lambdaKM}, and \eqref{KMNAmarg},
\begin{align*}
 \bigl|\log{\widehat{\overline F}_n(x)\over\overline F_X(x)}\bigr|
 =\bigl|-\log\widehat{\overline F}_n(x)-\log{ny\over k}\bigr|
 \leq\bigl|\widetilde\lambda_n(y)-\log{ny\over k}\bigr|+{4\over\kappa_nb_n-4}
 =O_{\P}\bigl(\kappa_n^{-1/2}\bigr)
\end{align*}
uniformly in such $x$.  By \eqref{KMpsi},
$\widehat\psi_n(x)/\psi_n(x)=\overline F_X(x)/\widehat{\overline F}_n(x)$, and
$|e^{-z}-1|\leq e^{|z|}-1$ for real $z$, so the displayed bound transfers to
$|\widehat\psi_n(x)/\psi_n(x)-1|$.  The proof is complete.
\end{proof}

\newpage
\section{Proofs of the corollaries}\label{sec:corollaryproofs}
This section provides the proofs of the four corollaries stated in Section~\ref{supp:mda}.
\subsection{Proof of Corollary~\ref*{cor:RVcons}}
We verify Assumptions~\ref{ass:structure}, \ref{a:plugin} and~\ref{a:plugincons}, and
appeal to Theorem~\ref{thm:consistency}; Assumption~\ref{a:divergence} is verified here
as well, for use in Corollary~\ref{cor:RVnorm}.

\emph{Assumption~\ref{a:FG}.}  The first half follows from \eqref{pnlimit}, dividing by
its value at $(1,1)$, with $F^\circ$ as defined in \eqref{Fcirc}.  This
function is continuous and strictly positive by \eqref{copulalower}, and it is bounded
above by $(s^{-1/\gamma_{X^{(1)}}}\wedge t^{-1/\gamma_{X^{(2)}}})/R^X(1,1)$,
which tends to zero as $s\vee t\to\infty$.  For the censoring, write $q=(s,t)$ and
$x_{j,n}=(n/k)\overline F_{X^{(j)}}(q_ju_{j,n})$, so that $c_n(q)=\Gamma(n/k;x_n)$
exactly.  Assumption~\ref{a:RVX} gives $x_{j,n}\to q_j^{-1/\gamma_{X^{(j)}}}$ locally
uniformly, and the local uniformity in Assumption~\ref{a:RC} with the continuity of $R^C$
then gives the second half with
\begin{align}
 G^\circ(s,t)=R^C\bigl(s^{-1/\gamma_{X^{(1)}}},t^{-1/\gamma_{X^{(2)}}}\bigr) .
 \label{Gcirc}
\end{align}
It is continuous and strictly positive, and tends to zero as either argument tends
to infinity with the other fixed, because the corresponding coordinate of $x$ then tends
to zero in Assumption~\ref{a:RC}.
Writing $\gamma_{C^{(j)}}=\gamma_{X^{(j)}}/\beta_j$, the exponent is
$1/\gamma_{C^{(j)}}$, and Assumption~\ref{a:Cbeta} together with Assumption~\ref{a:RVX} is exactly regular
variation of $\overline F_{C^{(j)}}$ with index $-1/\gamma_{C^{(j)}}$.

\emph{Assumption~\ref{a:divergence}.}  For $(s,t)\in K'=[\kappa,\kappa']^2$
and $v\leq1/\kappa'$, both $vs$ and $vt$ lie in $(0,1]$, so by \eqref{copulalower}
and, for the censoring, by the monotonicity of $R^C$ and \eqref{homogC},
\begin{align*}
 F^\circ(vs,vt)G^\circ(vs,vt)
 \geq(v\kappa')^{-\min_j1/\gamma_{X^{(j)}}}
 (v\kappa')^{-\beta^\ast\min_j1/\gamma_{X^{(j)}}}=(\kappa')^{-\beta}v^{-\beta},
\end{align*}
with $\beta=(1+\beta^\ast)\min_j\gamma_{X^{(j)}}^{-1}>0$, which diverges as $v\downarrow0$.

\emph{Assumption~\ref{a:rate}.}  By the definition of $r_n$,
\begin{align}
 r_n^2=\bigl\{{n\over k}p_n(1,1)\bigr\}\,k\,c_n(1,1),
 \label{rnRV}
\end{align}
whose first factor converges to $R^X(1,1)>0$ by \eqref{pnlimit}.  Hence Assumption~\ref{a:rate}
holds because Assumption~\ref{a:effective} does, which is assumed.

\emph{Assumption~\ref{a:radialmodulus}.}  By \eqref{ZisXC} and \eqref{RVpsi},
$p_n(vq)c_n(vq)=\P\bigl(Z^{(j)}>vq_ju_{j,n},\ j=1,\ldots,d\bigr)$.  Put $z_j=vq_ju_{j,n}$;
replacing $v$ by $Av$ replaces $z$ by $Az$, so Assumption~\ref{a:radialRV} is
Assumption~\ref{a:radialmodulus}.

\emph{Assumption~\ref{a:plugin}.}  Here the multiplicative
form of \eqref{RVpsi} is rather convenient. The distortion
\eqref{distortion} is the linear map
\begin{align}
 \varphi_{j,n}(y)={u_{j,n}\over\widehat u_{j,n}}\ y ,
 \label{RVratio}
\end{align}
whose ratio to the identity does not depend on $y$.  Here
$\widehat\psi_{j,n}(x)/\psi_{j,n}(x)=u_{j,n}/\widehat u_{j,n}$ for every $x$, so the
restriction to $\{x:\psi_{j,n}(x)\leq16T\}$ in Assumption~\ref{a:plugin} costs nothing and
Assumption~\ref{a:plugin} holds
with any deterministic $\eta_n\downarrow0$ for which
$\P\bigl(\max_j|u_{j,n}/\widehat u_{j,n}-1|\leq\eta_n\bigr)\to1$.  Also $\widehat\psi_{j,n}$ is
strictly increasing, hence injective on the whole sample, as
Section~\ref{sec:plugin} requires.  More importantly, the
identity
\begin{align}
 \widehat W_{i,n}(s,t)=W_{i,n}\bigl(s{\widehat u_{1,n}\over u_{1,n}},\,
 t{\widehat u_{2,n}\over u_{2,n}}\bigr),\qquad
 \widehat\Delta_{i,n}(s,t)=\Delta_{i,n}\bigl(s{\widehat u_{1,n}\over u_{1,n}},\,
 t{\widehat u_{2,n}\over u_{2,n}}\bigr)
 \label{randomindex}
\end{align}
holds exactly for every $n$ and every realization, so that the estimated sample is the
oracle sample at the random direction
\begin{align*}
 Q_n(s,t)=\bigl(s{\widehat u_{1,n}\over u_{1,n}},
 t{\widehat u_{2,n}\over u_{2,n}}\bigr),
\end{align*}
and
$\sup_q\|Q_n(q)-q\|\leq C\max_j|\widehat u_{j,n}/u_{j,n}-1|$ on compact direction sets,
which tends to zero in probability by Assumption~\ref{a:quantilecons}.  No angular
condition is needed.

\emph{Assumption~\ref{a:plugincons}.}  Take $\eta_n\downarrow0$ with
$\eta_n\log(n/k)\to0$ and
$\P\bigl(\max_j|u_{j,n}/\widehat u_{j,n}-1|\leq\eta_n\bigr)\to1$, which is possible by
Assumption~\ref{a:quantilecons}.  By \eqref{pnlimit} and \eqref{rnRV},
$n/r_n^2=\{p_n(1,1)c_n(1,1)\}^{-1}$ satisfies
$\log(n/r_n^2)=O\{\log(n/k)\}$, because $p_n(1,1)$ is of order $(k/n)R^X(\boldsymbol 1)$ by
\eqref{pnlimit} and $c_n(1,1)$ is regularly varying with index $-\beta^\ast$ in $n/k$ by
Assumption~\ref{a:RC}, so $\eta_n\log(n/r_n^2)\to0$.
 
 In conclusion, Theorem~\ref{thm:consistency} applies.

\subsection{Proof of Corollary~\ref*{cor:RVnorm}}
Assumptions~\ref{ass:structure}, \ref{a:divergence}, \ref{a:plugin}
and~\ref{a:plugincons} were verified in the proof of
Corollary~\ref{cor:RVcons}, Assumption~\ref{a:plugin} now with the $\eta_n$ constructed below.
It remains to verify Assumptions~\ref{a:bias} and~\ref{a:pluginnorm}.

\emph{Assumption~\ref{a:bias}.}  Continuity of the margins gives the exact identity
 \begin{align*}
 {p_n(s,t)\over p_n(1,1)}
 ={R^X_{n/k}\bigl(\omega_{1,n}(s),\omega_{2,n}(t)\bigr)\over R^X_{n/k}(1,1)},
 \qquad
 \omega_{j,n}(s)={\overline F_{X^{(j)}}(su_{j,n})\over
 \overline F_{X^{(j)}}(u_{j,n})} .
\end{align*}
By Assumption~\ref{a:Hall},
$\omega_{j,n}(s)=s^{-1/\gamma_{X^{(j)}}}[1+\{\delta_{X^{(j)}}(su_{j,n})
-\delta_{X^{(j)}}(u_{j,n})\}/\{\gamma_{X^{(j)}}+\delta_{X^{(j)}}(u_{j,n})\}]$, and
the uniform convergence theorem for the regularly varying $|\delta_{X^{(j)}}|$ gives
$|\delta_{X^{(j)}}(su_{j,n})|\leq C|\delta_{X^{(j)}}(u_{j,n})|$ uniformly for $s$ in a
compact subset of $(0,\infty)$; since
$\gamma_{X^{(j)}}+\delta_{X^{(j)}}(u_{j,n})\to\gamma_{X^{(j)}}>0$,
\begin{align*}
 \sup_{s}\bigl|\omega_{j,n}(s)-s^{-1/\gamma_{X^{(j)}}}\bigr|
 \leq C\bigl|\delta_{X^{(j)}}(u_{j,n})\bigr| .
\end{align*}
As $R^X_r$ and $R^X$ are Lipschitz with constant one in each argument, two applications
of the triangle inequality give, uniformly for $q$ in a compact subset of
$(0,\infty)^2$,
\begin{align*}
 \bigl|R^X_{n/k}\bigl(\omega_{1,n}(s),\omega_{2,n}(t)\bigr)
 -R^X\bigl(s^{-1/\gamma_{X^{(1)}}},t^{-1/\gamma_{X^{(2)}}}\bigr)\bigr|
 \leq C\bigl\{\alpha(n/k)+\max_j\bigl|\delta_{X^{(j)}}(u_{j,n})\bigr|\bigr\},
\end{align*}
the first term by Assumption~\ref{a:S} and the second by the previous inequality.  The same bound at
$(s,t)=(1,1)$ applies to the denominators, which converge to $R^X(1,1)>0$, so
\begin{align*}
 \sup_{q\in\widetilde K}r_n\bigl|{p_n(q)\over p_n(1,1)}-F^\circ(q)\bigr|
 \leq Cr_n\bigl\{\alpha(n/k)+\max_j\bigl|\delta_{X^{(j)}}(u_{j,n})\bigr|\bigr\} .
\end{align*}
The first term vanishes by Assumption~\ref{a:rateXR}.  For the second, $\overline
F_{Z^{(j)}}=\overline F_{X^{(j)}}\overline F_{C^{(j)}}\leq\overline F_{X^{(j)}}$, so
$u_{j,n}\geq U_{Z^{(j)}}(n/k)$, and since $|\delta_{X^{(j)}}|$ is regularly varying with
negative index, Potter's inequality gives
$|\delta_{X^{(j)}}(u_{j,n})|\leq C|\delta_{X^{(j)}}(U_{Z^{(j)}}(n/k))|$.  Moreover, by
\eqref{rnRV} and $\overline F_{X^{(j)}}(u_{j,n})=k/n$,
\begin{align}
 {r_n\over\sqrt k}=\big\{R^X_{n/k}(1,1)\,c_n(1,1)\big\}^{1/2}\longrightarrow0 ,
 \label{rksqrtk}
\end{align}
so that $r_n|\delta_{X^{(j)}}(u_{j,n})|=o(1)$ by Assumption~\ref{a:rateXR}.  This is Assumption~\ref{a:bias}.

\emph{Assumption~\ref{a:pluginnorm}.}  By Assumption~\ref{a:RC} and Potter's bound
we get, for $0<\varepsilon<\beta^\ast$,
\begin{align*}
 c_n(1,1)=O\bigl\{(n/k)^{-\beta^\ast+\varepsilon}\bigr\} .
\end{align*}
Together with \eqref{rksqrtk}, we further get
\begin{align}
 {r_n\log^{3}(n/k)\over\sqrt k}\longrightarrow0 ,
 \label{tworates}
\end{align}
the power of the logarithm being immaterial because $r_n/\sqrt k$ decays like a
negative power of $n/k$.
By Assumption~\ref{a:quantilerate} there is a deterministic $a_n\to\infty$, which by
\eqref{tworates} may be taken to increase slowly enough that
$r_na_n\log^3(n/k)/\sqrt k\to0$, such that
$\P\bigl(\max_j|u_{j,n}/\widehat u_{j,n}-1|\leq\eta_n\bigr)\to1$ for
$\eta_n=a_n\log^2(n/k)/\sqrt k$; thus we have
$r_n\eta_n\log(n/r_n^2)\leq Cr_na_n\log^3(n/k)/\sqrt k\to0$, which is
Assumption~\ref{a:pluginnorm}.

In conclusion, Theorem~\ref{thm:normality} applies.

\subsection{Proof of Corollary~\ref*{cor:KMcons}}
We verify Assumptions~\ref{ass:structure} and \ref{ass:plugin}.1--\ref{ass:plugin}.2, and
appeal to Theorem~\ref{thm:consistency}; Assumption~\ref{a:divergence} is verified here
as well, for use in Corollary~\ref{cor:KMnorm}.

\emph{Assumption~\ref{a:FG}.}  The first half is \eqref{KMF}, and the limit $F^\circ$ is
continuous, strictly positive by \eqref{copulalower}, and bounded by $(s^{-1}\wedge t^{-1})/R^X(1,1)$, which tends
to zero as $s\vee t\to\infty$.  For the censoring, $\psi_{j,n}(C^{(j)})>s_j$ means
$\overline F_{X^{(j)}}(C^{(j)})\leq(k/n)/s_j$, so
\begin{align}
 c_n(s,t)=\Gamma\bigl(n/k;\,1/s,1/t\bigr),
 \label{cnKM}
\end{align}
and Assumption~\ref{a:RC} gives, locally uniformly,
\begin{align*}
 {c_n(s,t)\over c_n(1,1)}\longrightarrow
 G^\circ(s,t)=R^C(1/s,1/t),
\end{align*}
which is continuous, strictly positive, and tends to zero as either argument tends
to infinity with the other fixed, by Assumption~\ref{a:RC}.

\emph{Assumption~\ref{a:divergence}.}  Both limits are exactly homogeneous, the
second by \eqref{homogC}, so
\begin{align}
 F^\circ(vs,vt)\,G^\circ(vs,vt)
 =v^{-(1+\beta^\ast)}F^\circ(s,t)\,G^\circ(s,t),\qquad v>0 ,
 \label{KMpower}
\end{align}
and the divergence follows because $F^\circ G^\circ$ is bounded away from zero on
the compact $K'$.

\emph{Assumption~\ref{a:rate}.}  By \eqref{pnKM} and \eqref{cnKM},
$r_n^2=n\,p_n(1,1)c_n(1,1)=k\,R^X_{n/k}(1,1)\,c_n(1,1)$, whose first factor
converges to $R^X(1,1)>0$ by Assumption~\ref{a:RX}.  Hence Assumption~\ref{a:rate} holds because of
Assumption~\ref{a:effective}, which is assumed.

\emph{Assumption~\ref{a:radialmodulus}.}  By \eqref{ZisXC} and \eqref{KMpsi},
$p_n(vq)c_n(vq)=\P\bigl(\overline F_{X^{(j)}}(Z^{(j)})<c_ju,\ j=1,\ldots,d\bigr)$ with
$u=(k/n)/v$ and $c_j=1/q_j$; thus $c\in[1/(16T),4]^d$ for $q\in[1/4,16T]^d$, and
replacing $v$ by $Av$ replaces $u$ by $u/A$.  Assumption~\ref{a:radialKM} is therefore
Assumption~\ref{a:radialmodulus}.

\emph{Assumption~\ref{a:plugin}.}  Proposition~\ref{prop:KM} applies to each coordinate. Indeed,
its hypotheses are the independence of $X^{(j)}$ and $C^{(j)}$,
Assumption~\ref{a:Cbeta}, the continuity of the corresponding two marginal
distribution functions, Assumption~\ref{a:Mstrict}, and $kG_j(n/k)\to\infty$, which follows from
Assumption~\ref{a:effective}.  This yields \eqref{KMrate}, and hence the claim with
$\eta_n=a_n\zeta_n$ for any deterministic $a_n\to\infty$ chosen to increase slowly enough,
as in the proof of Corollary~\ref{cor:RVnorm}.

\emph{Assumption~\ref{a:plugincons}.}  Since
$\zeta_n^2=\max_j\{kG_j(n/k)\}^{-1}=\{k\min_jG_j(n/k)\}^{-1}$ and
$c_n(1,1)\leq\min_jG_j(n/k)$, we get
\begin{align}
 \zeta_n\log{n\over k}\ \leq\ {\log(n/k)\over\{k\,c_n(1,1)\}^{1/2}}\longrightarrow0
 \label{rnzeta}
\end{align}
by Assumption~\ref{a:effective}.  Furthermore, by
\eqref{pnKM} and \eqref{cnKM},
$n/r_n^2=(n/k)\{R^X_{n/k}(1,1)c_n(1,1)\}^{-1}$, so
$\log(n/r_n^2)\leq C\log(n/k)$ for $n$ large, because $c_n(1,1)$ is regularly varying
with index $-\beta^\ast$ in $n/k$ by Assumption~\ref{a:RC}. Hence the deterministic sequence $a_n\to\infty$ of the
previous step may be taken to increase slowly enough such that
$\eta_n\log(n/r_n^2)\leq Ca_n\zeta_n\log(n/k)\to0$.

In conclusion, Theorem~\ref{thm:consistency} applies.

\subsection{Proof of Corollary~\ref*{cor:KMnorm}}
Assumptions~\ref{ass:structure}, \ref{a:plugin}, \ref{a:plugincons}
and~\ref{a:divergence} were verified in the proof of
Corollary~\ref{cor:KMcons}.  It remains to verify Assumption~\ref{a:bias} and to
decide which of the two conclusions applies.

\emph{Assumption~\ref{a:pluginnorm}.}  With $\eta_n=a_n\zeta_n$ and
$\log(n/r_n^2)\leq C\log(n/k)$ as in the proof of Corollary~\ref{cor:KMcons},
$r_n\eta_n\log(n/r_n^2)\leq Ca_nr_n\zeta_n\log(n/k)$, while \eqref{freeplugin} and
Assumptions~\ref{a:Cbeta} and~\ref{a:RC} give
\begin{align*}
 r_n^2\zeta_n^2=R^X_{n/k}(1,1)\,{c_n(1,1)\over\min_jG_j(n/k)}
 =(n/k)^{-(\beta^\ast-\max_j\beta_j)+o(1)}\,R^X_{n/k}(1,1).
\end{align*}
Under $\beta^\ast>\max_j\beta_j$ the right-hand side tends to zero faster than any
power of $\log(n/k)$, so $a_n\to\infty$ may be taken to increase slowly enough that
$r_n\eta_n\log(n/r_n^2)\to0$.  This is Assumption~\ref{a:pluginnorm}, so
Theorem~\ref{thm:normality} applies and gives the first conclusion.  By
Remark~\ref{rem:asmS1} the strict inequality holds when the coordinates of $C$ are
independent and at least two censoring indices are positive.  Without it one has
$\beta^\ast=\max_j\beta_j$, the marginal standardizations are estimated at the rate $r_n$
itself, and the second conclusion, via Remark~\ref{rem:regular}, is the one that
applies.

\emph{Assumption~\ref{a:bias}.}  Let $\delta\in(0,1/2)$ be arbitrary and
$\widetilde K=[1-\delta,T+\delta]^2$, as in the proof of Theorem~\ref{thm:normality}; since
Assumption~\ref{a:secondKM} is assumed for every $0<a<b<\infty$, the argument applies verbatim to any
compact $\widetilde K\subset(0,\infty)^2$.  Abbreviate $A_n=R^X_{n/k}(1/s,1/t)$,
$A=R^X(1/s,1/t)$, $B_n=R^X_{n/k}(1,1)$, and $B=R^X(1,1)>0$, and let
$\varepsilon_n$ denote the supremum in Assumption~\ref{a:secondKM} taken over
$[a,b]=[1/(T+\delta),1/(1-\delta)]$, which contains $1$ and every argument that occurs,
since $q=(s,t)$ ranges over $\widetilde K=[1-\delta,T+\delta]^2$.  By \eqref{pnKM},
\begin{align*}
 \bigl|{p_n(s,t)\over p_n(1,1)}-F^\circ(s,t)\bigr|
 =\bigl|{A_n\over B_n}-{A\over B}\bigr|
 ={\bigl|(A_n-A)B+A(B-B_n)\bigr|\over B_nB}
 \leq{\bigl\{B+(1-\delta)^{-1}\bigr\}\varepsilon_n\over B_nB} ,
\end{align*}
uniformly in $q\in\widetilde K$, using $A
\leq(1-\delta)^{-1}$ by \eqref{copulalower} and $B_n\to B>0$.  Multiplying by $r_n$ and
using Assumption~\ref{a:secondKM} gives the claim.

In conclusion, Theorem~\ref{thm:normality} gives the first statement and
Remark~\ref{rem:regular} the second.

\newpage

\phantomsection

\end{document}